\documentclass[11pt,letterpaper]{article}

\usepackage[left=1in,right=1in,top=1in,bottom=1in]{geometry}

\usepackage{pgfplots}
\usepackage{tikz}
\usepgfplotslibrary{groupplots}
\pgfplotsset{compat=1.18}
\usepackage{graphicx}
\usepackage{float}
\usepackage{caption}
\usepackage{cite}
\usepackage{amsfonts}
\usepackage{amsmath}
\usepackage{amssymb}
\usepackage{amsthm}
\usepackage{hyperref}
\usepackage{xcolor}
\usepackage{physics}
\usepackage{tikz}
\usepackage{titlesec}
\titleformat*{\section}{\normalsize\bfseries}
\titleformat*{\subsection}{\small\bfseries}
\titleformat*{\subsubsection}{\small\bfseries}
\titleformat*{\paragraph}{\small\bfseries}
\titleformat*{\subparagraph}{\small\bfseries}
\usepackage{kotex}
\usetikzlibrary{arrows.meta, calc, shapes.geometric, positioning, decorations.pathmorphing}
\definecolor{tensorcolor}{HTML}{9A6D99}
\usepackage{circuitikz}
\usepackage{quantikz}
\usepackage[font=small,labelfont=bf]{caption}

\hypersetup{colorlinks=true,urlcolor=[rgb]{0,0,0.5},citecolor=[rgb]{0.5,0,0},linkcolor=[rgb]{0,0,0.5}}

\newcommand{\q}{\mathbf{Q}}

\definecolor{myred}{RGB}{200, 50, 50}      
\definecolor{myblue}{RGB}{50, 80, 180}   
\definecolor{zone_red}{RGB}{250, 235, 235} 
\definecolor{zone_blue}{RGB}{235, 240, 255}
\definecolor{gate_fill}{RGB}{240, 240, 240}

\newtheorem{theorem}{Theorem}

\newtheorem{lemma}{Lemma}
\newtheorem{proposition}{Proposition}

\newtheorem{corollary}{Corollary}
\newtheorem{fact}{Fact}
\theoremstyle{definition}
\newtheorem{definition}{Definition}
\newtheorem{conjecture}{Conjecture}

\def\autorefapp#1{\hyperref[#1]{Appendix~\ref{#1}}}

\def\C{\mathbb{C}}

\def\vev#1{\langle{#1}\rangle}

\def\and{\quad {\rm and} \quad}

\def\ni{\noindent}

\def\CB{{\cal B}}

\renewcommand{\title}[1]{\vbox{\center{\LARGE #1}}\vspace{5mm}}
\renewcommand{\author}[1]{\vbox{\center{#1}}\vspace{5mm}}
\newcommand{\address}[1]{\vbox{\center\em#1}}

\definecolor{db}{rgb}{0,0,0.8}

\DeclareMathOperator*{\range}{{\rm range}}

\usetikzlibrary{calc}
\pgfdeclarelayer{background}
\pgfsetlayers{background,main}

\begin{document}

\title{Finite-bath projected ensembles in dual-unitary circuits}

\author{Rahul Arvind${}^a$ and Nicholas Hunter-Jones${}^{a,b}$}

\address{${}^a$Department of Physics, University of Texas at Austin, Austin, TX}
\address{${}^b$Department of Computer Science, University of Texas at Austin, Austin, TX}

\vspace*{12pt}
\begin{abstract}
We study the emergent randomness in projected ensembles generated by dual-unitary circuits. Specifically, we consider evolution by solvable dual-unitary brickwork circuits followed by measurements on part of the system, which induce an ensemble of states on the unmeasured subsystem. Previous work has established that this projected ensemble of states forms a quantum state design in the limit where the bath size is large, but the rate of convergence at finite bath size has not been rigorously established. Under three inverse-polynomial assumptions, local permutation mixing, a bound on recurrent backward motion, and contraction at the measured boundary, we fill this gap and prove that a bath of polynomial size suffices to form approximate $k$-designs. In establishing our results, we develop a microscopic transport interpretation of emergent randomness in the projected ensemble, in which deviations from local permutation operators appear as defects, and make connections with the theory of quantum walks. 

\end{abstract}

\tableofcontents

\tikzset{
    % Base style for all square tensor nodes
    tensor/.style={
        rectangle, 
        draw=black, 
        thick, 
        rounded corners=2pt, 
        minimum size=18pt, 
        inner sep=0pt, 
        align=center
    },
    U/.style={tensor, fill=red!60!white},
    Ustar/.style={tensor, fill=blue!60!white},
    ident/.style={tensor, fill=orange!80!white},
    leg/.style={thick},
    corner/.style={node, scale=0.6},
    bdot/.style={circle, draw=black, thick, fill=black, inner sep=2pt},
    wdot/.style={circle, draw=black, thick, fill=white, inner sep=2pt}
}

\section{Introduction}
\label{sec:introduction}

The late-time dynamics of generic quantum many-body systems seem to exhibit universal chaotic behavior in certain quantities. Some classic examples are the decay of out-of-time-ordered correlation functions (OTOCs) \cite{SSbutterfly,MSSbound,ChaosChannels,NVH17,vonKey17}, the spread of operator weight \cite{LocalizedShocks,ChaosChannels,NVH17,Parker2018}, and the ramp-and-plateau behavior of the spectral form factor (SFF) \cite{BrezinHikami1,Cotler2016,Chan2017,Bertini2018}. A more recently discovered example of such behavior is the onset of randomness in the \textit{projected ensembles} of these many-body systems \cite{Choi2021,Cotler2021,Ho2021,Claeys2022emergentquantum,Ippoliti2022,Ippoliti2022b}. In this set-up, we measure an extensive subsystem of the late-time many-body state, and study the (projected) ensemble of states on the complementary subsystem. If our extensive system is comprised of $N$ qubits, and we partition it into an output subsystem $N_{A}$ and a measured ``bath" subsystem $N_{B}$, then this ensemble contains $2^{N_{B}}$ states on a Hilbert space of size $(\mathbb{C}^{2})^{\otimes N_{A}}$. For generic dynamics and measurements, we expect this large number of states to be widely distributed over the $N_{A}$ Hilbert space, thereby appearing approximately random.

The concept of an approximately random state ensemble can be made precise in a few ways, but in this paper we will restrict ourselves solely to quantum state $k$-designs \cite{Delsarte1977,Renes2003,Scott2006,Ambainis2007}. Loosely, these are ensembles whose first $k$ moments match those of the (truly random) Haar ensemble. These state $k$-designs are widely used objects in quantum information, acting as an efficient source of randomness in several quantum protocols. For us, these state $k$-designs provide a benchmark which we use to tell us how much randomness is present in the projected ensemble that we have discussed.

Past work on projected ensembles has been largely in the $N_{B} \rightarrow \infty$ limit \cite{Ho2021,Claeys2022emergentquantum,Ippoliti2022,Lucas2022,Chan2024,Chang2024}. This is reminiscent of how we do thermodynamics --- we couple a system to a bath, which has a large number of degrees of freedom, and hence remains largely unchanged. Indeed, the $N_{B} \rightarrow \infty$ limit does have a thermodynamic `sensibility' in that the infinite extent removes finite effects from how the projected ensemble distributes over $N_{A}$, just as coupling an infinite thermal bath to a small system washes away any fluctuations in its temperature. There are a few studies which analyze the scaling of the degree of the $k$-design with the size of $N_{B}$, but they generally limit themselves to studying Haar random, or $k$-design initial states \cite{Cotler2021,Mark2024,Mok2026}, which contain a lot of built-in random structure. 

In this work, we fill in this gap by analyzing the formation of state $k$-designs in the projected ensemble of a dual-unitary brickwork circuit \cite{Bertini2018,Bertini2018b,Bertini2019b,Claeys2020}. Projected ensembles in dual-unitary dynamics have previously been studied in \cite{ho2022exact,Claeys2022emergentquantum,Ippoliti2022,Shrotriya2023,Yu2025}.
Our main contributions are:

\begin{enumerate}
    \item \textbf{A quantitative design bath theorem.} Under three quantitative assumptions, we are able to show that dual-unitary brickwork circuits form quantum state $k$-designs with a bath size that scales as $\mathrm{poly}(k)$.

    \item \textbf{A defect-transport mechanism for convergence to the Haar distribution.} We recast the problem of convergence, using a space-time flip, into the transport and contraction of defects in an absorbing discrete-time quantum process. Two of the three assumptions are introduced to ensure this transport process has no solitonic modes, that is, defects do not get stuck in the process without any contraction.

    \item \textbf{A rigorous finite-replica treatment of normalization.} The study of the discrete-time quantum process bounds the convergence of unnormalized moments of the distribution; to pass to the physical, normalized moment, past work has invoked a replica-trick consideration. Here, we do this on rigorous footing using the approximation theory of Jacobi polynomials. 
\end{enumerate}

The rest of this paper is organized as follows. In \autoref{sec:preliminaries}, we introduce some preliminary concepts that are required to follow the technical details in the paper, as well as introduce notation that is required for later sections. In \autoref{sec:proofsketch}, we then give a brief proof sketch of the main results of the paper, which are then written out in \autoref{sec:defect-structure} through \autoref{sec:normalization-to-design}. In \autoref{sec:comment} we discuss the genericity of one of the assumptions required for our proof. We then finish in \autoref{sec:outlook} with an outlook detailing open questions, as well as what we think are promising directions for future research.
 
\section{Preliminaries}
\label{sec:preliminaries}
We introduce the central object of study, the dual-unitary brickwork circuit \cite{Bertini2018,Bertini2018b,Bertini2019b,Claeys2020}, which here refers to multiple iterations of the extensive single-layer unitary operator
\begin{equation}
\label{eq:brickwork-layer}
    \mathbb{U} = \mathbb{U}_{e}\mathbb{U}_{o}\,, \quad{\rm with}\quad \mathbb{U}_{e} = \bigotimes_{x= 0}^{L-1} u_{x,x+\frac{1}{2}}\,,\quad  \mathbb{U}_{o} = \bigotimes_{x= 1}^{L-1} u_{x-\frac{1}{2},x}\,,
\end{equation}
and where we label single-qubit sites by half-integers as $\{0, \frac{1}{2}, 1, \dots L-\frac{1}{2}\}$ for a total of $2L$ qubits, and $u_{x,y}$ refers to a constituent two-qubit brickwork unitary that specifies the circuit. Dual unitarity here refers to the individual brickwork unitary being dual unitary, meaning that these operators retain unitarity when the space and time indices are swapped with each other. The $t$-layer brickwork circuit is just given then by $\mathbb{U}^{t} = (\mathbb{U}_{e}\mathbb{U}_{o})^{t}$. In this section, we will give a brief introduction to the theory and description of these dual-unitary operators and circuits, and connect them to our study of the projected ensemble.

\subsection{Dual-unitarity and graphical notation.}
\label{sec:dual-unitarity}

We begin with the basic ingredient in the brickwork circuit: the two-qubit constituent unitary $u_{cd,ab}$, where $(a,b)$ represent the input legs, each with local dimension $2$, and $(c,d)$ represent the output legs. Here, $a$ spans the basis vectors of one qubit and $b$ of another, so altogether the Cartesian product $(a,b)$ specifies a basis in the four-qubit Hilbert space. Now, if we think of a brickwork circuit propagating the state of the system forward in time, we can think of the indices $(c,d)$ ``belonging" to a time later than $(a,b)$. Now assume (without loss of generality) that we have labeled $(a,b,c,d)$ such that $a$ is to the left of $b$ and $c$ is to the left of $d$. Then we can think of an alternate picture where we input states from the right into the four-qubit Hilbert space described by the input legs $(b,d)$, and define a \textit{space evolution} operator $\tilde{u}$ such that $\tilde{u}_{ac,bd} = u_{ac,bd}$ --- which maps into output legs $(a,c)$. If the map $\tilde{u}$ is unitary, then we say that it is \textit{dual unitary}. That is, we require

\begin{equation}
\label{eq:unitarity}
    uu^{\dag} = u^{\dag}u = I,
\end{equation}

\begin{equation}
\label{eq:dual-unitarity}
    \tilde{u}^{\dag}\tilde{u} = \tilde{u}\tilde{u}^{\dag}= I.
\end{equation}

It is convenient to use graphical notation for $u_{i,j}$ and its complex conjugate $u^{*}_{i,j}$, which we adapt from \cite{Claeys2022emergentquantum}.

\begin{equation}
    u_{i,j}^{k,l} \quad = \quad
    \begin{tikzpicture}[baseline=(current bounding box.center), scale=0.8]
        % Main Tensor
        \node[U] (T) at (0,0) {$\urcorner$};
        
        % Legs with labels
        \draw[leg] (T.south west) -- ++(-0.4,-0.4) node[below] {$i$};
        \draw[leg] (T.south east) -- ++(0.4,-0.4) node[below] {$j$};
        \draw[leg] (T.north west) -- ++(-0.4,0.4) node[above] {$k$};
        \draw[leg] (T.north east) -- ++(0.4,0.4) node[above] {$l$};
    \end{tikzpicture}
    \qquad \qquad \qquad
    (u^\dagger)_{i,j}^{k,l} \quad = \quad
    \begin{tikzpicture}[baseline=(current bounding box.center), scale=0.8]
        % Main Tensor
        \node[Ustar] (T) at (0,0) {$\lrcorner$};
        
        % Legs with labels
        \draw[leg] (T.south west) -- ++(-0.4,-0.4) node[below] {$i$};
        \draw[leg] (T.south east) -- ++(0.4,-0.4) node[below] {$j$};
        \draw[leg] (T.north west) -- ++(-0.4,0.4) node[above] {$k$};
        \draw[leg] (T.north east) -- ++(0.4,0.4) node[above] {$l$};
    \end{tikzpicture}
\end{equation}
In this graphical notation, we can describe unitarity and dual unitarity as
\begin{equation}
\begin{tikzpicture}[baseline=(current bounding box.center), scale=0.8]
    % Term 1
    \node[U] (U1) at (0, 0.75) {$\urcorner$};
    \node[Ustar] (U2) at (0, -0.75) {$\lrcorner$};
    
    % Outer legs
    \draw[leg] (U1.north west) -- ++(-0.4, 0.4) node[above] {$k$};
    \draw[leg] (U1.north east) -- ++(0.4, 0.4) node[above] {$l$};
    \draw[leg] (U2.south west) -- ++(-0.4, -0.4) node[below] {$i$};
    \draw[leg] (U2.south east) -- ++(0.4, -0.4) node[below] {$j$};
    
    % Smooth loop connections (tightened looseness for condensed spacing)
    \draw[leg] (U1.south west) to[out=-135, in=135, looseness=1.4] (U2.north west);
    \draw[leg] (U1.south east) to[out=-45, in=45, looseness=1.4] (U2.north east);
    
    \node at (1.4, 0) {$=$};
    
    % Term 2
    \begin{scope}[shift={(2.8, 0)}]
        \node[Ustar] (U3) at (0, 0.75) {$\lrcorner$};
        \node[U] (U4) at (0, -0.75) {$\urcorner$};
        
        \draw[leg] (U3.north west) -- ++(-0.4, 0.4) node[above] {$k$};
        \draw[leg] (U3.north east) -- ++(0.4, 0.4) node[above] {$l$};
        \draw[leg] (U4.south west) -- ++(-0.4, -0.4) node[below] {$i$};
        \draw[leg] (U4.south east) -- ++(0.4, -0.4) node[below] {$j$};
        
        \draw[leg] (U3.south west) to[out=-135, in=135, looseness=1.4] (U4.north west);
        \draw[leg] (U3.south east) to[out=-45, in=45, looseness=1.4] (U4.north east);
    \end{scope}
    
    \node at (4.2, 0) {$=$};
    
    % Term 3 (Result)
    \begin{scope}[shift={(5.2, 0)}]
        \draw[leg] (-0.4, 1.15) node[above] {$k$} -- (-0.4, -1.15) node[below] {$i$};
        \draw[leg] (0.4, 1.15) node[above] {$l$} -- (0.4, -1.15) node[below] {$j$};
        \node at (0.8, -0.2) {$,$};
    \end{scope}
\end{tikzpicture}
\end{equation}

% Bottom Equation: Dual/Spatial Unitarity
\begin{equation}
\begin{tikzpicture}[baseline=(current bounding box.center), scale=0.8]
    % Term 1
    \node[U] (U1) at (-0.8, 0) {$\urcorner$};
    \node[Ustar] (U2) at (0.8, 0) {$\ulcorner$};
    
    % Outer legs
    \draw[leg] (U1.north west) -- ++(-0.4, 0.4) node[above] {$k$};
    \draw[leg] (U1.south west) -- ++(-0.4, -0.4) node[below] {$i$};
    \draw[leg] (U2.north east) -- ++(0.4, 0.4) node[above] {$l$};
    \draw[leg] (U2.south east) -- ++(0.4, -0.4) node[below] {$j$};
    
    % Smooth loop connections (tightened looseness for condensed spacing)
    \draw[leg] (U1.north east) to[out=45, in=135, looseness=1.35] (U2.north west);
    \draw[leg] (U1.south east) to[out=-45, in=-135, looseness=1.35] (U2.south west);
    
    \node at (1.9, 0) {$=$};
    
    % Term 2
    \begin{scope}[shift={(3.8, 0)}]
        \node[Ustar] (U3) at (-0.8, 0) {$\ulcorner$};
        \node[U] (U4) at (0.8, 0) {$\urcorner$};
        
        \draw[leg] (U3.north west) -- ++(-0.4, 0.4) node[above] {$k$};
        \draw[leg] (U3.south west) -- ++(-0.4, -0.4) node[below] {$i$};
        \draw[leg] (U4.north east) -- ++(0.4, 0.4) node[above] {$l$};
        \draw[leg] (U4.south east) -- ++(0.4, -0.4) node[below] {$j$};
        
        \draw[leg] (U3.north east) to[out=45, in=135, looseness=1.35] (U4.north west);
        \draw[leg] (U3.south east) to[out=-45, in=-135, looseness=1.35] (U4.south west);
    \end{scope}
    
    \node at (5.7, 0) {$=$};
    
    % Term 3 (Result)
    \begin{scope}[shift={(7.6, 0)}]
        \draw[leg] (-1, 0.4) node[left] {$k$} -- (0.6, 0.4) node[right] {$l$};
        \draw[leg] (-1, -0.4) node[left] {$i$} -- (0.6, -0.4) node[right] {$j$};
        \node at (1.3, -0.2) {$.$};
    \end{scope}
\end{tikzpicture}
\end{equation}

The time evolution operator from Equation~\eqref{eq:brickwork-layer} can be similarly written diagrammatically. The diagrammatic approach makes clear where the nomenclature brickwork circuit appears from. We will use diagrams interchangeably in our work, depending on convenience. It is important to emphasize that the diagrammatic notation is entirely rigorous, and any graphical identities that we introduce are precise statements about the operators they represent. 

\begin{equation}
\mathbb{U} = \mathbb{U}_e \mathbb{U}_o \quad = \quad
\begin{tikzpicture}[
    baseline=(current bounding box.center), 
    scale=0.75, 
    tensor/.style={rectangle, draw=black, thick, rounded corners=2pt, minimum size=18pt, inner sep=0pt, align=center},
    U/.style={tensor, fill=red!60!white},
    leg/.style={thick}
]

    % Bottom Layer (Even layer: U_e)
    \node[U] (U00) at (0,0) {$\urcorner$};
    \node[U] (U20) at (2,0) {$\urcorner$};
    \node[U] (U60) at (6,0) {$\urcorner$};

    % Top Layer (Odd layer: U_o)
    \node[U] (U11) at (1,1) {$\urcorner$};
    \node[U] (U51) at (5,1) {$\urcorner$};

    % Shortened from 0.5 to 0.35
    \draw[leg] (U00.south west) -- ++(-0.35,-0.35) node[below=2pt] {$0$};
    \draw[leg] (U00.south east) -- ++(0.35,-0.35) node[below=2pt] {$\frac{1}{2}$};
    \draw[leg] (U20.south west) -- ++(-0.35,-0.35) node[below=2pt] {$1$};
    \draw[leg] (U20.south east) -- ++(0.35,-0.35) node[below=2pt] {$\frac{3}{2}$};
    
    \draw[leg] (U60.south west) -- ++(-0.35,-0.35) node[below=2pt] {$L-1$};
    \draw[leg] (U60.south east) -- ++(0.35,-0.35) node[below=2pt] {$L-\frac{1}{2}$};

    % Even Layer -> Odd Layer
    \draw[leg] (U00.north east) -- (U11.south west);
    \draw[leg] (U20.north west) -- (U11.south east);
    \draw[leg] (U60.north west) -- (U51.south east);

    % Leftmost wire (Site 0)
    \draw[leg] (U00.north west) -- ++(-0.4, 0.4);

    % Rightmost wire (Site L-1/2)
    \draw[leg] (U60.north east) -- ++(0.4, 0.4);

    \draw[leg] (U11.north west) -- ++(-0.35,0.35);
    \draw[leg] (U11.north east) -- ++(0.35,0.35);
    \draw[leg] (U51.north west) -- ++(-0.35,0.35);
    \draw[leg] (U51.north east) -- ++(0.35,0.35);

    % Wires extending into the dotted gap from the left
    \draw[leg] (U20.north east) -- ++(0.35,0.35);
    
    % Wires extending into the dotted gap from the right
    \draw[leg] (U51.south west) -- ++(-0.35,-0.35);

    % The vertical column of horizontal dots
    \node at (3.5, 0)   {$\cdots$};
    \node at (3.5, 0.5) {$\cdots$};
    \node at (3.5, 1)   {$\cdots$};
    \node at (3.5, 1.5) {$\cdots$};
\end{tikzpicture}
\end{equation}

\subsection{Folded dual-unitary circuits}
\label{sec:folded-circuits}

We now introduce the notion of folded brickwork unitary circuits. For some motivation, consider the usual Heisenberg picture dynamics, where operators evolve as $O \rightarrow uOu^{\dag}$. This admits a \textit{vectorized} picture, where we represent operators as higher dimensional states. Specifically, the vectorization isomorphism is $O \rightarrow \mathrm{vec}(O) := \ket{O} := \sum_{i} (I \otimes O^{T})\ket{i} \otimes \ket{i}$. This uniquely extends to a map from channels to operators. In the vectorized picture, the channel $(\cdot) \rightarrow u(\cdot)u^{\dag}$ is mapped to the operator $u \otimes u^{*}$ that acts on vectors of the form $\ket{O}$. Consequently, our overall brickwork layer acts as a tensor product of two ``usual" brickwork circuits, one of which has $u^{*}$ as its constituent unitary instead of $u$.

To represent these tensor products compactly diagrammatically, we overlay one of these circuits over the other, with the understanding that if there are gates at the same location, we are to take their tensor product. A more compact notation can be introduced for such pairs of $u$ and $u^*$.

\begin{equation}
    u \otimes u^* \ = \ 
    \begin{tikzpicture}[
        baseline=(current bounding box.center), 
        scale=0.8,
        % Define styles locally
        tensor/.style={rectangle, draw=black, thick, rounded corners=2pt, minimum size=18pt, inner sep=0pt, align=center},
        U/.style={tensor, fill=red!60!white},
        Ustar/.style={tensor, fill=blue!60!white},
        ident/.style={tensor, fill=orange!80!white},
        leg/.style={thick}
    ]

        % 1. Back Tensor (U*, Blue)
        % Drawn first so it naturally sits behind the red tensor
        \begin{scope}[shift={(0.2, 0.2)}]
            % Draw legs from the center outwards
            \draw[leg] (0,0) -- ++(-0.6, -0.6);
            \draw[leg] (0,0) -- ++(0.6, -0.6);
            \draw[leg] (0,0) -- ++(-0.6, 0.6);
            \draw[leg] (0,0) -- ++(0.6, 0.6);
            % Draw the blue node on top of its own legs
            \node[Ustar] at (0,0) {\hspace{2pt}\vspace{2pt}$\urcorner$};
        \end{scope}

        % 2. Front Tensor (U, Red)
        % Drawn second so it covers the blue tensor and its inner legs
        \begin{scope}[shift={(0, 0)}]
            % Draw legs from the center outwards
            \draw[leg] (0,0) -- ++(-0.6, -0.6);
            \draw[leg] (0,0) -- ++(0.6, -0.6);
            \draw[leg] (0,0) -- ++(-0.6, 0.6);
            \draw[leg] (0,0) -- ++(0.6, 0.6);
            % Draw the red node on top of its own legs
            \node[U] at (0,0) {$\urcorner$};
        \end{scope}

    \end{tikzpicture}
    \quad = \quad
    \begin{tikzpicture}[
        baseline=(current bounding box.center), 
        scale=0.8,
        tensor/.style={rectangle, draw=black, thick, rounded corners=2pt, minimum size=18pt, inner sep=0pt, align=center},
        ident/.style={tensor, fill=orange!80!white},
        leg/.style={thick}
    ]
        % 3. Identity Tensor (Orange)
        % Draw legs
        \draw[leg] (0,0) -- ++(-0.6, -0.6);
        \draw[leg] (0,0) -- ++(0.6, -0.6);
        \draw[leg] (0,0) -- ++(-0.6, 0.6);
        \draw[leg] (0,0) -- ++(0.6, 0.6);
        
        % Draw the orange node
        \node[ident] at (0,0) {$1$};
    \end{tikzpicture}
    \ .
\end{equation}
The number inside the orange box refers to the number of pairs of $u$ and $u^{*}$ at a particular location. We call the dual-unitary circuit consisting of $k$ pairs of $u$ and $u*$ at each location the \textit{$k$-fold dual-unitary brickwork circuit}. For complete clarity, we illustrate in the next equation one layer of this object.

\begin{equation}
\mathbb{U}^{(k)} := \mathbb{U}^{(k)}_e \mathbb{U}^{(k)}_o =
\begin{tikzpicture}[
    baseline=(current bounding box.center), 
    scale=0.75, 
    tensor/.style={rectangle, draw=black, thick, rounded corners=2pt, minimum size=18pt, inner sep=0pt, align=center},
    Ufold/.style={tensor, fill=orange!80!white},
    leg/.style={thick}
]
    % Bottom Layer (Even layer)
    \node[Ufold] (U00) at (0,0) {$k$};
    \node[Ufold] (U20) at (2,0) {$k$};
    \node[Ufold] (U60) at (6,0) {$k$};

    % Top Layer (Odd layer)
    \node[Ufold] (U11) at (1,1) {$k$};
    \node[Ufold] (U51) at (5,1) {$k$};

    \draw[leg] (U00.south west) -- ++(-0.35,-0.35) node[below=2pt] {$0$};
    \draw[leg] (U00.south east) -- ++(0.35,-0.35) node[below=2pt] {$\frac{1}{2}$};
    \draw[leg] (U20.south west) -- ++(-0.35,-0.35) node[below=2pt] {$1$};
    \draw[leg] (U20.south east) -- ++(0.35,-0.35) node[below=2pt] {$\frac{3}{2}$};
    
    \draw[leg] (U60.south west) -- ++(-0.35,-0.35) node[below=2pt] {$L-1$};
    \draw[leg] (U60.south east) -- ++(0.35,-0.35) node[below=2pt] {$L-\frac{1}{2}$};

    % Even Layer -> Odd Layer
    \draw[leg] (U00.north east) -- (U11.south west);
    \draw[leg] (U20.north west) -- (U11.south east);
    \draw[leg] (U60.north west) -- (U51.south east);

    % Leftmost wire (Site 0)
    \draw[leg] (U00.north west) -- ++(-0.4, 0.4);

    % Rightmost wire (Site L-1/2)
    \draw[leg] (U60.north east) -- ++(0.4, 0.4);

    \draw[leg] (U11.north west) -- ++(-0.35,0.35);
    \draw[leg] (U11.north east) -- ++(0.35,0.35);
    \draw[leg] (U51.north west) -- ++(-0.35,0.35);
    \draw[leg] (U51.north east) -- ++(0.35,0.35);

    % Wires extending into the dotted gap from the left
    \draw[leg] (U20.north east) -- ++(0.35,0.35);
    
    % Wires extending into the dotted gap from the right
    \draw[leg] (U51.south west) -- ++(-0.35,-0.35);

    % The vertical column of horizontal dots
    \node at (3.5, 0)   {$\cdots$};
    \node at (3.5, 0.5) {$\cdots$};
    \node at (3.5, 1)   {$\cdots$};
    \node at (3.5, 1.5) {$\cdots$};
\end{tikzpicture}
\end{equation}

We emphasize that most of our work in this paper will be in this vectorized setting. When we talk of a $k$-folded circuit acting on an operator, we will mean its action on the vectorization of that operator.

\subsection{Quantum state $k$-designs}
\label{sec:state-designs}
We now introduce some of the quantum information theoretic tools we need, specifically the Haar measure, moments of Haar random states, and state $k$-designs. 

Since the unitary group is a compact topological group, there exists a unique normalized measure on it which is translation invariant under the group action. We call this the Haar measure $H$, and it provides us with a notion of true randomness which we can compare the distribution of our ensembles to. This Haar measure on unitaries also extends to states in the obvious way. Fix a pure state $\rho$, and evolve it by a Haar random unitary. The resulting distribution of states is said to be Haar random as well. We denote the Haar measure on the space of states as $\mu$.

A quantum state $k$-design is an approximation to the Haar ensemble which matches its first $k$ moments ($k+1$ if we include the $0$ moment, $\mathbb{I}$). We will consider in this work approximate state $k$-designs, which match these moments up to some error.

\begin{definition}[State $k$-design]
Given an ensemble of states on $\mathbb{C}^{d}$, denoted as $\mathcal{E} = \{(p_i,\ket{\psi_i})\}$, we call it an $\varepsilon$-approximate state $k$-design if
\begin{equation}
    \left\Vert \sum_{i}p_{i}\ketbra{\psi_{i}}^{\otimes k} -  \int \ketbra{\psi}^{\otimes k} \dd \mu(\psi) \right\Vert_{1} \leq \varepsilon.
\end{equation}
\end{definition} 
Here, the norm $\Vert \cdot \Vert_{1}$ refers to the trace norm. It is useful to write the Haar moment in terms of
permutation matrices. For $\pi\in S_k$, define
\begin{equation}
    P_\pi^{(d)}\ket{i_1,\ldots,i_k} :=\ket{i_{\pi^{-1}(1)},\ldots,i_{\pi^{-1}(k)}}.
    \label{eq:replica-permutation}
\end{equation}

It turns out that the moments of the Haar ensemble can be written as the (normalized) projector on the symmetric subspace of $(\mathbb{C}^{d})^{\otimes k}$ as \cite{harrow2013church} 
\begin{equation}
    \int \ketbra{\psi}^{\otimes k} \dd \mu(\psi) = \frac{1}{k!\binom{d+k-1}{k}}\sum_{\pi \in S_{k}} P_{\pi}.
\end{equation}
As a result, $k$-design formation in an ensemble of quantum states is equivalent to the corresponding moment operator converging to this normalized orthogonal projector. 

The $k$-th moment of Haar random states is equal to the $k$-fold twirl of a fixed state $\ket\varphi$ by a Haar random unitary $U$
\begin{equation}
    \int \ketbra{\psi}^{\otimes k} \dd \mu(\psi) = \int U^{\otimes k}\ketbra{\varphi}^{\otimes k}U^{\dag}{}^{\otimes k} \dd \mu(U),
\end{equation}
where on the left-hand side the integration is over the Haar measure on states, while on the right-hand side the integration is over unitaries, which, for simplicity, we also denote as $\mu$. Define the $k$-fold Haar twirl on $(\mathbb{C}^{d})^{\otimes k}$ as 
\begin{equation}
    \Pi^{(k)}(X) := \int U^{\otimes k}X(U^\dagger)^{\otimes k}\dd \mu(U),
\end{equation}
It can be shown \cite{brandao2012local,harrow2013church} that this $k$-fold channel, represented on vectorized operators, projects onto the span of permutation operators $\operatorname{span}\{P_\pi^{(d)}:\pi\in S_k\}$.

\subsection{Projected ensembles and the finite bath setting}
\label{sec:finite-bath-setting}

Consider a state $\ket{\psi}$ on $N$ qubits, and evolve it via a unitary brickwork quantum circuit, say $\mathbb{U}$. Split this system into two disjoint parts $A$ and $B$, and at some later time $t_{\mathrm{phys}}$, measure the subsystem $B$ in some basis. We use time here loosely --- this is a discrete-time quantum process, so we really mean $t_{\mathrm{phys}}$ applications of $\mathbb{U}$. Write $N=N_A+N_B$ and $d=2^{N_A}$. The process yields the projected ensemble of states on the subsystem $A$ indexed by the measurement outcomes on the subsystem $B$
\begin{equation}
    \mathcal{E}_{A} = \Big\{\ket{\psi^{\mathcal{E}}_{i_{1}, i_{2}, \dots i_{N_{B}}}}, ~~p^{\mathcal{E}}_{i_{1}, i_{2}, \dots i_{N_{B}}} \Big\},
\end{equation}
where $\{{i_{1}, i_{2}, \dots i_{N_{B}}}\}$ is really a binary string of length $N_{B}$, the $j$ entry of which tells us which of the two basis states were measured at qubit $j$. $\ket*{\psi_{i_{1}, i_{2}, \dots i_{N_{B}}}}$ then refers to which of the $2^{N_{B}}$ states in the ensemble is realized with measurement outcome $\{{i_{1}, i_{2}, \dots i_{N_{B}}} \}$ on $B$. Writing $i=(i_1,\ldots,i_{N_B})$ for an outcome string, the unnormalized state and its Born probability are
\begin{equation}
    \ket{\tilde\psi_i^{\mathcal E}}
    :=(\mathbb I_A\otimes\bra{i})\mathbb U^{t_{\mathrm{phys}}}\ket{\psi},
\end{equation}
\begin{equation}
    p_i^{\mathcal E}=\|\tilde\psi_i^{\mathcal E}\|^2,
\end{equation}
\begin{equation}
    \ket{\psi_i^{\mathcal E}}
    =\frac{1}{\sqrt{p_i^{\mathcal E}}} \ket{\tilde\psi_i^{\mathcal E}}.
    \label{eq:physical-outcome-states}
\end{equation}
Here, we are assuming the circuit runs for some $t_{\mathrm{phys}}$ timesteps. The last expression is used only when $p_i^{\mathcal E}>0$; zero-probability outcomes contribute zero to the physical moments. The moments of this ensemble can be calculated as 
\begin{equation}
    \rho^{(k)}_{\mathcal{E}_{A}} = \sum_{i=1}^{2^{N_{B}}} (\ketbra{\psi^{\mathcal{E}}_{i}})^{\otimes k} p^{\mathcal{E}}_{i},
    \label{eq:physical-projected-moment}
\end{equation}
where $i$ indexes the various measurement outcomes in the ensemble. In particular, the first moment is just the reduced density matrix corresponding to the subsystem $A$. In light of the previous subsection, the projected ensemble $\mathcal{E}$ forms an approximate state $k$-design if $\rho^{(k)}_{\mathcal{E}_{A}}$ is close to the Haar $k$-moment operator.\\

\subsubsection*{The set-up and results of \cite{Claeys2022emergentquantum}}

In this reference, it is shown that for a large class of dual-unitary projected ensembles, convergence to the Haar moments can be obtained from an analysis of rotated transfer matrices. We use the product-state and single-site measurement geometry of Section 2.2 of that reference, with $A$ adjacent to a closed endpoint and $t_{\mathrm{phys}}>N_A/2$. We describe their argument in some detail here. They begin by considering a modified $k$-th moment operator 
\begin{equation}
    {\rho}^{(n',k)}_{\mathcal{E}_{A}} = \sum_{i=1}^{2^{N_{B}}} \big(\ketbra{\psi^{\mathcal{E}}_{i}}\big)^{\otimes k} (p^{\mathcal{E}}_{i})^{n'},
\end{equation}
where $\rho^{(k)}_{\mathcal{E}_{A}}$ can be obtained by setting $n'=1$. It is useful to rewrite this in terms of unnormalized states $\ket*{\tilde{\psi}^{\mathcal{E}}_{i}}$ as 
\begin{equation}
    \tilde{\rho}^{(n,k)}_{\mathcal{E}_{A}} = \sum_{i=1}^{2^{N_{B}}} \Big(\ketbra{\tilde{\psi}^{\mathcal{E}}_{i}}\Big)^{\otimes k} (p^{\mathcal{E}}_{i})^{n},
\end{equation}
where $n=n'-k$, so that $\tilde\rho^{(n,k)}_{\mathcal E_A}=\rho^{(n+k,k)}_{\mathcal E_A}$, and now the physical moment operator $\rho^{(k)}_{\mathcal{E}_{A}}$ is obtained by setting $n=1-k$. For nonnegative integer $n$, put $m=n+k$. Ref.~\cite{Claeys2022emergentquantum} then represents $\tilde\rho^{(n,k)}_{\mathcal{E}_{A}}$ by the following diagram. Each gate carries $m$ folded replicas; at $N_{A}$, $n$ replicas are traced out and the remaining $k$ are left open. 

\begin{equation}
\tilde\rho_{\mathcal E_A}^{(n,k)}
=
\sum_{z_B}
\;
\begin{tikzpicture}[baseline=12mm, 
tensor/.style={rectangle, draw=black, thick, rounded corners=2.5pt, fill=orange!80!white, minimum width=7.5mm, minimum height=7.5mm, inner sep=0pt, font=\small\bfseries},
leg/.style={thick, draw=black, line cap=round, line join=round},
arrowleg/.style={thick, draw=black, -{Latex[length=2.4mm,width=1.9mm]}},
dot/.style={circle, fill=black, inner sep=1.45pt},
cup/.style={thick,draw=black}
]

% Top row
\node[tensor] (T1) at (0.0,3.20) {$m$};
\node[tensor] (T2) at (2.0,3.20) {$m$};
\node[tensor] (T3) at (4.0,3.20) {$m$};
\node[tensor] (T4) at (6.0,3.20) {$m$};

% Second row
\node[tensor] (M1) at (1.0,2.20) {$m$};
\node[tensor] (M2) at (3.0,2.20) {$m$};
\node[tensor] (M3) at (5.0,2.20) {$m$};
\node[tensor] (M4) at (7.0,2.20) {$m$};

% Third row
\node[tensor] (B1) at (0.0,1.20) {$m$};
\node[tensor] (B2) at (2.0,1.20) {$m$};
\node[tensor] (B3) at (4.0,1.20) {$m$};
\node[tensor] (B4) at (6.0,1.20) {$m$};

% Fourth row
\node[tensor] (C1) at (1.0,0.20) {$m$};
\node[tensor] (C2) at (3.0,0.20) {$m$};
\node[tensor] (C3) at (5.0,0.20) {$m$};
\node[tensor] (C4) at (7.0,0.20) {$m$};

% Fifth row
\node[tensor] (D1) at (0.0,-0.80) {$m$};
\node[tensor] (D2) at (2.0,-0.80) {$m$};
\node[tensor] (D3) at (4.0,-0.80) {$m$};
\node[tensor] (D4) at (6.0,-0.80) {$m$};

% Top -> second
\draw[leg] (T1.south east) -- (M1.north west);
\draw[leg] (T2.south west) -- (M1.north east);

\draw[leg] (T2.south east) -- (M2.north west);
\draw[leg] (T3.south west) -- (M2.north east);

\draw[leg] (T3.south east) -- (M3.north west);
\draw[leg] (T4.south west) -- (M3.north east);

\draw[leg] (T4.south east) -- (M4.north west);

% Second -> third
\draw[leg] (M1.south west) -- (B1.north east);
\draw[leg] (M1.south east) -- (B2.north west);

\draw[leg] (M2.south west) -- (B2.north east);
\draw[leg] (M2.south east) -- (B3.north west);

\draw[leg] (M3.south west) -- (B3.north east);
\draw[leg] (M3.south east) -- (B4.north west);

\draw[leg] (M4.south west) -- (B4.north east);

% Third -> fourth
\draw[leg] (B1.south east) -- (C1.north west);
\draw[leg] (B2.south west) -- (C1.north east);

\draw[leg] (B2.south east) -- (C2.north west);
\draw[leg] (B3.south west) -- (C2.north east);

\draw[leg] (B3.south east) -- (C3.north west);
\draw[leg] (B4.south west) -- (C3.north east);

\draw[leg] (B4.south east) -- (C4.north west);

% Fourth -> fifth
\draw[leg] (C1.south west) -- (D1.north east);
\draw[leg] (C1.south east) -- (D2.north west);

\draw[leg] (C2.south west) -- (D2.north east);
\draw[leg] (C2.south east) -- (D3.north west);

\draw[leg] (C3.south west) -- (D3.north east);
\draw[leg] (C3.south east) -- (D4.north west);

\draw[leg] (C4.south west) -- (D4.north east);

% Caps on T1 and T2
\coordinate (c11) at ($(T1.north west)+(-0.26,0.26)$);
\coordinate (c12) at ($(T1.north east)+( 0.26,0.26)$);
\coordinate (c21) at ($(T2.north west)+(-0.26,0.26)$);
\coordinate (c22) at ($(T2.north east)+( 0.26,0.26)$);

\draw[leg] (T1.north west) -- (c11);
\draw[leg] (T1.north east) -- (c12);
\draw[leg] (T2.north west) -- (c21);
\draw[leg] (T2.north east) -- (c22);

\foreach \C in {c11,c12,c21,c22}{
    \node[draw,rectangle,fill=white,inner sep=1.5pt,font=\scriptsize]
        (Q\C) at (\C) {$\mathrm{tr}_n$};
    \draw[leg] (Q\C.north) -- ++(0,0.38);
}

% z1, z2, z3 on the next top legs
\coordinate (z1b) at ($(T3.north west)+(-0.26,0.26)$);
\coordinate (z2b) at ($(T3.north east)+( 0.26,0.26)$);
\coordinate (z3b) at ($(T4.north west)+(-0.26,0.26)$);
\coordinate (z4b) at ($(T4.north east)+( 0.26,0.26)$);

\draw[leg] (T3.north west) -- (z1b);
\draw[leg] (T3.north east) -- (z2b);
\draw[leg] (T4.north west) -- (z3b);
\draw[leg] (T4.north east) -- (z4b);

\draw[arrowleg] (z1b) -- ++(0,0.92);
\node at ($(z1b)+(0,1.22)$) {$z_1$};

\draw[arrowleg] (z2b) -- ++(0,0.92);
\node at ($(z2b)+(0,1.22)$) {$z_2$};

\draw[arrowleg] (z3b) -- ++(0,0.92);
\node at ($(z3b)+(0,1.22)$) {$z_3$};

\node at ($(z4b)+(-0.25,1.05)$) {$\cdots$};

% Upper-right leg goes to z_{N_B}
\coordinate (rtop) at ($(M4.north east)+(0.26,0.26)$);
\draw[leg] (M4.north east) -- (rtop);
\draw[arrowleg] (rtop) -- ++(0,2.05);
\node at ($(rtop)+(0,2.35)$) {$z_{N_B}$};

% Lower-right leg goes out, straight down, then diagonally back into C4
\coordinate (rupper) at ($(M4.south east)+(0.26,-0.26)$);
\coordinate (rlower) at ($(C4.north east)+(0.26,0.26)$);

\draw[leg] (M4.south east) -- (rupper) -- (rlower) -- (C4.north east);

% Left legs
\coordinate (d0) at ($(B1.south west)+(-0.38,-0.38)$);
\draw[leg] (B1.south west) -- (d0);
\node[dot] at (d0) {};

\coordinate (d1) at ($(T1.south west)+(-0.38,-0.38)$);
\draw[leg] (T1.south west) -- (d1);
\node[dot] at (d1) {};

% Bottom row dots from D row
\foreach \X in {D1,D2,D3,D4}{
    \coordinate (dl) at ($(\X.south west)+(-0.38,-0.38)$);
    \coordinate (dr) at ($(\X.south east)+( 0.38,-0.38)$);

    \draw[leg] (\X.south west) -- (dl);
    \draw[leg] (\X.south east) -- (dr);

    \node[dot] at (dl) {};
    \node[dot] at (dr) {};
}.

\end{tikzpicture}
\end{equation}

This can be split into three vertical parts, of which one repeats $\ell=(N_B-1)/2$ times. We use this paired-column convention, with $N_B=2\ell+1$, and a circuit, initial state, and measurement pattern that repeat across the two-site bath cells. 

It is convenient to view this in the rotated picture, with space and time flipped --- this follows from the dual-unitary of the underlying circuit. In this `auxiliary` picture, one applies the circuit right to left; the right boundary is shown to correspond to an input state, the repeated segment (which we call the auxiliary column, and is indexed by the measurement) is applied $\ell$ times, and then we have a left boundary which traces out $n$ replicas and maps the previous output from  $\mathbb{C}^{D}$ to $\mathbb{C}^{d}$ (with $D=2^{2t_{\mathrm{phys}}-1}$, $d=2^{N_A}$). Define $t:=2t_{\mathrm{phys}}-1$ as the number of free auxiliary sites, indexed $1,\ldots,t$. In general, $d,D,t$ remain fixed when the bath size grows at fixed physical time and subsystem size.

In Ref.~\cite{Claeys2022emergentquantum} they then restrict themselves to \textit{solvable} dual unitary circuits (we shall do this as well). Let us make this precise. Take $C_\xi:\mathbb C^D\to\mathbb C^D$ an unnormalized auxiliary column for the two-site outcome $\xi\in\{0,1\}^2$, $F_A:\mathbb C^D\to\mathbb C^d$ the end map onto $A$, and $\ket{r_a}\in\mathbb C^D$ the remaining endpoint vector, with $a\in\{0,1\}$. The unfolded network gives
\begin{equation}
    \ket{\tilde\psi_i^{\mathcal E}}
    =F_A C_{\xi_\ell}\cdots C_{\xi_1}\ket{r_a},\ i=(a,\xi_1,\ldots,\xi_\ell).
    \label{eq:raw-auxiliary-boundary-decomposition}
\end{equation}

\begin{definition}[Solvability in the single-site measurement geometry]
\label{def:boundary-solvability}
    We call the circuit, initial product state, and measurement basis solvable in this geometry when their auxiliary tensors satisfy
    \begin{equation}
    C_\xi^\dagger C_\xi=C_\xi C_\xi^\dagger
        =\frac14\mathbb I_D,\ \xi\in\{0,1\}^2,
        \label{eq:solvable-columns}
\end{equation}
    \begin{equation}
    F_AF_A^\dagger
        =2^{t_{\mathrm{phys}}-N_A}\mathbb I_d,
\end{equation}
\begin{equation}
    \langle r_a|r_a\rangle=2^{t_{\mathrm{phys}}-2},\ a\in\{0,1\},
        \label{eq:solvable-end-maps}
\end{equation}
    and the normalized single-copy initial-state edge $a_L$ of each transfer step is unitary. These identities are required for the physical times and subsystem sizes under consideration, with $t_{\mathrm{phys}}>N_A/2$.
\end{definition}

 In fact, we demand no structure from $A_{L}$ and $A_{R}$ other than their solvability (although we will conjecture some properties for the boundaries to be generic, and we will assume that solvability does not restrict us to the vanishing subset of these boundaries where the properties do not hold). Solvable boundaries are explicitly constructed in \cite{Claeys2022emergentquantum}, and ensure unitary auxiliary steps and the stated boundary normalizations.

Define $U_\xi:=2C_\xi$. By Eq.\ \eqref{eq:solvable-columns}, each $U_\xi$ is unitary. At replica order $m$, they define the transfer matrix, which we present here in the notation that we have developed. 
\begin{equation}
\begin{aligned}
    \mathbb T^{(m)}
    &:=\frac14\sum_{\xi\in\{0,1\}^2}
      (U_\xi\otimes U_\xi^*)^{\otimes m}.
\end{aligned}
    \label{eq:normalized-auxiliary-transfer}
\end{equation}

It acts on a vectorized space of dimension $D^{2m}=2^{2mt}$, or equivalently on operators in $\mathcal L((\mathbb C^D)^{\otimes m})$. 

In our notation, we will call the left and right boundaries of $\mathbb T^{(m)}$ $A_L$ and $A_R$ respectively. Here $A_L=(a_L\otimes a_L^*)^{\otimes m}$ is the folded initial-state edge, distinct from the end map $F_A$. Write the outcome-dependent single-copy right edge as $a_{R,\xi}$ and set
\begin{equation}
    A_{R,\xi}:=(a_{R,\xi}\otimes a_{R,\xi}^*)^{\otimes m},
\end{equation}
\begin{equation}
    A_R:=\frac14\sum_{\xi\in\{0,1\}^2}A_{R,\xi}.
    \label{eq:averaged-right-edge}
\end{equation}
The single-copy column factors as $U_\xi=(a_L\otimes\mathbb I)u_{\mathrm{bulk}}(\mathbb I\otimes a_{R,\xi})$. Since $U_\xi$, $a_L$, and the space-time-dual bulk circuit $u_{\mathrm{bulk}}$ are unitary, each $a_{R,\xi}$ is unitary. The transfer matrix itself is a sum over unitary terms; it is $A_R$ which can be non-unitary when the measurement outcomes are averaged over, and other than this right boundary, nothing about this operator is contractive. Indeed, if we were to apply the transfer matrix without the right boundary to an arbitrary operator in $\mathcal L((\mathbb C^D)^{\otimes m})$, the Hilbert--Schmidt norm would stay constant.

Altogether, in our notation,
\begin{equation}
    \mathbb T^{(m)}
    =(A_L\otimes\mathbb I_{t-m_L})\,\mathbb U^{(m)}
      (\mathbb I_{t-m_R}\otimes A_R),
\end{equation}
\begin{equation}
    \mathbb U^{(m)}:=(u_{\mathrm{bulk}}\otimes u_{\mathrm{bulk}}^*)^{\otimes m},
    \label{eq:auxiliary-transfer-factorization}
\end{equation}
which we depict below. Here $m_L,m_R$ are the numbers of auxiliary sites in the boundary factors, $\mathbb I_s$ is the identity on $s$ folded auxiliary sites, and $\mathbb U^{(m)}$ includes the gates coupling the boundary sites to the interior. Note that in the figure, we have rotated the transfer matrix clockwise $90$ degrees --- the unitarity of its rotated bulk gates follows from dual-unitarity.

\begin{figure}[h!]
\centering
\begin{equation}
\mathbb{T}^{(m)} = \frac14\sum_{\xi=(z_0,z_1)\in\{0,1\}^2} \quad
\begin{tikzpicture}[
    baseline=(current bounding box.center), 
    scale=0.8,
    tensor/.style={rectangle, draw=black, thick, rounded corners=2pt, minimum size=20pt, inner sep=0pt, align=center},
    T/.style={tensor, fill=orange!80!white},
    leg/.style={thick},
    >=latex 
]

    \node[T] (AL) at (0,1) {$A_L$};
    \node[T] (B1) at (1,0) {$m$};
    \node[T] (T1) at (2,1) {$m$};
    
    % Internal connections (Left)
    \draw[leg] (AL.south east) -- (B1.north west);
    \draw[leg] (B1.north east) -- (T1.south west);
    \draw[leg] (T1.south east) -- ++(0.4, -0.4);

    % Left boundary wires with dots (A_L)
    \draw[leg] (AL.north west) -- ++(-0.4, 0.4) node[circle, fill=black, inner sep=1.2pt] {};
    \draw[leg] (AL.south west) -- ++(-0.4, -0.4) node[circle, fill=black, inner sep=1.2pt] {};

    % Open wires (Left)
    \draw[leg] (AL.north east) -- ++(0.4, 0.4);
    \draw[leg] (T1.north west) -- ++(-0.4, 0.4);
    \draw[leg] (T1.north east) -- ++(0.4, 0.4);
    \draw[leg] (B1.south west) -- ++(-0.4, -0.4);
    \draw[leg] (B1.south east) -- ++(0.4, -0.4);

    \node at (3.5, 1) {$\cdots$};
    \node at (3.5, 0) {$\cdots$};

    \node[T] (Bn) at (5,0) {$m$};
    \node[T] (Tn) at (6,1) {$m$};
    \node[T] (AR) at (7,0) {$A_{R,\xi}$};
    
    % Internal connections (Right)
    \draw[leg] (Bn.north west) -- ++(-0.4, 0.4);
    \draw[leg] (Bn.north east) -- (Tn.south west);
    \draw[leg] (Tn.south east) -- (AR.north west);

    % Open wires (Right)
    \draw[leg] (Tn.north west) -- ++(-0.4, 0.4);
    \draw[leg] (Tn.north east) -- ++(0.4, 0.4);
    \draw[leg] (Bn.south west) -- ++(-0.4, -0.4);
    \draw[leg] (Bn.south east) -- ++(0.4, -0.4);
    \draw[leg] (AR.south west) -- ++(-0.4, -0.4);

    % Right output arrows (A_R)
    \draw[leg, ->] (AR.north east) -- ++(0.3, 0.3) -- ++(0.4, 0) node[right, inner sep=2pt] {$z_0$};
    \draw[leg, ->] (AR.south east) -- ++(0.3, -0.3) -- ++(0.4, 0) node[right, inner sep=2pt] {$z_1$};

    \draw[<->, >=stealth, thick, draw=black!60] (0, -1.3) -- (7, -1.3) 
        node[midway, fill=white, font=\footnotesize, text=black!80] {$t$ auxiliary sites};

\end{tikzpicture}
\end{equation}
\caption{The $m$-fold transfer matrix with space and time flipped. Each folded site has local dimension $2^{2m}$ (or $q^{2m}$ for qudits), and the folded two-site operators $u^{\otimes m}\otimes u^{*\otimes m}$ are $2^{4m}\times2^{4m}$ matrices. The boundary tensors are normalized as in Equation~\eqref{eq:normalized-auxiliary-transfer}; averaging the fixed-outcome edge $A_{R,\xi}$ gives $A_R$.}
\label{fig:transfer-matrix}
\end{figure}

They then demonstrate that in the $N_{B} \rightarrow \infty$ limit, one has convergence of $\mathbb{T}^{(m)}$ to the (vectorization of) $\Pi^{(m)}$, by identifying that the fixed left- and right- eigenvectors of this operator are exactly the permutation operators. Since $\mathbb{T}^{(m)}$ is made up of unitary operators and contractions, it is itself a contraction, and hence as the bath size goes to infinity, all contributions from non-unimodular eigenvalues of $\mathbb{T}^{(m)}$ must vanish.

After this, they contract the boundaries of the space-time flipped circuit; these consist of contractions with an initial state, and the map ${F}_{A}$. They show that the initial state contraction contributes a scalar, while $F_{A}$ converts the permutations on the auxiliary $\mathbb{C}^{D}$ to those on $\mathbb{C}^{d}$. For $t_{\mathrm{phys}}>N_A/2$, the result is
proportional to the Haar $k$ moment, with the extra $n$ replicas affecting only the scalar prefactor. They handle this prefactor with an analytical continuation to $n = 1-k$.

Our work will use this exact setting, except for the assumption that $N_{B} \rightarrow \infty$. We study the exact same question, and follow this construction, but at every step ask: how does $N_{B}$ have to scale with $k$? This will complicate the mathematics considerably, and we will have to develop the analysis carefully. In the next section, we sketch the proof that the rest of the paper details.

\section{Proof sketch}
\label{sec:proofsketch}

Throughout this work, $\Vert\cdot\Vert$ denotes the Hilbert-space norm on vectors and the corresponding induced operator norm on linear maps. When an operator $O$ is regarded as a vector under vectorization, $\Vert O\Vert$ is therefore its Hilbert--Schmidt norm. We also use
\begin{equation}
    \|X\|_1:=\operatorname{tr}\sqrt{X^\dagger X},
    \qquad
    \|X\|_2:=\sqrt{\operatorname{tr}(X^\dagger X)},
    \qquad
    \|\Phi\|_{2\to2}:=\sup_{X\ne0}\frac{\|\Phi[X]\|_2}{\|X\|_2}.
    \label{eq:moment-norm-conventions}
\end{equation}

As discussed previously, $\mathrm{vec}(\cdot)$ will denote the \textit{vectorization} operation that one gets by taking the Choi state of the operator, and $P_{\pi}$ is the permutation matrix on $k$ qubits corresponding to a particular permutation in $S_{k}$. 

We break our proof into two broad sub-problems:

\begin{enumerate}
    \item Show that $\mathbb{T}^{(m)}$ converges to $\Pi^{(m)}$, and determine the rate of convergence as a function of the replica number $m$.

    \begin{enumerate}
        \item We begin by studying the dynamics of defects under evolution by $\mathbb{T}^{(m)}$. These defects are the set of all operators which lie outside the permutation sector, or equivalently, form $\mathrm{ker}(\Pi^{(m)})$. Since the only contractive part of $\mathbb{T}^{(m)}$ is at the right boundary $A_{R}$, defects which are localized far away from the right edge must somehow spread to the end, where they may be contracted.
        
        To this end, we replace $A_{R}$ by a hard absorber at the location $t$, thereby going from $\mathbb{T}^{(k)}$ to a different process, $\mathbb{K}_{t}$. We then make two assumptions. The first of these, local permutation mixing (LPM), ensures that the defects cannot be stalled; over two steps, they must either move or be absorbed. The second, a subspace recurrence bound, makes sure the defects do not get stuck within the bulk and bounds the right defect backflow current, modulo an initial transient. Combining these two assumptions, we show a positive worst-case defect current over a $\mathrm{poly}(m)$ number of steps, which means that the defects spread efficiently. This is detailed in Sections~\ref{sec:free-walk}--\ref{sec:negative-current}.
        
        \item We then show a `drift-or-hit' dichotomy for $\mathbb{K}_{t}$, which, along with the fact that defects can only drift so far on a finite chain, turns into a hitting time bound for the hard absorber. Using this, we show that a constant fraction of a defect is absorbed within
        $\mathrm{poly}(m,t)$ steps. This is the content of
        Section~\ref{sec:hard-absorber}.

        \item Finally, in Section ~\ref{sec:soft-absorption}, we transition back to the physical $\mathbb{T}^{(k)}$ by comparing the hard absorber to the physical $A_{R}$. We assume that $A_R$, aided by mixing from the brickwork circuit, contracts boundary defect weight
        by an inverse-polynomial amount over a block of size
        $q\leq\mathrm{poly}(m)$. With this, we are able to show that the comparison only loses a polynomial factor.

    Altogether, we are able to show that $\mathbb{T}^{(k)}$ contracts defect vectors a constant amount in a timescale that goes as $N_{B} \sim \mathrm{poly}(m,t)$.
    
    \end{enumerate}
    \item Show that applying $F_{A}$ to the result and accounting for the normalization indeed leads to state-design formation on $A$. For this, the main difficulty is the normalization of the projected states. Applying $F_A$ is linear, but normalizing its output introduces inverse powers of an outcome-dependent norm.

    We handle this by approximating the inverse power of the norm by a polynomial in the positive powers, which are controlled by the previous subproblem. At fixed accuracy and subsystem size, the required polynomial degree, and hence the largest replica order needed, grow only polynomially with $k$. The polynomial coefficients and dimension factors enter logarithmically in the required bath size, since the auxiliary error decays exponentially. 
    
    Combining the two parts gives approximate state-design formation for $N_B\leq\mathrm{poly}(k,t)$, under our quantitative assumptions. This is described in detail in \autoref{sec:normalization-to-design}.
\end{enumerate}

\section{Defect dynamics: structure and assumptions}
\label{sec:defect-structure}
In this section, we set up the terminology and assumptions that we need for our construction.

 \subsection{Local permutation sectors in $\mathcal{L}((\mathbb{C}^{2})^{\otimes 2m})$.}
\label{sec:local-permutation-sectors}

To warm up, we recall a well-known fact regarding brickwork dual-unitary circuits. Take a two-site $m$-fold dual-unitary operator $u_{i,i+1}$, and act with it on a single-site \textit{traceless} operator. How does this operator propagate through the unitary? It turns out that the dual unitarity of $u$ necessitates that the operator must spread --- it cannot stay fixed at the same site \cite{Bertini2026SpaceTimeDuality}. The same argument trivially holds if the identity is replaced by a permutation matrix, as long as the permutations match on both sides of the unitary. This is immediately suggestive to us --- we know that the projected-ensemble argument requires control of $\mathbb{T}^{(m)}$ outside the global permutation sector. Here, we see that dual unitarity assigns to the permutation operators a privileged position even in this very simple case.

Let us now look at the spread of an arbitrary operator across two of the $m$-fold local sites. Take $O \in \mathcal{L}((\mathbb{C}^{2})^{\otimes 2m})$. Let us write $\mathcal{L}((\mathbb{C}^{2})^{\otimes 2m})$ as $\mathcal{L}((\mathbb{C}^{2})^{\otimes m}) \otimes \mathcal{L}((\mathbb{C}^{2})^{\otimes m})$, separating the parts of the operator acting on each of the individual local systems ($m$-fold qubits). For brevity, we term each of these $\mathcal{V}$. We now use our previous discussion as an invitation to single out $S_{m}$ within $\mathcal{V}$. To this end, split $\mathcal{V}$ into $\mathcal{P} \oplus (\mathcal{V} \ominus \mathcal{P})$, where $\mathcal{P} = \mathrm{span}\{P_{\pi}, \pi \in S_{m} \}$. Our process yields 
\begin{equation}
    \mathcal{L}\Big((\mathbb{C}^{2})^{\otimes 2m}\Big)  = \mathcal{V} \otimes \mathcal{V} = \mathcal{V} \otimes (\mathcal{P} \oplus (\mathcal{V} \ominus \mathcal{P})) = (\mathcal{V} \otimes \mathcal{P}) \oplus (\mathcal{V} \otimes (\mathcal{V} \ominus \mathcal{P})).
\end{equation}
Define the two-local transparent subspace
\begin{equation}
    \mathcal{S}_{\mathrm{diag}} = \mathrm{span}(\{P_{\pi} \otimes P_{\pi}, \pi \in S_{m} \}).
\end{equation}
Since $\mathcal{S}_{\mathrm{diag}} \subset (\mathcal{V} \otimes \mathcal{P})$, we have $(\mathcal{V} \otimes \mathcal{P}) = \mathcal{S}_{\mathrm{diag}} \oplus ((\mathcal{V} \otimes \mathcal{P}) \ominus  \mathcal{S}_{\mathrm{diag}})$. Substituting this into the previous expression gets us
\begin{equation}
        \mathcal{L}\Big((\mathbb{C}^{2})^{\otimes 2m}\Big)  = \mathcal{S}_{\mathrm{diag}} \oplus ((\mathcal{V} \otimes \mathcal{P}) \ominus  \mathcal{S}_{\mathrm{diag}}) \oplus (\mathcal{V} \otimes (\mathcal{V} \ominus \mathcal{P})).
\end{equation}
The latter two spaces will turn out to have a nice physical interpretation in our model, and so we christen them with the dedicated notation
\begin{equation}
        \mathcal{L}\Big((\mathbb{C}^{2})^{\otimes 2m}\Big)  := \mathcal{S}_{\mathrm{diag}} \oplus \mathcal{B} \oplus \mathcal{G}.
\end{equation}

Altogether, we have split the Hilbert space of linear operators on $2m$ qubits into three parts: $\mathcal{S}_{\mathrm{diag}}$, which passes through the brickwork unitary unchanged, since these are matched permutations on both legs, the bad sector $\mathcal{B}$, which consists of defects for which the right local system lies in the permutation span, and the good sector $\mathcal{G}$, which consists of defects for which the right local system is orthogonal to the permutation span. A more detailed description of defects as well as a justification for this nomenclature will be provided in the next subsection.

For now, let us write down a basis for these spaces. A basis for $\mathcal{S}_{\mathrm{diag}}$ is given by taking $\{ P_{i} \otimes {P}_{i}\}$ and using the Gram-Schmidt process, discarding zero residuals (say starting with $\mathbb{I}^{\otimes m} \otimes \mathbb{I}^{\otimes m}$) to arrive at a set $\{S_{i}\}$. A basis for $\mathcal{B}$ is given by taking any product basis for $\mathcal{V} \otimes \mathcal{P}$, computing the orthogonal projection away from $\mathcal{S}_{\mathrm{diag}}$, and again using the Gram-Schmidt process. Lastly, a basis for $\mathcal{G}$ is given by projecting away from $\mathcal{S}_{\mathrm{diag}}$ and $\mathcal{B}$, and then (again) applying the Gram-Schmidt process. We call these latter two bases $\{\mathcal{B}_{i}\}$ and $\{\mathcal{G}_{i}\}$ respectively. Note also that every element in these latter two bases must have trace $0$, since it is orthogonal to the identity, which belongs to $\mathcal S_{\mathrm{diag}}$.

Our brickwork $m$-folded dual unitary takes an especially simple form in this basis. Anything in $\mathcal{S}_{\mathrm{diag}}$ must pass straight through one layer of the brickwork dual unitary. Further, it is impossible for any operator outside $\mathcal{S}_{\mathrm{diag}}$  to enter this subspace after propagation through $u$, because of unitarity in reverse. As a result, we have the following form for our $m$-fold brickwork unitary
\begin{equation}
    (u \otimes u^{*})^{\otimes m} = \begin{pmatrix}
        1_{\mathcal{S}_{\mathrm{diag}}, \mathcal{S}_{\mathrm{diag}}} & 0 & 0\\
        0 & (u \otimes u^{*})^{\otimes m}|_{\mathcal{B},\mathcal{B}} & (u \otimes u^{*})^{\otimes m}|_{\mathcal{B},\mathcal{G}} \\
        0 & (u \otimes u^{*})^{\otimes m}|_{\mathcal{G},\mathcal{B}} & (u \otimes u^{*})^{\otimes m}|_{\mathcal{G},\mathcal{G}}
    \end{pmatrix}.
\end{equation}
It is then useful to define the restricted matrix $\tilde{u}^{(m)}$ as 
\begin{align}
    \tilde{u}^{(m)} &:= \begin{pmatrix}
         (u \otimes u^{*})^{\otimes m}|_{\mathcal{B},\mathcal{B}} & (u \otimes u^{*})^{\otimes m}|_{\mathcal{B},\mathcal{G}} \\
         (u \otimes u^{*})^{\otimes m}|_{\mathcal{G},\mathcal{B}} & (u \otimes u^{*})^{\otimes m}|_{\mathcal{G},\mathcal{G}}
        \end{pmatrix}, \\&= \begin{pmatrix}
         \tilde{u}^{(m)}_{\mathcal{B},\mathcal{B}} & \tilde{u}^{(m)}_{\mathcal{B},\mathcal{G}} \\
         \tilde{u}^{(m)}_{\mathcal{G},\mathcal{B}} & \tilde{u}^{(m)}_{\mathcal{G},\mathcal{G}}
        \end{pmatrix}.
\end{align}

\subsection{The defect picture: right defect weight, and the checkerboard of defects.}
\label{sec:defect-picture}

We begin by defining the space of defects in the most obvious manner, as the complement of the global transparent space. We will consider $t$ single-qubit sites in the auxiliary Hilbert space, with $t$ odd, so that there are $n_{p}=(t-1)/2$ aligned pairs and the singleton site $1$. 

\begin{definition}[Interval transparent space]
\label{def:interval-transparent}
    Let $I \subseteq [1,t]$ be an interval. We define its transparent space $\mathcal{S}_{I}$ as $\mathrm{span}(\{\otimes_{I}P_{\pi}, \pi \in S_{m}\})$. In particular we have $\mathcal{S} = \mathcal{S}_{[1,t]}$.
\end{definition}

\begin{definition}[Defects]
\label{def:defects}
    Define $\mathcal{K} = \mathcal{S}^{\perp}$ as the space of all defects.
\end{definition}

The remainder of this section will focus on forming a simpler characterization of these defects, which we employ in studying design formation. We will do this in two ways, one of which is independent of the basis decomposition of the previous subsection, and another which has more synergy with the brickwork circuit's locality. Both of these constructions will be useful. As a rule of thumb, whenever we have to analyze the situation microscopically, and track the motion of defects, the basis picture becomes essential; from the bird's-eye view, however, basis independence is very useful. Let us start with this.

\begin{definition}[Nested defect space]
\label{def:nested-defect}
    Take $y \in \{1, \dots t\}$. We denote the projector onto the nested defect space at $y$ as $Q_{y}$ and define 
    \begin{equation}
        Q_{y} = \mathbb{I}_{1 \dots y} \otimes \Pi_{\mathcal{S}_{[y+1,\dots,t]}}
    \end{equation}
    with $Q_t=\mathbb I$ (the empty suffix), such that $Q_{i} \leq Q_{i+1}$ for all $i$. It is also convenient to define $P_{y} = Q_{y}-Q_{y-1}$, and $Q_{0} = \Pi_{\mathcal{S}}$.
\end{definition}

\begin{lemma}
\label{lem:defect-decomposition}
    The defect space $\mathcal{K} = \bigoplus_{y=1}^{t} \range(P_{y})$. Equivalently, $\mathbb{I}_{\mathcal{K}} = \sum_{y} P_{y}$, and the $P_{i}$ are mutually orthogonal.
\end{lemma}
 \begin{proof}
     This follows directly from the definition of $Q_{y}$. Since the various $P_{i}$ are the differences of nested orthogonal projectors, they must be mutually orthogonal. Further, their sum must be equivalent to $Q_{t}-Q_{0}$ itself, since the sum telescopes, and hence gives $\mathbb{I}_{\mathcal{K}}$. Since this is a sum of mutually orthogonal projectors, it is also a projector, and it projects onto its range, so $\mathcal{K} = \bigoplus_{y=1}^{t} \range(P_{y})$.
 \end{proof}

\begin{definition}[Right defect edge]
\label{def:right-defect-edge}
    Consider an operator $O \in \mathcal{L}((\mathbb{C}^{2})^{\otimes mt})$. We define its right defect edge $r(O)$, as that $y \in \{1, \dots t \}$ such that $\Vert P_yO\Vert \neq 0$ but $\Vert P_{y'}O\Vert=0$ for all $y'>y$. If there is no right defect edge, we say $r(\cdot) = \emptyset$.
\end{definition}

\begin{definition}[Right defect subspace]
    Define the right defect subspace at $y \in \{1, \dots t \}$ as $\range(P_{y})$.
\end{definition}

\begin{definition}[Right defect weight]
    Define the right defect weight at $y \in \{1, \dots t \}$ of an operator $O \in \mathcal{L}((\mathbb{C}^{2})^{\otimes mt})$ as $\Vert P_{y}O\Vert^{2}$.
\end{definition}

Let us briefly discuss how these defects travel in the bulk upon application of our transfer matrix, which we split into its aligned and staggered sublayers for clarity. With $u^{(m)}_{i,i+1}$ denoting the folded gate, define
\begin{equation}
    \mathbb A:=A_L\prod_{c=1}^{(t-1)/2}u^{(m)}_{2c,2c+1},
\end{equation}
\begin{equation}
    \mathbb S:=\prod_{c=0}^{(t-3)/2}u^{(m)}_{2c+1,2c+2},
\end{equation}
\begin{equation}
    \mathbb W:=\mathbb S\mathbb A.
    \label{eq:finite-unitary-layers}
\end{equation}
Here $A_L$ acts on site $1$, is unitary, and fixes every local permutation; empty products are identities. The change from the auxiliary-column ordering is given in Section~\ref{sec:soft-absorption}.

\begin{figure}[htbp]
\centering
\begin{tikzpicture}[
    scale=1.15,
    % Standard template styles
    tensor/.style={rectangle, draw=black, thick, rounded corners=3pt, fill=orange!80!white, minimum size=24pt, font=\small\bfseries},
    defectG/.style={rectangle, draw=green!60!black, thick, fill=green!20, fill opacity=0.8, rounded corners=4pt, inner sep=3pt, font=\small\bfseries, text=green!50!black},
    defectB/.style={rectangle, draw=red!60!pink, thick, fill=red!20, fill opacity=0.8, rounded corners=4pt, inner sep=2.5pt, font=\scriptsize\bfseries, text=red!70!black},
    redcirc/.style={circle, draw=red!80!black, thick, fill=white, inner sep=1.2pt, font=\scriptsize, text=red!80!black},
    leg/.style={thick, draw=black!80},
    arrowleg/.style={thick, draw=red!80!black, ->, >=stealth}
]

    \begin{scope}[shift={(0,0)}]
        % Subcaption
        \node[font=\bfseries] at (0, 1.8) {(a) $G$-defect upward travel};

        % Vertical Red Arrow drawn FIRST (behind the tensor)
        \draw[arrowleg] (0, -0.8) -- (0, 1.1);

        % Single Sublayer Tensor
        \node[tensor] (K) at (0, 0) {$m$};

        % Leg connections with centered decorations
        \draw[leg] (K.north west) -- ++(-1.0, 1.0);
        \draw[leg] (K.north east) -- ++(1.0, 1.0);
        \draw[leg] (K.south east) -- ++(1.0, -1.0) node[pos=0.5, redcirc] {$\sigma$};
        \draw[leg] (K.south west) -- ++(-1.0, -1.0) node[pos=0.5, defectG] {$G$};

        \node[text=red!80!black, font=\scriptsize, align=center] at (0, -1.4) {Travels forward\\vertically through $k$};
    \end{scope}

    \begin{scope}[shift={(5.8,0)}]
        % Subcaption
        \node[font=\bfseries] at (0, 1.8) {(b) $B$-defect ($B = \sigma'$)};

        % Single Sublayer Tensor
        \node[tensor] (K) at (0, 0) {$m$};

        % Leg connections with combined B = \sigma' box centered on the lower-left leg
        \draw[leg] (K.north west) -- ++(-1.0, 1.0);
        \draw[leg] (K.north east) -- ++(1.0, 1.0);
        \draw[leg] (K.south east) -- ++(1.0, -1.0) node[pos=0.5, redcirc] {$\sigma$};
        \draw[leg] (K.south west) -- ++(-1.0, -1.0) node[pos=0.5, defectB] {$B = \sigma'$};

        \node[text=red!80!black, font=\scriptsize, align=center] at (0, -1.4) {Possible stalled\\channel};
    \end{scope}

\end{tikzpicture}
\caption{Upward defect propagation channels through a single $m$ tensor.}
\label{fig:single-gate-defect-channels}
\end{figure}
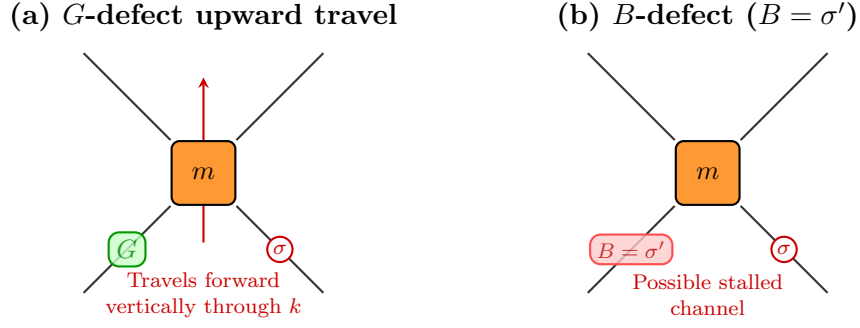

\begin{lemma}[Locality of the defect sectors]
\label{lem:sector-locality}
    A gate on $(a,a+1)$ commutes with $Q_y$ for $y\ne a$, and $A_L$ commutes with every $Q_y$. Consequently, the gate preserves $\range(P_a+P_{a+1})$ and each of the other single-site sectors.
\end{lemma}
\begin{proof}
    If $y\geq a+1$, the gate acts within the identity factor of $Q_y$. If $y\leq a-1$, it acts inside the transparent suffix and fixes every matched permutation string there. Its adjoint does so too, so it commutes with the corresponding projector. The same argument applies to $Q_{0}$ and to $A_L$. Taking differences proves the statements for $P_y$, and then grouping the sectors proves the statements for the layers.
\end{proof}

For the aligned sublayer, the situation is therefore simple: right defect weight cannot change its even cell or leave the singleton sector at site $1$. For the staggered sublayer, weight can move between the two single-site sectors in an odd cell, which belong to neighbouring even cells. This is the checkerboard mechanism by which the paired defect position changes.

\begin{figure}[htbp]
\centering
\begin{tikzpicture}[
    scale=1.15,
    % Standard template styles
    tensor/.style={rectangle, draw=black, thick, rounded corners=3pt, fill=orange!80!white, minimum size=24pt, font=\small\bfseries},
    defectG/.style={rectangle, draw=green!60!black, thick, fill=green!20, fill opacity=0.8, rounded corners=4pt, inner sep=3pt, font=\small\bfseries, text=green!50!black},
    defectB/.style={rectangle, draw=red!60!pink, thick, fill=red!20, fill opacity=0.8, rounded corners=4pt, inner sep=2.5pt, font=\scriptsize\bfseries, text=red!70!black},
    redcirc/.style={circle, draw=red!80!black, thick, fill=white, inner sep=1.2pt, font=\scriptsize, text=red!80!black},
    leg/.style={thick, draw=black!80},
    arrowleg/.style={thick, draw=red!80!black, ->, >=stealth}
]

    \begin{scope}[shift={(0,0)}]
        % Subcaption
        \node[font=\bfseries] at (0, 1.8) {(a) Non-permutation leg};

        % Vertical Red Arrow drawn FIRST (behind the tensor)
        \draw[arrowleg] (0, -0.8) -- (0, 1.1);

        % Single Sublayer Tensor
        \node[tensor] (K) at (0, 0) {$m$};

        % Leg connections with centered decorations
        \draw[leg] (K.north west) -- ++(-1.0, 1.0);
        \draw[leg] (K.north east) -- ++(1.0, 1.0);
        \draw[leg] (K.south east) -- ++(1.0, -1.0) node[pos=0.5, redcirc] {$\sigma$};
        \draw[leg] (K.south west) -- ++(-1.0, -1.0) node[pos=0.5, defectG] {$O\perp\mathcal P$};

        \node[text=red!80!black, font=\scriptsize, align=center] at (0, -1.4) {Forward propagation\\channel};
    \end{scope}

    \begin{scope}[shift={(5.8,0)}]
        % Subcaption
        \node[font=\bfseries] at (0, 1.8) {(b) Permutation leg};

        % Single Sublayer Tensor
        \node[tensor] (K) at (0, 0) {$m$};

        % Leg connections with combined B = \sigma' box centered on the lower-left leg
        \draw[leg] (K.north west) -- ++(-1.0, 1.0);
        \draw[leg] (K.north east) -- ++(1.0, 1.0);
        \draw[leg] (K.south east) -- ++(1.0, -1.0) node[pos=0.5, redcirc] {$\sigma$};
        \draw[leg] (K.south west) -- ++(-1.0, -1.0) node[pos=0.5, defectB] {$\sigma'$};

        \node[text=red!80!black, font=\scriptsize, align=center] at (0, -1.4) {Possible stalled\\channel};
    \end{scope}

\end{tikzpicture}
\caption{Local propagation channels in the checkerboard picture. The labels on the legs distinguish a component orthogonal to the local permutation span from a permutation component. The orthogonal defect weights used in the proof are defined by the projectors $P_y$.}
\label{fig:upward_defect_behavior}
\end{figure}
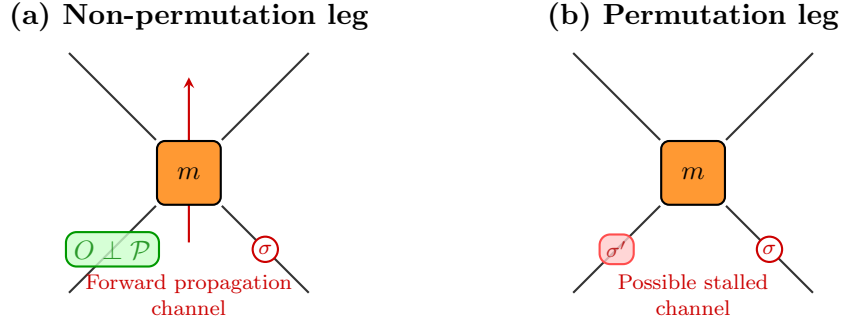

\begin{definition}[Position operator]
\label{def:position-operator}
    Define the paired position and the position from the end by
    \begin{equation}
    X:=\sum_{y=1}^{t}\left\lfloor\frac y2\right\rfloor P_y,
\end{equation}
\begin{equation}
    X_{\max}:=\frac{t-1}{2},
\end{equation}
\begin{equation}
    X^{\mathrm{end}}:=X-X_{\max}\mathbb I_{\mathcal K}.
\end{equation}
\end{definition}

For a normalized defect $O$, $\braket{O|X|O}$ provides the expectation value of the paired location of its right defect weight. The aligned layer preserves this position, so $[\mathbb A,X]=0$. We use this same $X$ throughout, including when discussing the odd pairing.

To study contraction, we will replace the physical right boundary by a hard absorber. Let us restrict ourselves to $\mathcal K$ for the time being, and define
\begin{equation}
    P:=P_t,
\end{equation}
\begin{equation}
    Q:=\sum_{y=1}^{t-1}P_y=Q_{t-1}-\Pi_{\mathcal S},
\end{equation}
\begin{equation}
    \mathbb K_t:=Q\mathbb WQ.
    \label{eq:finite-hard-walk}
\end{equation}
Thus $P+Q=\mathbb I_{\mathcal K}$, and $Q$ projects onto surviving defects. On the full operator space,
\begin{equation}
    P_t=\mathbb I_{[1,t-1]}\otimes(\mathbb I-\Pi_{\mathcal P}).
\end{equation}
The absorber removes the component whose last site is orthogonal to the local permutation span. The staggered layer commutes with $P$ and $Q$, by \autoref{lem:sector-locality}. Hence, for $v\in\range(Q)$,
\begin{equation}
    \|v\|^2-\|\mathbb K_t v\|^2
       =\|P_t\mathbb A v\|^2.
\end{equation}
The vectors $\mathbb K_t^n v$ are not renormalized: their decreasing squared norm accounts for absorption. All subsequent currents and backward-current subspaces are defined on $\mathcal K$ before restricting initial vectors to $\range(Q)$.

\subsection{Making defects move: the local permutation mixing assumption on dual-unitary operators}
\label{sec:lpm}

Now that we have developed the concept of a defect, we describe our method of modeling $k$-design formation in the transfer matrix setting. Operator weight can be split into two subspaces: the globally transparent $\mathcal{S}$, and the space of defects, which can be further decomposed into the single-site sectors $\range(P_y)$, or grouped into the checkerboard of even and odd cells. The idea is to treat $A_R$ as a lone source of defect contraction. The bulk, being unitary, preserves operator weight no matter what goes into it, and only at the right boundary can the norm reduce. Approximate $k$-design formation requires that every operator outside the transparent subspace is contracted, with the rate of contraction matching the design time. This motivates us to explain design formation as a two-step physical process, where

\begin{enumerate}
    \item Right defect weight flows \textit{in the worst case} to $A_R$ fast enough as a consequence of bulk dual unitary dynamics;
    \item The right defect weight, once it reaches $A_R$, gets contracted at some rate that is bounded from below.
\end{enumerate}

Let us focus now on the first point. If all of the defects flowed in a simple, coherent way to the boundary, things would be quite straightforward for us. We could treat the unitary process as some sort of Markov chain on the right defect weight, and then simply calculate the hitting time to the boundary. We saw earlier that the aligned layer cannot move weight between even cells, whereas the staggered layer couples neighboring even cells. The local $\mathcal G$ sector detects a non-permutation component on the right leg, while the $\mathcal B$ sector does not. In the checkerboard picture, permutations may match across the shifted pairing, allowing backward motion. This situation is illustrated in \autoref{fig:backflow-defects-main}. Unfortunately for us, it turns out that the staggered layer can push defects backwards.

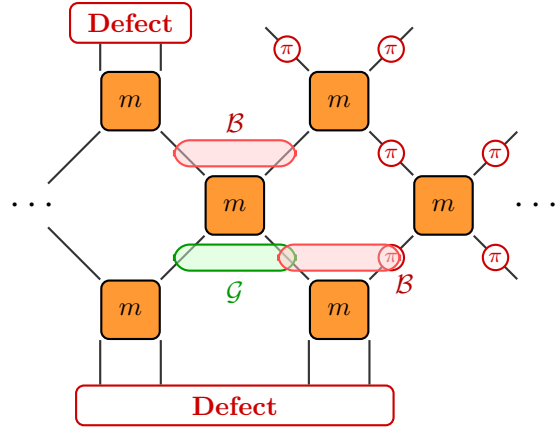
\begin{figure}[h!]
\centering
\begin{tikzpicture}[
baseline=(current bounding box.center),
scale=1.15,
% Standard template styles
tensor/.style={rectangle, draw=black, thick, rounded corners=3pt, fill=orange!80!white, minimum size=22pt, font=\small\bfseries},
defect/.style={rectangle, draw=red!80!black, thick, fill=white, rounded corners=3pt, inner sep=4pt, font=\small\bfseries, text=red!80!black},
redcirc/.style={circle, draw=red!80!black, thick, fill=white, inner sep=1.2pt, font=\scriptsize, text=red!80!black},
leg/.style={thick, draw=black!80},
% Capsule operator styles
capsuleB/.style={draw=red!60!pink, thick, fill=red!20, fill opacity=0.5, rounded corners=6pt},
capsuleY/.style={draw=green!60!black, thick, fill=green!20, fill opacity=0.5, rounded corners=6pt}
]

% Top-Left Defect Box
\node[defect, minimum width=1.6cm] (DefectTop) at (-1.8, 2.7) {Defect};

% Bottom Defect Box spanning both bottom tensors
\node[defect, minimum width=4.2cm] (DefectBot) at (-0.6, -1.7) {Defect};

% Black Boundary Dots (\cdots)
\node[text=black, font=\Large] at (-2.9, 0.6) {$\cdots$};
\node[text=black, font=\Large] at (2.9, 0.6) {$\cdots$};

\node[tensor] (TTL)  at (-1.8, 1.8)  {$m$};
\node[tensor] (TMID) at (0.6, 1.8)   {$m$};
\node[tensor] (CMID) at (-0.6, 0.6)  {$m$};
\node[tensor] (TR)   at (1.8, 0.6)   {$m$};
\node[tensor] (BL)   at (-1.8, -0.6) {$m$};
\node[tensor] (BR)   at (0.6, -0.6)  {$m$};

% --- Top-Left Tensor (TTL) ---
\draw[leg] (TTL.north west) -- (TTL.north west |- DefectTop.south);
\draw[leg] (TTL.north east) -- (TTL.north east |- DefectTop.south);
\draw[leg] (TTL.south west) -- ++(-0.6, -0.6);
\draw[leg] (TTL.south east) -- (CMID.north west);

% --- Top-Middle Tensor (TMID) ---
\draw[leg] (TMID.north west) -- ++(-0.5, 0.5) node[pos=0.5, redcirc] {$\pi$};
\draw[leg] (TMID.north east) -- ++(0.5, 0.5) node[pos=0.5, redcirc] {$\pi$};
\draw[leg] (TMID.south west) -- (CMID.north east);
\draw[leg] (TMID.south east) -- (TR.north west) node[pos=0.5, redcirc] {$\pi$};

% --- Center-Middle Tensor (CMID) ---
\draw[leg] (CMID.south west) -- (BL.north east);
\draw[leg] (CMID.south east) -- (BR.north west);

% --- Bottom-Left Tensor (BL) ---
\draw[leg] (BL.north west) -- ++(-0.6, 0.6);
\draw[leg] (BL.south west) -- (BL.south west |- DefectBot.north);
\draw[leg] (BL.south east) -- (BL.south east |- DefectBot.north);

% --- Bottom-Right Tensor (BR) ---
\draw[leg] (BR.north east) -- (TR.south west) node[pos=0.5, redcirc] {$\pi$};
\draw[leg] (BR.south west) -- (BR.south west |- DefectBot.north);
\draw[leg] (BR.south east) -- (BR.south east |- DefectBot.north);

% --- Top-Right Tensor (TR) ---
\draw[leg] (TR.north east) -- ++(0.5, 0.5) node[pos=0.5, redcirc] {$\pi$};
\draw[leg] (TR.south east) -- ++(0.5, -0.5) node[pos=0.5, redcirc] {$\pi$};

% Upper B Capsule (spanning legs above CMID)
\draw[capsuleB] (-1.3, 1.05) rectangle (0.1, 1.35);
\node[text=red!70!black, font=\small\bfseries] at (-0.6, 1.55) {$\mathcal{B}$};

% Middle Green y Capsule (spanning legs below CMID)
\draw[capsuleY] (-1.3, -0.15) rectangle (0.1, 0.15);
\node[text=green!60!black, font=\small\bfseries] at (-0.6, -0.4) {$\mathcal{G}$};

% Lower Red B Capsule (spanning two legs and intersecting green y)
\draw[capsuleB] (-0.1, -0.15) rectangle (1.3, 0.15);
\node[text=red!70!black, font=\small\bfseries] at (1.35, -0.3) {$\mathcal{B}$};
\end{tikzpicture}
\caption{The figure shows how backflow of defects may occur microscopically in our setup. The idea is essentially that if we have two $\mathcal{B}$ sector defects in a row, then if the permutations align just right, the defect can flow backwards. Since design formation requires worst-case contraction of everything outside the globally transparent sector, this is a genuine problem. Notice that the second $\mathcal{B}$ sector defect must come from a $\mathcal{G} \rightarrow \mathcal{B}$ transition.}
\label{fig:backflow-defects-main}
\end{figure}

It is important that we find some structural property of generic dual-unitary circuits which prevents this indefinite backflow of defects --- otherwise, defects could very well be stuck in the bulk forever, never arriving at the boundary. We will assume a bound on accumulated backward current which permits an initial transient. A full description of this assumption requires some work, so we defer it to the next subsection for now. 

Other than control on backflow, we also need the right defect weight to have some guaranteed mobility. We accomplish this by assuming that a local gate pushes a definite fraction of every $\mathcal B$-sector vector into $\mathcal G$.

\begin{definition}[Local permutation mixing (LPM)]
\label{def:lpm}
    Let $u \in \mathrm{DU}_{(2)}$, and take the $m$-fold $(u \otimes u^{*})^{\otimes m}$. Define $u^{(m)}$ as in the previous section. We define $u$ to be $\epsilon(m)$ local permutation mixing ($\epsilon(m)$-LPM) if $\Vert \Pi_{\mathcal B}u^{(m)}\Pi_{\mathcal B}\Vert \leq 1 -  \epsilon(m)$.
\end{definition}

This condition ensures that every bad-sector vector leaks into the good sector in one application. Equivalently, for $v\in\mathcal B$,
\begin{equation}
    \|\Pi_{\mathcal G}u^{(m)}v\|^2
    \geq\bigl(1-(1-\epsilon(m))^2\bigr)\|v\|^2.
\end{equation}
For good convergence, we will require that the largest singular value of $u^{(m)}_{\mathcal B,\mathcal B}$ is at least $1/\mathrm{poly}(m)$ below $1$. 

We conjecture that this is true of a generic dual unitary. Specifically, a stronger $m$-independent gap would follow from the exact universality of single-site unitaries and a fixed dual unitary.
\begin{lemma}\label{lem:dugap}
For a dual unitary $u\in \mathrm{DU}_{(2)}$, let $H = I \otimes SU(2)$ and $H'= u^\dagger H u = \{u^\dagger (I\otimes v)u ~|~ v\in SU(2)\}$. If $H \cup H'$ exactly generates $SU(4)$, i.e.\ if for all $g\in SU(4)$ we have $g = h_1 h_2 \ldots h_\ell$ with $h_i\in H \cup H'$ and $\ell\leq L$ for some finite $L$, then
\begin{equation}
    \Vert \Pi_{\mathcal B}u^{(m)}\Pi_{\mathcal B}\Vert \leq 1 - \frac{1}{4L^2}\,.
\end{equation}
\end{lemma}

\begin{proof}
For a fixed dual unitary $u$, let $H = \{ I\otimes v ~|~ v\in SU(2)\}$ and $H' = u^\dagger Hu = \{ u^\dagger (I\otimes v) u ~|~ v\in SU(2)\}$, where $H, H'\subset SU(4)$. Again, we use the notation for $m$-folded unitaries $U^{(m)} = U^{\otimes m}\otimes U^*{}^{\otimes m} $. Just as $\Pi_{SU(4)} = \int_{SU(4)} U^{(m)} \dd\mu(U)$ is the projector onto the span of permutation operators for $(\mathbb{C}^d)^{\otimes m}$, we define projectors
\begin{equation}
    \Pi_{H} = \int (I\otimes v)^{(m)} \dd\mu(v) \quad{\rm and}\quad \Pi_{H'} = \int \big(u^\dagger (I\otimes v) u\big)^{(m)} \dd\mu(v)\,.
\end{equation}
Now note that the projector onto the subspace of bad defects $\Pi_\mathcal{B} $ is simply $\Pi_\mathcal{B} = \Pi_{H} (I-\Pi_{SU(4)})$. 
We first write the norm as
\begin{align}
    \|\Pi_{\mathcal B}u^{(m)}\Pi_{\mathcal B}\|^2 &= \sup_\psi \vev{\psi| \Pi_{\mathcal B}u^{(m)}{}^\dagger\Pi_{\mathcal B} \Pi_{\mathcal B}u^{(m)}\Pi_{\mathcal B}|\psi} = \sup_{\psi\in \mathcal{B}} \vev{\psi| u^{(m)}{}^\dagger \Pi_{\mathcal B}u^{(m)}|\psi}\\
    &=\sup_{\psi\in \mathcal{B}} \vev{\psi| u^{(m)}{}^\dagger \Pi_{H} \Pi_{SU(4)}^\perp u^{(m)}|\psi} = \sup_{\psi\in \mathcal{B}} \vev{\psi| u^{(m)}{}^\dagger \Pi_{H} u^{(m)}|\psi}\,,
\end{align}
where in the last step we used that $(I-\Pi_{SU(4)}) u^{(m)} = u^{(m)} (I-\Pi_{SU(4)})$ by the invariance of the Haar measure and that $(I-\Pi_{SU(4)})\ket\psi = \ket\psi$ for states $\ket\psi\in \mathcal{B}$. Thus, by the definition of $H'$ we simply need to consider 
\begin{equation}
    \|\Pi_{\mathcal B}u^{(m)}\Pi_{\mathcal B}\|^2 = \sup_{\psi\in \mathcal{B}} \vev{\psi| \Pi_{H'} |\psi}\,.
\end{equation}
For any fixed $\ket\psi\in \mathcal{B}$, we equivalently consider $\vev{\psi| \Pi_{H'} |\psi} = 1- \vev{\psi| \Pi_{H'}^\perp |\psi}$ and lower bound the second term. 

First, for any $h\in H'$ note that $(h^{(m)}-I)\ket\psi =  (h^{(m)}-I)(I-\Pi_{H'})\ket\psi$ as $h^{(m)}\Pi_{H'}=\Pi_{H'}$ by the invariance of the Haar measure and thus $(h^{(m)}-I) \Pi_{H'} \ket\psi=0$. Proceeding, we can bound
\begin{equation}
    \| (h^{(m)}-I)\ket\psi\| = \|(h^{(m)}-I)(I-\Pi_{H'})\ket\psi\|
    \leq 2 \| (I-\Pi_{H'})\ket\psi\|\,,
\end{equation}
as $\|(h^{(m)}-I)\|_\infty \leq 2$.

We now use the exact universality of $H\cup H'$ for $SU(4)$. For every $g\in SU(4)$, there exists a decomposition $g=g_1 g_2 \ldots g_\ell$ with $\ell\leq L$ for some finite $L$, and where each $g_i\in H\cup H'$. Meaning we can construct any $SU(4)$ unitary from single-site unitaries conjugated by a fixed dual unitary. If $g\in H'$ we know that
\begin{equation}
    \| (g^{(m)}-I)\ket\psi\| \leq 2 \| (I-\Pi_{H'})\ket\psi\|\,,
\end{equation}
and if $g\in H$ then $\| (g^{(m)}-I)\ket\psi\|=0$ as $\ket\psi\in \CB$ so $\Pi_H\ket\psi=\ket\psi$. Writing the product as a telescoping sum $(\prod_i g_i-I) = \sum_i g_1\ldots g_{i-1}(g_i-I) $ we find for any $g \in SU(4)$
\begin{align}
    \| (g^{(m)}-I)\ket\psi\| &= \| ((g_1\ldots g_\ell)^{(m)}-I)\ket\psi\|\leq \sum_i \|g_1\ldots g_{i-1}\|_\infty \| (g_i^{(m)}-I)\ket\psi\|\\
    &\leq 2 L \| (I-\Pi_{H'})\ket\psi\|\,,
\end{align}
using the fact that $\|g_i\|_\infty=1$. Lastly, we simply integrate $\| (g^{(m)}-I)\ket\psi\|^2$ over the Haar measure to compute the desired lower bound. Proceeding, integrate over the Haar measure on $SU(4)$
\begin{equation}
    \int \| (g^{(m)}-I)\ket\psi\|^2\, \dd\mu(g) = 2-2\int \vev{\psi|g^{(m)}|\psi}\, \dd\mu(g) = 2 \vev{\psi|(I-\Pi_{SU(4)})|\psi} = 2\,,
\end{equation}
as $\ket\psi\in \CB$ so $(I-\Pi_{SU(4)})\ket\psi=1$. We conclude that $\| (I-\Pi_{H'})\ket\psi\|^2 \geq 1/2L^2$ and thus
\begin{equation}
    \|\Pi_{\mathcal B}u^{(m)}\Pi_{\mathcal B}\|^2 \leq 1 - \frac{1}{2L^2}\,,
\end{equation}
from which the claim follows.

\end{proof}

\begin{conjecture}[Almost all dual unitaries are LPM]
\label{conj:lpm}
For almost every fixed two-qubit dual unitary $u$, there is an LPM bound
\begin{equation}
    \epsilon(m)\geq c(u)>0, \qquad m\geq 1\,,
\end{equation}
such that there is a constant lower bound for all $m\geq 1$, where the constant depends on the choice of $u$.
\end{conjecture}

The conjecture would follow from \autoref{lem:dugap} and the exact universality of dual unitaries and single qubit unitaries. Specifically, we conjecture that the above constant is nonzero for all dual unitaries $u$ that are not self dual.

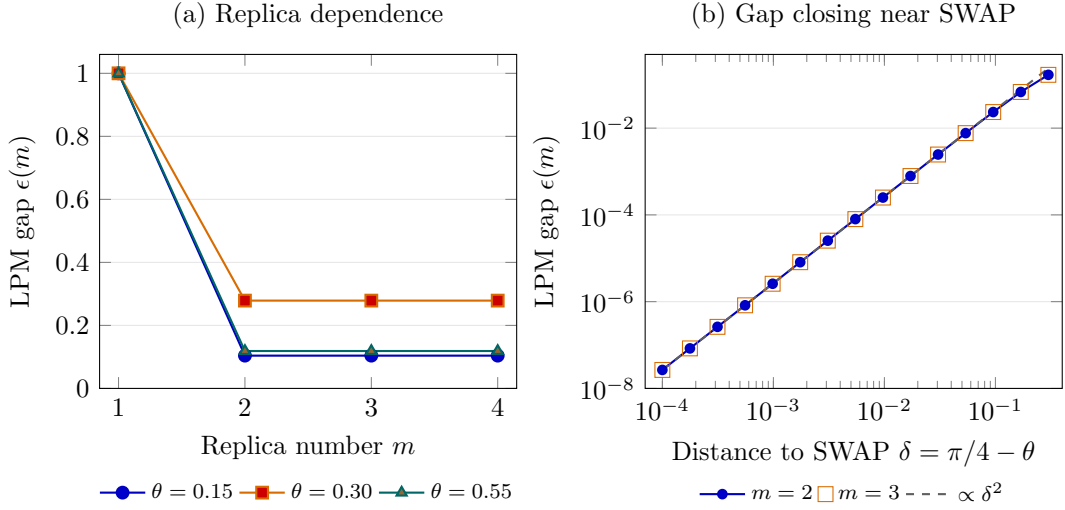
\begin{figure}[t]
\centering
\begin{tikzpicture}
\begin{groupplot}[
 group style={group size=2 by 1,horizontal sep=1.7cm},
 width=0.43\textwidth,height=6.0cm,
 tick label style={font=\small},label style={font=\small},
 title style={font=\small},ymajorgrids=true,grid style={gray!20},
 legend style={font=\scriptsize,draw=none,fill=none,
   at={(0.5,-0.25)},anchor=north,legend columns=3},
 ]
\nextgroupplot[title={(a) Replica dependence},
 xlabel={Replica number $m$},ylabel={LPM gap $\epsilon(m)$},
 xmin=0.85,xmax=4.15,xtick={1,2,3,4},ymin=0,ymax=1.06,
 ytick={0,0.2,0.4,0.6,0.8,1}]
\addplot+[color=blue!75!black,mark=*,mark size=2.2pt,thick] coordinates {(1,1) (2,0.1040234829854) (3,0.1040234829854) (4,0.1040234829854)};
\addlegendentry{$\theta=0.15$}
\addplot+[color=orange!85!black,mark=square*,mark size=2.2pt,thick] coordinates {(1,1) (2,0.2786754754636) (3,0.2786754754636) (4,0.2786754754636)};
\addlegendentry{$\theta=0.30$}
\addplot+[color=teal!80!black,mark=triangle*,mark size=2.2pt,thick] coordinates {(1,1) (2,0.1186929658332) (3,0.1186929658332) (4,0.1186929658332)};
\addlegendentry{$\theta=0.55$}
\nextgroupplot[title={(b) Gap closing near SWAP},
 xlabel={Distance to SWAP $\delta=\pi/4-\theta$},
 ylabel={LPM gap $\epsilon(m)$},
 xmode=log,ymode=log,xmin=0.00007,xmax=0.4,
 ymin=1e-8,ymax=0.5,
 xtick={0.0001,0.001,0.01,0.1},
 ytick={1e-8,1e-6,1e-4,1e-2},
 legend columns=3]
\addplot+[blue!75!black,thick,mark=*,mark size=1.6pt] coordinates {(0.0001,2.66666649873e-08) (0.0001771600539669,8.369514870044e-08) (0.0003138568472156,2.626828847774e-07) (0.000556028959906,8.244475245833e-07) (0.0009850612054411,2.587578165847e-06) (0.001745134963167,8.121223839308e-06) (0.003091682042541,2.548835306604e-05) (0.005477225575052,7.999040056339e-05) (0.009703455784651,0.0002509909293116) (0.01719064750474,0.0007871179711396) (0.03045496039666,0.002464186139696) (0.05395402427433,0.007672890328104) (0.09558497852173,0.02348964656746) (0.1693383995334,0.06820237741734) (0.3,0.1692460221178)};
\addlegendentry{$m=2$}
\addplot+[orange!85!black,only marks,mark=square,mark size=2.8pt] coordinates {(0.0001,2.666666454321e-08) (0.0001771600539669,8.369514736817e-08) (0.0003138568472156,2.62682883112e-07) (0.000556028959906,8.244475238062e-07) (0.0009850612054411,2.587578164182e-06) (0.001745134963167,8.121223838087e-06) (0.003091682042541,2.548835306493e-05) (0.005477225575052,7.99904005625e-05) (0.009703455784651,0.0002509909293104) (0.01719064750474,0.0007871179711388) (0.03045496039666,0.002464186139696) (0.05395402427433,0.007672890328104) (0.09558497852173,0.02348964656745) (0.1693383995334,0.06820237741734) (0.3,0.1692460221178)};
\addlegendentry{$m=3$}
\addplot[gray!75!black,dashed,thick,domain=0.0001:0.3,samples=50]
 {2.6666664987295974*x^2};
\addlegendentry{$\propto\delta^2$}
\end{groupplot}
\end{tikzpicture}
\caption{We consider dual unitary gates of the form $u_{\theta} = \exp[-i(\frac{\pi}{4}XX) + \frac{\pi}{4}YY +\theta ZZ)]$. Panel (a) clearly shows how the LPM gap is $m$-independent for $m \geq 2$. In (b), we study the convergence of the LPM to $0$ as a function of the distance $\delta = \pi/4 - \theta$ to the SWAP gate. The graph is consistent with a quadratic approach, which connects back to \autoref{lem:dugap}.}
\label{fig:lpm-numerics}
\end{figure}

Let us circle back to the second point presently. What could go wrong with boundary contraction? For one, maybe the defects that reach the end do not get contracted. We will separately assume a quantitative contraction gap for $A_R$ on the absorbing sector. What could happen, however, is that defects near the boundary remain in the $\mathcal B$ sector and never enter the sector where they are contracted. This is also addressed by LPM. We will investigate the resulting mobility-or-absorption bound in \autoref{sec:unsigned-current}.

\subsection{Bounding defects' backward motion: backflow recurrence and the first-return assumption}
\label{sec:first-return}

We now return to limiting the backflow of right defect weight, starting with a review of \cite{Bourgain2014Recurrence}. A reader with knowledge of the contents of Section 2 of this reference may safely skip the next subsection and proceed to the one after.

\subsubsection*{Subspace recurrence in quantum processes}
\label{sec:subspace-recurrence}
Consider a discrete-time quantum process, where at each time-step we apply a \textit{fixed} unitary $U$ to a state initialized to $\ket{\psi}$. Fix some privileged finite-dimensional subspace $V$ of the overall Hilbert space $\mathcal{H}$, and demand that $\ket{\psi}$ belongs to this subspace. We study the dynamics of this system, and in particular the relation of the time evolution of the state to the subspace $V$. For notational simplicity, define $\tilde{U} = (\mathbb{I}-\Pi_{V})U$. Notice that $\tilde{U}\ket{\psi}$ has norm at most $1$, and roughly quantifies the part of the state that remains outside of $V$ after one iteration of the unitary $U$. We define a few important operators on $V$:

\begin{definition}[$n$-step return operator]
\label{def:return-operator}
    Consider the setting from above, with the unitary operator applied at each time-step given by $U$, and the privileged subspace given by $V$. We define the $n$-step return operator 
    \begin{equation}
    \mu_n  = \Pi_{V}U^{n}\Pi_{V}.
    \end{equation}
\end{definition}

\begin{definition}[$n$-step first-return operator]
\label{def:first-return-operator}
    Retain the same setting with the unitary $U$, and the privileged subspace $V$. Recall also the notation $\tilde{U} = (1-\Pi_{V})U$. Then the $n$-step first-return operator 
    \begin{equation}
    a_{n}  = \Pi_{V}U\tilde{U}^{n-1}\Pi_{V}.
    \end{equation}
\end{definition}

\begin{definition}[Generating operators for the discrete-time process]
\label{def:return-generating-operators}
Let $\mu_{n}$ and $a_{n}$ be defined as in \autoref{def:return-operator} and \autoref{def:first-return-operator}. We also inherit their setting. Let $z \in \mathbb{C}$, $|z| < 1$. Then the generating operators corresponding to the discrete-time quantum process are analytic functions on the unit disc with values in $\mathcal{L}(V)$ given by 
\begin{equation}
    \hat{\mu}(z) = \sum_{n\geq 0} \mu_{n}z^{n},
\end{equation}
\begin{equation}
    \hat{a}(z) = \sum_{n\geq 1} a_{n}z^{n}.
\end{equation}
\end{definition}
\noindent The Green's function for the operators $U$ and $\tilde{U}$ are given by 
\begin{equation}
    G(z) = (\mathbb I_{\mathcal H} - zU)^{-1},
\end{equation}
\begin{equation}
    \tilde{G}(z) = (\mathbb I_{\mathcal H} - z\tilde{U})^{-1}.
\end{equation}
These can be checked to obey the identities
\begin{equation}
\label{eq:green-resolvent}
    G(z) - \tilde{G}(z) = z\tilde G(z)\Pi_VUG(z),
\end{equation}
\begin{equation}
\label{eq:compressed-green}
    \Pi_V\tilde G(z)\Pi_V=\Pi_V.
\end{equation}

These identities imply the so-called renewal equation for the generating operators, which can also be written in terms of the single-time operators $a_{n}$ and $\mu_{n}$.
\begin{equation}
\label{eq:renewal}
    (\mathbb{I}_{V} - \hat{a}(z))\hat{\mu}(z) = \mathbb{I}_{V},
\end{equation}
\begin{equation}
    \mu_n=\sum_{j=1}^{n}a_j\mu_{n-j},
\end{equation}
\begin{equation}
    \mu_0=\mathbb I_V,\ n\geq1.
\end{equation}

We now define the Schur function of the process, which motivates our later first-return formulation. Define the projection-valued spectral measure $E(d\lambda)$ of $U$ on the unit circle $S^1\subset\mathbb C$, that is, such that $U = \int_{S^{1}}\lambda E(d\lambda)$. Define the spectral measure of the subspace $V$ as $\mu(d \lambda) := \Pi_{V}E(d\lambda)\Pi_{V}$. 

\begin{definition}[Schur function]
\label{def:schur-function}
    Retain the setting of the above discussion. Let
    \begin{equation}
        F(z):=\int_{S^1}\frac{\lambda+z}{\lambda-z}\,\mu(d\lambda).
    \end{equation}
    We call this the Carath\'eodory function associated with the walk. Then the Schur function on the open disc $|z|<1$ is defined as (with the value at $z=0$ obtained by continuity)

    \begin{equation}
        f(z) = \frac{1}{z}(F(z) - \mathbb{I}_{V})(F(z) +\mathbb{I}_{V})^{-1}.
    \end{equation}
\end{definition}

Ref.~\cite{Bourgain2014Recurrence} then proves the following for the Schur function and its relation to the generating operators corresponding to the return operators.

\begin{fact}[Below Eq.~(6) in \cite{Bourgain2014Recurrence}]
\label{fact:schur-first-return}
    \begin{equation}
        \hat a(z)=z f(\bar z)^\dagger.
    \end{equation}
\end{fact}

\begin{fact}[Proposition 2.2 of \cite{Bourgain2014Recurrence}]
\label{fact:schur-singular-measure}
    The measure $\mu(d\lambda)$ is purely singular if and only if $f(z)$ is inner, that is, $f(e^{i\theta})$ is unitary for almost every $\theta$. In particular, $\mu(d \lambda)$ is finitely supported iff $f$ is a rational inner function.
\end{fact}

\begin{fact}[Theorem 2.4 of \cite{Bourgain2014Recurrence}]
\label{fact:recurrence-schur}
    Let $V$ be finite-dimensional. Then all the states of $V$ are recurrent iff $f(z)$ is inner. We call $\ket{\psi} \in V$ recurrent if it returns to $V$ with probability $1$ under monitoring after each step.
\end{fact}

Let us comment on the significance of the facts that we have established so far. We fixed a privileged subspace $V$ of an overall Hilbert space, took a state within it, and had it undergo some time-invariant dynamics. We now ask: when does the state return to $V$? And how much of it returns to $V$? The operator $\mu_{n}$ from earlier is then the operator that takes a state in $V$ back to $V$ after $n$ iterations of the unitary $U$. Here, in the $n-1$ intermediate steps, the state is allowed to enter and/or exit any subspace of its choice, it is completely unrestricted. The quantity $\|\mu_n\ket{\psi}\|^2$ tells us what portion of the evolved state is in $V$ after $n$ steps; these probabilities do not form a distribution over return times. A similar analysis holds for $a_{n}$. However, due to the contraction $\tilde{U}$, we remove the parts of $\ket{\psi}$ within $V$ in any step before the last. In this way, we get instead the portion of the state that returns to $V$ for the first time at time step $n$ in the monitored process.

One can also ask whether the state always returns to $V$. This is the content of Facts~\ref{fact:schur-first-return}, \ref{fact:schur-singular-measure}, and \ref{fact:recurrence-schur}. Build the operator $\hat{a}(z)$. This $a(z)$ turns out to be intimately related to the Schur function $f(z)$, which in turn controls the recurrence of the walk based on whether it is inner or not. Innerness concerns almost-everywhere radial boundary values on the unit circle. For the absorbing process below, we instead prove a renewal identity inside the disc and impose a quantitative bound on its coefficients.

\subsubsection*{The finite absorbing process and the first-return assumption}
\label{sec:bulk-first-return}

Let us return to our setting. Here, our states are the vectorization of operators, and the discrete-time quantum process is given by our $m$-folded circuit. However, $\mathbb K_t$ is not unitary, so the results of our previous discussion do not carry over directly. We use the finite walk from Equation~\eqref{eq:finite-hard-walk}, which already includes the physical left boundary $A_L$.

Intuitively, we want to make sure that not too much defect weight flows backwards, so that the entirety of the state is visible to $A_R$ within a reasonable time-frame. We adapt the Schur function story to our purpose, beginning (as always) with some definitions.

\begin{definition}[$h$-step current]
\label{def:m-step-current}
    Define the ($m$-fold) $h$-step right defect edge current as
    \begin{equation}
        J_h^m:=(\mathbb W^\dagger)^hX\mathbb W^h-X.
    \end{equation}
\end{definition}

For a normalized defect $O$, $\bra{O}J_h^m\ket{O}$ gives the expected \textit{change} in paired position after $h$ steps of the free evolution. This is easy to see --- expanding out gives the difference between the expectation of $X$ after $h$ steps and its initial expectation. We write $J:=J_1^m$, when the replica order is understood.

Let us now focus on the one-step current. It is Hermitian, so we can perform a spectral decomposition. Its negative part measures backward current, and its absolute value measures unsigned current.

\begin{definition}[Negative one-step current]
\label{def:negative-current}
Let $J$ be the one-step current from \autoref{def:m-step-current}. Define
\begin{equation}
    J_-:=\frac{|J|-J}{2},
\end{equation}
\begin{equation}
    J_+:=\frac{|J|+J}{2}.
\end{equation}
We have positive operators $J_-,J_+ \geq 0$, such that $J=J_+-J_-$ and $|J|=J_++J_-$.
\end{definition}

We consider the backward current subspace as the support of this operator.

\begin{definition}[Backward current subspace]
\label{def:backward-current-subspace}
We define the ($m$-fold) backward current subspace
\begin{equation}
    \mathcal U_m:= \range \mathbf 1_{(-\infty,0)}(J),
\end{equation}
\begin{equation}
    \Pi_{\mathcal U_m}=\mathbf 1_{(-\infty,0)}(J).
\end{equation}
\end{definition}

This immediately presents itself as a candidate for subspace recurrence analysis. The idea is that if the parts of an operator which flow backwards find it difficult to rearrange themselves into this same subspace at later times, the right defect weight will tend to flow right. We are thus invited to take $\mathcal U_m$ as the privileged subspace $V$ in the discussion above. For our purposes, however, we need to control how much an arbitrary state returns to this subspace, rather than impose a blanket ban on returns. 

Further, an initially extended defect may undergo a long backward transient. Consider a state $w$ that undergoes evolution under the free unitary walk $\mathbb W$ for some timescale $H$, such that its right-defect spreads by some $O(H)$. If we consider the (dual-unitary) backwards circuit, we have that the right defect edge shrinks for this period. Thus, it is possible to organize coherences in such a way that the right defect weight shrinks backwards in such a reversed free history. Since the maximum width we can reach is $t$, this motivates allowing a transient that is roughly $O(t)$. Any assumption we make on the dual-unitary circuit must allow for this. 

\begin{figure}[t]
\centering
\begin{tikzpicture}[
    scale=0.95,
    tensor/.style={
        rectangle, draw=black, thick, rounded corners=3pt,
        fill=orange!80!white, minimum size=20pt,
        font=\small\bfseries
    },
    leg/.style={thick, draw=black!75},
    front/.style={very thick, draw=red!75!black},
    guide/.style={dashed, draw=red!55!black, thick},
    lbl/.style={font=\scriptsize},
    title/.style={font=\small\bfseries},
]

\begin{scope}[xshift=-4.6cm]

\node[title] at (0,3.7) {(a) Forward evolution};

\fill[red!20,opacity=.42,rounded corners=2pt]
(-2.75,-1.65) -- (-0.70,-1.65) -- (1.65,3.05) -- (-2.75,3.05) -- cycle;

% right defect edge
\draw[front] (-0.70,-1.65) -- (1.65,3.05);

% endpoints
\fill[red!75!black] (-0.70,-1.65) circle (2pt);
\fill[red!75!black] ( 1.65, 3.05) circle (2pt);

% layer 1
\node[tensor] (L11) at (-1.8,-1.0) {$m$};
\node[tensor] (L12) at ( 0.6,-1.0) {$m$};

% layer 2
\node[tensor] (L21) at (-0.6,0.15) {$m$};
\node[tensor] (L22) at ( 1.8,0.15) {$m$};

% layer 3
\node[tensor] (L31) at (-1.8,1.30) {$m$};
\node[tensor] (L32) at ( 0.6,1.30) {$m$};

% layer 4
\node[tensor] (L41) at (-0.6,2.45) {$m$};
\node[tensor] (L42) at ( 1.8,2.45) {$m$};

% bottom legs
\draw[leg] (L11.south west) -- ++(-0.6,-0.55);
\draw[leg] (L11.south east) -- ++( 0.6,-0.55);
\draw[leg] (L12.south west) -- ++(-0.6,-0.55);
\draw[leg] (L12.south east) -- ++( 0.6,-0.55);

% L1 -> L2
\draw[leg] (L11.north west) -- ++(-0.6,0.55);
\draw[leg] (L11.north east) -- (L21.south west);

\draw[leg] (L12.north west) -- (L21.south east);
\draw[leg] (L12.north east) -- (L22.south west);

\draw[leg] (L22.south east) -- ++(0.6,-0.55);

% L2 -> L3
\draw[leg] (L21.north west) -- (L31.south east);
\draw[leg] (L21.north east) -- (L32.south west);

\draw[leg] (L22.north west) -- (L32.south east);
\draw[leg] (L22.north east) -- ++(0.6,0.55);

\draw[leg] (L31.south west) -- ++(-0.6,-0.55);

% L3 -> L4
\draw[leg] (L31.north west) -- ++(-0.6,0.55);
\draw[leg] (L31.north east) -- (L41.south west);

\draw[leg] (L32.north west) -- (L41.south east);
\draw[leg] (L32.north east) -- (L42.south west);

\draw[leg] (L42.south east) -- ++(0.6,-0.55);

% top legs
\draw[leg] (L41.north west) -- ++(-0.6,0.55);
\draw[leg] (L41.north east) -- ++( 0.6,0.55);
\draw[leg] (L42.north west) -- ++(-0.6,0.55);
\draw[leg] (L42.north east) -- ++( 0.6,0.55);

% time arrow
\draw[->, thick] (-3.25,-1.7) -- (-3.25,3.15);
\node[lbl, rotate=90] at (-3.50,0.7) {time};

% labels on front
\node[lbl, text=red!75!black, anchor=east]
    at (-0.52,-1.28) {$r(0)$};

\node[lbl, text=red!75!black, anchor=west]
    at (1.77,3.02) {$r(H)$};

% defect label
\node[font=\scriptsize,text=red!65!black] at (-1.8,0.6) {defect};

\node[font=\scriptsize, text=red!65!black] at (-1.8,0.3) {support};

\end{scope}

\begin{scope}[xshift=4.6cm]

\node[title] at (0,3.7)
    {(b) Same history viewed backwards};

\fill[red!20, opacity=.42, rounded corners=2pt] (-2.75,-1.65) -- (1.65,-1.65) -- (-0.70,3.05) -- (-2.75,3.05) -- cycle;

% shrinking right edge
\draw[front] (1.65,-1.65) -- (-0.70,3.05);

% endpoints
\fill[red!75!black] ( 1.65,-1.65) circle (2pt);
\fill[red!75!black] (-0.70, 3.05) circle (2pt);

\node[tensor] (R11) at (-1.8,-1.0) {$m$};
\node[tensor] (R12) at ( 0.6,-1.0) {$m$};

\node[tensor] (R21) at (-0.6,0.15) {$m$};
\node[tensor] (R22) at ( 1.8,0.15) {$m$};

\node[tensor] (R31) at (-1.8,1.30) {$m$};
\node[tensor] (R32) at ( 0.6,1.30) {$m$};

\node[tensor] (R41) at (-0.6,2.45) {$m$};
\node[tensor] (R42) at ( 1.8,2.45) {$m$};

% bottom
\draw[leg] (R11.south west) -- ++(-0.6,-0.55);
\draw[leg] (R11.south east) -- ++( 0.6,-0.55);
\draw[leg] (R12.south west) -- ++(-0.6,-0.55);
\draw[leg] (R12.south east) -- ++( 0.6,-0.55);

% R1 -> R2
\draw[leg] (R11.north west) -- ++(-0.6,0.55);
\draw[leg] (R11.north east) -- (R21.south west);

\draw[leg] (R12.north west) -- (R21.south east);
\draw[leg] (R12.north east) -- (R22.south west);

\draw[leg] (R22.south east) -- ++(0.6,-0.55);

% R2 -> R3
\draw[leg] (R21.north west) -- (R31.south east);
\draw[leg] (R21.north east) -- (R32.south west);

\draw[leg] (R22.north west) -- (R32.south east);
\draw[leg] (R22.north east) -- ++(0.6,0.55);

\draw[leg] (R31.south west) -- ++(-0.6,-0.55);

% R3 -> R4
\draw[leg] (R31.north west) -- ++(-0.6,0.55);
\draw[leg] (R31.north east) -- (R41.south west);

\draw[leg] (R32.north west) -- (R41.south east);
\draw[leg] (R32.north east) -- (R42.south west);

\draw[leg] (R42.south east) -- ++(0.6,-0.55);

% top
\draw[leg] (R41.north west) -- ++(-0.6,0.55);
\draw[leg] (R41.north east) -- ++( 0.6,0.55);
\draw[leg] (R42.north west) -- ++(-0.6,0.55);
\draw[leg] (R42.north east) -- ++( 0.6,0.55);

\draw[->, thick] (-3.25,-1.7) -- (-3.25,3.15);
\node[lbl, rotate=90] at (-3.50,0.7) {time};

\node[lbl, text=red!75!black, anchor=west] at (1.77,-1.48) {$r(0)$};

\node[lbl, text=red!75!black, anchor=east] at (-0.92,2.72) {$r(H)$};

% defect label
\node[font=\scriptsize, text=red!65!black] at (-1.8,0.6) {defect};

\node[font=\scriptsize, text=red!65!black] at (-1.8,0.3) {support};
\end{scope}

\end{tikzpicture}

\caption{
Finite-width states with the right correlations can lead to long backflows, as illustrated in the figure. A fast moving defect for one dual unitary moves backwards when considering the Hermitian conjugate circuit. This does not reverse the absorbing process.
}
\label{fig:backward-transient}
\end{figure}
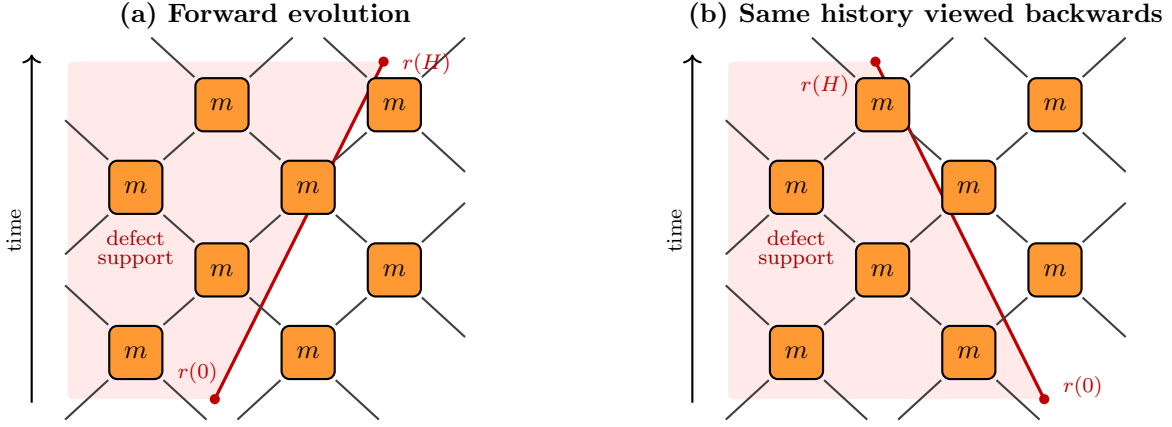
We state our assumption in terms of the first-return generating function. Let us first fix
\begin{equation}
    C:=\mathbb I_{\mathcal K}-\Pi_{\mathcal U_m},
\end{equation}
\begin{equation}
    \widetilde{\mathbb K}_t:=C\mathbb K_t.
\end{equation}

\begin{definition}[First-return and generating operator for the contractive process.] Let $\mathbb{K}_{t}$ be the contractive process from before, and $\mathcal{U}_{m}$ be the backward current subspace.  Then, for replica order $m$ and $n \geq 1$, we define 
\begin{equation}
    a_{m,n}:=\Pi_{\mathcal U_m}\mathbb K_t
       \widetilde{\mathbb K}_t^{\,n-1}\Pi_{\mathcal U_m}
       \big|_{\mathcal U_m},\ n\geq1,
\end{equation}
to be the $n$-step first-return operator corresponding to $\mathbb{K}_{t}$, and for $|z|<1$
\begin{equation}
    \widehat a_m(z):=\sum_{n\geq1}a_{m,n}z^n
\end{equation}
to be the generating operator for the walk.
\end{definition}

\begin{definition}[First arrival generating operator]
Retain the same setting as before, and let $v$ be an arbitrary vector. Then, we define the first arrival generating operator as
    \begin{equation}
    \widehat A_v(z):=\Pi_{\mathcal U_m}v
       +z\Pi_{\mathcal U_m}\mathbb K_t
           (\mathbb I_{\mathcal K}-zC\mathbb K_t)^{-1}Cv.
\end{equation}
\end{definition}

We note the similarity of the first two definitions to our discussion in the previous section. All the same intuition carries forward, although we have to be careful of importing the exact results. In that setting $U$ was unitary, here $\mathbb{K}_{t}$ is instead a contraction. Here, in fact $\widehat a_m(z)/z$ is an operator-valued Schur function and $\mathbb I_{\mathcal U_m}-\widehat a_m(z)$ is invertible for $|z|<1$. This is a little technical, and the details are not needed for the main argument. 

Let us interpret the third definition presently. The first term on the right-hand side is simple: it projects the vector $v$ into the backward current subspace. Let us focus on the second term, and expand it out by writing out the geometric series for the inverse. We get $z\Pi_{\mathcal U_m}\mathbb K_t(\mathbb I_{\mathcal K}-zC\mathbb K_t)^{-1}Cv = z\Pi_{\mathcal U_m}\mathbb K_t \sum_{i=0}^{\infty}(zC\mathbb{K}_{t})^{i}Cv = \sum_{i=0}^{\infty} z^{i+1}\Pi_{\mathcal U_m}\mathbb K_t(C\mathbb{K}_{t})^{i}Cv$. This is a generating function of the first arrival of the vector $v$ into the backward current subspace. The $C$ projects away from here, and the sum adds up the contributions of the various first arrival times, weighted by powers of $z$. This is very similar in spirit to the $\widehat a_{m}$ that we have dealt with so far.

\begin{definition}
Let $f$ be an analytic vector-valued function  $f(z)=\sum_{n\geq0}f_nz^n$. We define the Hardy $H^2$ norm of the Taylor polynomial through degree $H-1$ as 
\begin{equation}
    \|f\|_{H^2;H}^{\,2}
    :=\frac{1}{2\pi}\int_0^{2\pi}
       \left\|\sum_{n=0}^{H-1}f_ne^{in\theta}\right\|^2\,d\theta
    =\sum_{n=0}^{H-1}\|f_n\|^2,\ H\geq1.
    \label{eq:finite-horizon-hardy}
\end{equation}

\end{definition}

The last equality is Parseval's identity. This squared norm can be thought of as the ``response energy" that counts the amount of squared norm that is present at times $0,\ldots,H-1$ of the signal.

\begin{conjecture}[Finite-chain first-return bound (A2)]
\label{conj:first-return-gap}
\label{conj:finite-backflow}
    Fix a bulk gate $u$ and solvable $A_{L}$. Further assume that there exist width-dependent constants $b_m(t)\geq0$ and $0\leq\beta_m(t)<1/2$ such that for every $v\in\range(Q)$ and every integer $H\geq1$,
    \begin{equation}
    \begin{aligned}
        &\left\|J_-^{1/2}
          (\mathbb I_{\mathcal U_m}-\widehat a_m(z))^{-1}
          \widehat A_v(z)\right\|_{H^2;H}^{\,2}\\
        &\qquad\leq\beta_m(t)
          \left\||J|^{1/2}
          (\mathbb I_{\mathcal K}-z\mathbb K_t)^{-1}v
          \right\|_{H^2;H}^{\,2}
          +b_m(t)\|v\|^2.
    \end{aligned}
    \label{eq:finite-schur-backflow}
    \end{equation}
    For polynomial bounds, we require
    \begin{equation}
    b_m(t)\leq\operatorname{poly}(m,t).
\end{equation}
\begin{equation}
    (1-2\beta_m(t))^{-1}
        \leq\operatorname{poly}(m,t).
        \label{eq:finite-backflow-polynomial}
\end{equation}
\end{conjecture}

Let us go step by step in interpreting the conjecture, starting with $(\mathbb I_{\mathcal U_m}-\widehat a_m(z))^{-1}
=\sum_{r=0}^{\infty}\widehat a_m(z)^r$. Here it is instructive to look at $\widehat a_m^{2} = \sum_{n_1, n_2} a_{m,n_1}a_{m,n_2}z^{n_1 +n_2}$. This $a_{m,n_1}a_{m,n_2}$ is the two-return operator with the first return at $n_{2}$, and the second at $n_{2}+n_{1}$. The $n'$ coefficient then sums up all of the amplitudes that we have to return twice to $\mathcal{U}_{m}$, with the second return at $n'$. The sum over $r$ then adds up all numbers of returns. 

This is a general intuitive guideline to keep in mind: multiplying generating operators (or functions) is equivalent to convolution of the underlying coefficients, similar to how it is in Fourier analysis. Multiplying $(\mathbb I_{\mathcal U_m}-\widehat a_m(z))^{-1}$ with $\widehat A_{v}$ then also performs a convolution of the same sort, and counts all possible entries and returns into the backward current subspace. Finally, we weight the amplitude by $J_{-}^{1/2}$, so its squared norm measures the negative current. 

Now let us look at the right-hand side. Here we have instead $(\mathbb I_{\mathcal K}- z\mathbb{K}_{t})^{-1}$. By a similar argument as for the previous inverse, this represents summing up all possible contractive process outcomes given a vector $v$. These are all then weighted by their unsigned current. Lastly, we have a $t$-dependent allowance for backflow that is meant to account for the transient, correlated packets. 

All put together, we are trying to say that the backflow return dynamics weighted by how much the negative current is, is dominated by the unsigned mobility. In Section \ref{sec:hard-absorber}, we will use this to show that defects must travel to the contractive boundary.

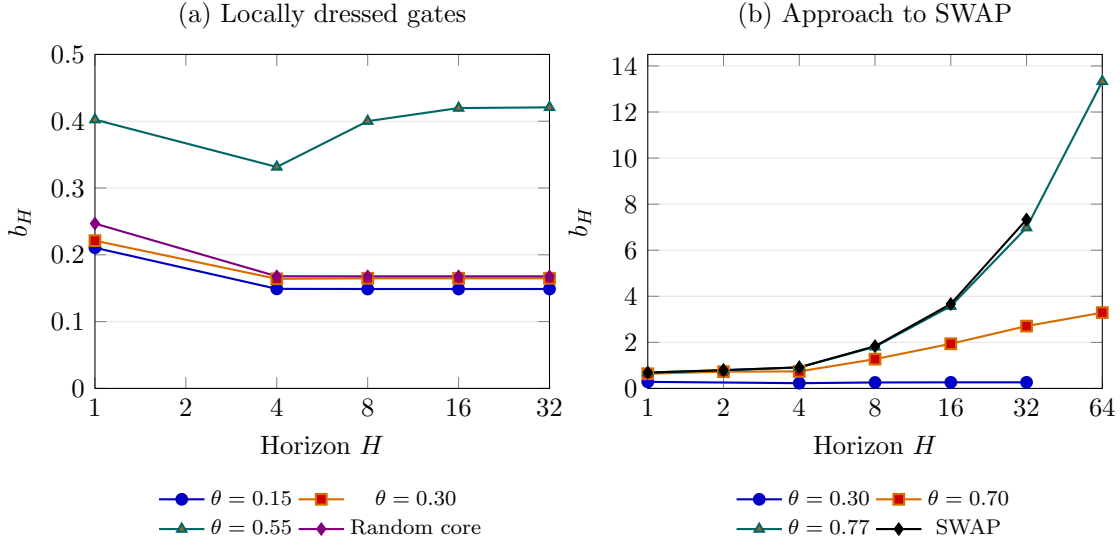
\begin{figure}[t]
\centering
\begin{tikzpicture}
\begin{groupplot}[
 group style={group size=2 by 1,horizontal sep=1.3cm},
 width=0.46\textwidth,height=6.0cm,
 xmode=log,log basis x=2,
 xlabel={Horizon $H$},ylabel={$b_H$},
 xtick={1,2,4,8,16,32,64},
 xticklabels={1,2,4,8,16,32,64},
 tick label style={font=\small},label style={font=\small},
 title style={font=\small},
 ymajorgrids=true,grid style={gray!20},
 legend style={font=\scriptsize,draw=none,fill=none,
   at={(0.5,-0.27)},anchor=north,legend columns=2},
 ]
\nextgroupplot[title={(a) Locally dressed gates},xmin=1,xmax=32,
 ymin=0,ymax=0.5,ytick={0,0.1,0.2,0.3,0.4,0.5}]
\addplot+[color=blue!75!black,mark=*,mark size=2pt,thick] coordinates {(1,0.210576557264) (4,0.149035595346) (8,0.148824776802) (16,0.148824660005) (32,0.148824660005)};
\addlegendentry{$\theta=0.15$}
\addplot+[color=orange!85!black,mark=square*,mark size=2pt,thick] coordinates {(1,0.221172790143) (4,0.164062988713) (8,0.164792384258) (16,0.164792270347) (32,0.164792270321)};
\addlegendentry{$\theta=0.30$}
\addplot+[color=teal!80!black,mark=triangle*,mark size=2pt,thick] coordinates {(1,0.402352704248) (4,0.331492309648) (8,0.400102485831) (16,0.419709467174) (32,0.420790860792)};
\addlegendentry{$\theta=0.55$}
\addplot+[color=violet,mark=diamond*,mark size=2pt,thick] coordinates {(1,0.246800154098) (4,0.167867315283) (8,0.167793569913) (16,0.167793539529) (32,0.167793539529)};
\addlegendentry{Random core}
\nextgroupplot[title={(b) Approach to SWAP},xmin=1,xmax=64,
 ymin=0,ymax=14.5,ytick={0,2,4,6,8,10,12,14}]
\addplot+[color=blue!75!black,mark=*,mark size=2pt,thick] coordinates {(1,0.284615856201) (4,0.227256924895) (8,0.258585371342) (16,0.264112105042) (32,0.264211767109)};
\addlegendentry{$\theta=0.30$}
\addplot+[color=orange!85!black,mark=square*,mark size=2pt,thick] coordinates {(1,0.637964146319) (2,0.727839243755) (4,0.738520059794) (8,1.27057423742) (16,1.93872195583) (32,2.70567227935) (64,3.29091479725)};
\addlegendentry{$\theta=0.70$}
\addplot+[color=teal!80!black,mark=triangle*,mark size=2pt,thick] coordinates {(1,0.685637057655) (2,0.790055450541) (4,0.908632607798) (8,1.80656192641) (16,3.57106197403) (32,6.97801155294) (64,13.3296226903)};
\addlegendentry{$\theta=0.77$}
\addplot+[color=black,mark=diamond*,mark size=2pt,thick] coordinates {(1,0.687386354243) (2,0.792352455979) (4,0.916515138991) (8,1.83303027798) (16,3.66606055596) (32,7.33212111193)};
\addlegendentry{SWAP}
\end{groupplot}
\end{tikzpicture}
\caption{We again consider dual unitary gates of the form $u_{\theta} = \exp[-i(\frac{\pi}{4}XX) + \frac{\pi}{4}YY +\theta ZZ)]$. In (a) these are dressed with independently sampled Haar random single unitaries at each leg, while in (b), we get rid of the local dressings, and study the approach to the SWAP gate, which is solitonic. We fix $\beta = 1/4$, and look at what value of $b$ is needed to control all transients at horizon $H$, calling it $b_{H}$. The conjecture requires a single, polynomially bounded $b_m(t)$ satisfying $b_H\leq b_m(t)$ for every $H$. The numerics in (a) show that $b_H$ seems to saturate with $H$, as required. In (b), however, in the SWAP case, $b_{H}$ does not stabilize at all, and as we get closer to it, the transients take longer and longer to do so.}
\label{fig:finite-chain-recurrence-numerics}
\end{figure}

\section{Right defect dynamics of the finite chain}
\label{sec:free-walk}

In this section, we focus on the finite-chain dynamics of defects under the discrete-time process $\mathbb W$ defined in the previous section. We begin by showing that LPM forces a defect either to move or to be absorbed, and then use the first-return assumption to control how much of this motion is backwards. Our strategy for the same is as follows.

\begin{enumerate}
    \item \textbf{LPM implies mobility}. Using LPM, we show that the sum of the two-step unsigned current and the two-step absorption is lower bounded by some constant that depends on the replica number $m$. This argument proceeds in three steps: first, we show that LPM brickwork unitaries are inherently mobile, that is, microscopically, they move a fraction of defect weight from one leg to the next. This is the content of \autoref{lem:finite-tail-mobility}. We then put this in the context of the contracting walk, $\mathbb{K}_{t}$. Here, we show that the current over one iteration of the process of the \textit{unabsorbed} part of a vector is lower bounded by how much right defect weight it moves (\autoref{lem:current-sector-transfer}). Then, we show that the process cannot stall the right defect of the unabsorbed vector either over two time periods (\autoref{lem:stalled-two-step}). Putting these together gives a lower bound on the overall unsigned current in terms of the fraction that was absorbed, as well as a constant that is fixed by LPM. 
    \item {\textbf{Backwards subspace recurrence bounds negative current}}. Here we use our second dynamical conjecture to bound the fraction of the unsigned current that backflows over a block, modulo an initial transient. Using $J= |J|-2J_{-}$, this then yields a lower bound on the positive current over a block in terms of the fraction of norm that was absorbed, modulo the same initial transient.
\end{enumerate}

We retain the physical left boundary $A_L$ throughout. Recall that $t$ is odd, the sites are $1,\ldots,t$, and
\begin{equation}
    \mathbb W=\mathbb S\mathbb A,
\end{equation}
\begin{equation}
    Q=\sum_{y=1}^{t-1}P_y,
\end{equation}
\begin{equation}
    P=P_t,
\end{equation}
\begin{equation}
    \mathbb K_t=Q\mathbb WQ.
\end{equation}
The projectors $P,Q$ act on $\mathcal K$. By \autoref{lem:sector-locality},
\begin{equation}
    [\mathbb A,X]=0,
\end{equation}
\begin{equation}
    [\mathbb S,Q]=[\mathbb S,P]=0.
    \label{eq:finite-layer-commutations}
\end{equation}

\subsection{LPM gives an unsigned-current bound (and hence mobility)}
\label{sec:unsigned-current}

Recall that (by definition), the aligned layer preserves $\range(P_1)$ and each pair of sectors $\range(P_{2c}+P_{2c+1})$ for $c\in\{1,\ldots,(t-1)/2\}$, while the staggered layer preserves $\range(P_{2c+1}+P_{2c+2})$ for $c\in\{0,\ldots,(t-3)/2\}$, as well as the endpoint sector $\range(P_t)$. However, in the definition of LPM, we referenced the $\mathcal B/\mathcal G$ decomposition --- so, to use this assumption effectively, we will have to somehow translate between the
$\mathcal{B}/\mathcal{G}$ picture and the nested defect picture. To this end, let us begin by asking the following question. If we have some vector $v \in \range(P_{y})$ for some $y \leq t-1$, and apply a local two-site brickwork unitary $u^{(m)}_{y,y+1}$, what does the output look like?

First notice that if we take a vector $v \in \range(P_{y})$, and look at the $\mathcal{S}_{\mathrm{diag}}$ weight across the sites $(y,y+1)$, its support must lie in either $\mathcal{B}$ or $\mathcal{S}_{\mathrm{diag}}$. This is because $v \in \range(P_y) \subset \range(Q_{y})$ implies that $y+1$ must be in $\mathrm{span}(\{P_{\pi}, \pi \in S_{m}\})$. At first glance, it is a little curious that $v$ could locally live in $\mathcal{S}_{\mathrm{diag}}$, since $Q_{y-1} \perp P_{y}$. It is worth clarifying this subtlety. Define $F_{y,y+1}$ as the projection onto $\mathcal{S}_{\mathrm{diag}}$ on the neighboring $y,y+1$ sites. Then $F_{y,y+1}v := v_{\mathrm{S}}$ is not necessarily zero, even though $v \perp \range(Q_{y-1}) = \range(Q_{y}) \cap \range(F_{y,y+1})$, since the projection onto the local $\mathcal{S}_{\mathrm{diag}}$ could take us out of $P_{y}$.

If we had $v_{S} = 0$, then we could use the LPM property, such that we know that the defect is in the $\mathcal{B}$ sector and that a fraction of it was going to move to the mobile $\mathcal{G}$ sector. Since the $\mathcal{G}$ sector means that $y+1$ no longer has a permutation, we would have moved the right defect edge one further. This suggests bounding the projection onto $F_{y,y+1}$. To do so we use the result from Ref.~\cite{brandao2012local}

\begin{fact}
Let $F_{i,i+1}$ be the orthogonal projector onto $\mathcal{S}_{\mathrm{diag}}$ on the neighboring $i,i+1$ sites, and define the nearest-neighbor Hamiltonian $H_{\ell,m}:=\sum_{i=1}^{\ell-1}(\mathbb I-F_{i,i+1}).$ Then there exists $ 0 < \Delta_m \leq 1$, independent of $\ell\geq2$, such that
\begin{equation}
    H_{\ell,m}\geq 
   \Delta_m(\mathbb I-\Pi_{\mathcal S_{[1,\ell]}}),
\end{equation}
\begin{equation}
    \Delta_m^{-1}\leq\mathrm{poly}(m).
\label{eq:haar-comparison-gap}
\end{equation}
\end{fact}
The ground space and gap follow from Lemmas~16--17 and Theorem~5 of Ref.~\cite{brandao2012local}. With this fact, we are ready to show that defects must move.

\begin{lemma}[Right-defect mobility]
\label{lem:finite-tail-mobility}
    Define $a_m:=\bigl[1-(1-\epsilon(m))^2\bigr]\Delta_m$, and $\kappa_m:=\sqrt{1-a_m}$. For $1\leq y<t$ and $x\in \range(P_y)$,
    \begin{equation}
    \|P_{y+1}u^{(m)}_{y,y+1}x\|^2\geq a_m\|x\|^2,
\end{equation}
\begin{equation}
    \|P_yu^{(m)}_{y,y+1}P_y\|\leq\kappa_m.
        \label{eq:finite-tail-mobility}
\end{equation}
\end{lemma}
\begin{proof}
    For now, let us forget about the sites to the left of $y$, allowing an arbitrary tensor factor there. For simplicity of notation, let us define $F=F_{y,y+1}$, $L=\mathbb I_y\otimes\Pi_{\mathcal P_{y+1}}$, and $u=u^{(m)}_{y,y+1}$, with identities on the remaining sites.

    Since $x\in\range(P_y)$, it is also in $\range(Q_{y})$, and so we have $Lx=x$, that is, the $y+1$ site must have a permutation. Now from the discussion above, $(\mathbb I-F)x$ lies in $\mathcal B$ on the first pair, tensored with the other sites. 

    Let us evaluate $\Vert (\mathbb{I}-L)u(x)\Vert^{2}$. We can write $\Vert (\mathbb{I}-L)u(x)\Vert^{2} = \Vert (\mathbb{I}-L)u(F(x) + (\mathbb{I} -F)(x))\Vert^{2}$. Now because $Lx = x$, $F(x) \in \mathcal{S}_{\mathrm{diag}}$, and since $u$ fixes $\mathcal{S}_{\mathrm{diag}}$, this term is annihilated by $(\mathbb{I}-L)$, so we don't need to care about it anymore. Note that $(\mathbb{I}-L)$ projects onto the $\mathcal{G}$ sector, since it picks where $y+1$ is not a permutation. Thus  
    \begin{multline}
        \Vert (\mathbb{I}-L)u(x)\Vert^{2} = \Vert (\mathbb{I}-L)u(\mathbb{I}-F)x\Vert^{2} =  \Vert \Pi_{\mathcal{G}}u(\mathbb{I}-F)x\Vert^{2} =\\
        \Vert u(\mathbb{I}-F)x\Vert^{2} - \Vert \Pi_{\mathcal B} u((\mathbb{I}-F)x)\Vert^{2} \geq (1-( 1-\epsilon(m))^{2})\Vert (\mathbb{I}-F)x\Vert^{2}.
    \end{multline}

    Here, we used LPM in the final inequality. Now we will work on bounding $\Vert (\mathbb{I}-F)x\Vert^{2}$. For this we use Equation~\eqref{eq:haar-comparison-gap}, with the Hamiltonian $H_{[y,t],m}:=\sum_{j=y}^{t-1}(\mathbb I-F_{j,j+1})$, with identities on the spectator sites. Clearly every term of this Hamiltonian annihilates $x$ except for $\mathbb{I}- F$. Thus we have 

    \begin{equation}
        \Vert (\mathbb{I}-F)x\Vert^{2} \geq \Delta_m\Vert (\mathbb{I}-Q_{y-1})x\Vert^{2} = \Delta_m\Vert x\Vert^{2}.
    \end{equation}

    Since $Q_y\preceq L$ and $ux\in\operatorname{range}(P_y+P_{y+1})$ by Lemma~\ref{lem:sector-locality},
    \begin{equation}
        \|P_{y+1}ux\|^2=\|(\mathbb I-Q_y)ux\|^2
        \geq\|(\mathbb I-L)ux\|^2.
    \end{equation}
    This yields the first part of Equation \eqref{eq:finite-tail-mobility}. The second follows from locality, and the orthogonality of $P_y$ and $P_{y+1}$, since any defect weight not in $P_{y+1}$ must stay in $P_{y}$.
\end{proof}

\begin{figure}[htbp]
\centering
\begin{tikzpicture}[
    scale=1.15, font=\small,
    tensor/.style={rectangle, draw=black, thick, rounded corners=3pt,
        fill=orange!80!white, minimum size=24pt, inner sep=0pt},
    leg/.style={thick, draw=black!80},
    perm/.style={circle, draw=red!80!black, thick, fill=white,
        inner sep=1.5pt, font=\scriptsize, text=red!80!black},
    inputdefect/.style={draw=red!70!black, thick, fill=red!12,
        rounded corners=4pt, minimum height=0.5cm},
    outputdefect/.style={draw=green!60!black, thick, fill=green!12,
        rounded corners=4pt, minimum height=0.5cm}
]
    % One brickwork layer; the middle gate straddles the right defect edge.
    \node[tensor] (leftgate)  at (-2.4,0) {$m$};
    \node[tensor] (edgegate)  at (0,0)    {$m$};
    \node[tensor] (rightgate) at (2.4,0)  {$m$};
    \node at (-4.0,0) {$\cdots$};
    \node at (4.0,0) {$\cdots$};

    % Diagonal wires, as in the manuscript's brickwork diagrams.
    \foreach \g/\c in {leftgate/-2.4,edgegate/0,rightgate/2.4}{
        \draw[leg] (\g.south west) -- (\c-0.8,-0.9);
        \draw[leg] (\g.south east) -- (\c+0.8,-0.9);
        \draw[leg] (\g.north west) -- (\c-0.8,0.9);
        \draw[leg] (\g.north east) -- (\c+0.8,0.9);
    }

    % Input: arbitrary left part, with a matched permutation tail.
    \node[inputdefect, minimum width=3.1cm] at (-2.05,-1.15)
        {Defect};
    \foreach \q in {0.8,1.6,3.2}{
        \node[perm] at (\q,-0.9) {$\pi$};
    }
    \node at (-0.8,-1.65) {$y$};
    \node at (0.8,-1.65) {$y+1$};
    \node[font=\scriptsize] at (2.4,-1.65) {matched tail};
    \node[anchor=east, text=red!70!black] at (-3.8,-1.15)
        {$x\in\range(P_y)$};

    % Display the projected component whose right edge has moved to y+1.
    \node[outputdefect, minimum width=4.7cm] at (-1.25,1.15)
        {Defect};
    \foreach \q in {1.6,3.2}{
        \node[perm] at (\q,0.9) {$\pi$};
    }
    \node[anchor=east, text=green!50!black] at (-3.8,1.15)
        {$P_{y+1}ux$};
    \node at (-0.8,1.65) {$y$};
    \node at (0.8,1.65) {$y+1$};

    % The orange middle gate is the local gate in the lemma.
    \node[font=\scriptsize, anchor=west] at (0.45,0)
        {$u^{(m)}_{y,y+1}$};
    \draw[thick, -{Stealth[length=2mm]}, green!50!black]
        (-0.8,2.05) -- (0.8,2.05);
    \node[above, text=green!50!black] at (0,2.05) {right-edge transfer};

    \node[align=center, text=red!70!black] at (-0.65,-2.55)
        {Hamiltonian gap: a detectable $\mathcal B$ component\\[4pt]
         $\|(\mathbb I-F)x\|^2\geq\Delta_m\|x\|^2$};
    \node[align=center, text=green!50!black] at (-0.65,-3.8)
        {LPM: transfer from $\mathcal B$ into $\mathcal G$\\[4pt]
         $\|\Pi_{\mathcal G}ux\|^2\geq
         [1-(1-\epsilon(m))^2]\|(\mathbb I-F)x\|^2$};
    \node[draw=black!50, rounded corners=3pt, inner sep=7pt]
        at (-0.65,-5.05)
        {$\|P_{y+1}ux\|^2\geq
        \underbrace{[1-(1-\epsilon(m))^2]\Delta_m}_{a_m}\|x\|^2$};
\end{tikzpicture}
\caption{We summarize the contents of \autoref{lem:finite-tail-mobility} here. The Hamiltonian gap result guarantees  $\mathcal{B}$ weight, which by LPM becomes $\mathcal{G}$ weight, which then flows across to the next defect site.}
\label{fig:right-defect-mobility}
\end{figure}

Having used the unitary walk $\mathbb{W}$ to define the current and establish bulk mobility, let us now start analyzing the walk with a hard absorbing boundary $\mathbb{K}_{t}$. To this end, we split the staggered layer $\mathbb{S}$ into sector-preserving and sector-changing parts. Given a right defect edge $y \in [1,t]$, the part of $\mathbb{S}$ which preserves $y$ is given by $P_{y}\mathbb{S}P_{y}$. In words, we restrict the domain to vectors in $\range(P_y)$, evolve it by $\mathbb{S}$, and then again restrict the output to $\range(P_y)$. We define the sector-preserving $S_{0}$ as the sum of this quantity over $y$, that is,  

\begin{equation}
    S_0:=\sum_{y=1}^tP_y\mathbb S P_y.
    \label{eq:finite-mobile-stalled-operators}
\end{equation}
The part of the staggered layer that doesn't stall the defect is then given by 
\begin{equation}
    T:=\mathbb S-S_0.
\end{equation}

Recall that $Q$ is the projector onto the set of defects whose right edge is before $t$ (or equivalently before $A_{R}$). Then we can define another operator, $Z$, which retains the part of a vector which is not absorbed and does not change its single-site sector during the staggered layer.

\begin{equation}
     Z:=S_0 Q\mathbb A.
\end{equation}

On the vectors which survive contraction, we have $\mathbb{K}_{t} = Z +TQ\mathbb{A}$. To see this, consider $\mathbb{K}_{t}v = Q\mathbb{W}Qv = Q\mathbb{S}\mathbb{A}v$, using $Qv = v$ (surviving vector). Further, $\mathbb{S}$ does not act on the $t$ site, so $[\mathbb{S}, P_{t}] = [\mathbb{S}, Q] = 0$, and we get $\mathbb{S}Q\mathbb{A}v = (S_{0}+T)Q\mathbb{A}v = (Z+TQ\mathbb{A})v$. This second term will turn out to lower bound the unsigned current in the next Lemma. Physically, since $Z$ retains the part of a vector which does not change its sector, $TQ\mathbb{A}$ retains vectors which come out $\mathbb{K}_{t}$ with their defect edge moved.

\begin{lemma}[Current controls sector transfer]
\label{lem:current-sector-transfer}
Retain the setting from before. We then have $\|J\|\leq1$, $\|T\|\leq1$, and for every $w\in\range(\ Q)$,
\begin{equation}
    \|TQ\mathbb A w\|^2\leq\langle w,|J|w\rangle.
    \label{eq:current-sector-transfer}
\end{equation}
\end{lemma}
\begin{proof}
Since the aligned layer preserves $X$ (recall that we had defined current as being over two cells in the aligned pair notation), we can immediately write
\begin{equation}
    J=\mathbb A^\dagger J_{\mathbb S}\mathbb A,
\end{equation}
\begin{equation}
    J_{\mathbb S}:=\mathbb S^\dagger X\mathbb S-X.
\end{equation}
Fix a staggered block at location $(2c+1, 2c+2)$ for some $c\in\{0,\ldots,(t-3)/2\}$, and consider the set of operators with a right defect edge here, that is $\range(P_{2c+1}+P_{2c+2})$. For simplicity of notation, define $s$ to be the restriction of $\mathbb{S}$ to this block, and let $\Pi$ be $P_{2c+2}|_{\range(P_{2c+1}+P_{2c+2})}$.

    The two positions differ by one, since the two sites in a staggered block have different positions, so $X |_{\range(P_{2c+1}+P_{2c+2})}= c\mathbb{I} + \Pi$. This is especially clear if we write it out in matrix form --- $X$ has value $\lfloor (2c+1)/2\rfloor$ at $2c+1$, whereas it has value $\lfloor (2c+2)/2\rfloor$, so these are the diagonal entries. The $\Pi$ compensates for the extra $1$. As a result, a simple substitution yields
    
    \begin{equation}
        J_{\mathbb S}=s^\dagger\Pi s-\Pi
    \end{equation}
    on this block. Put $R=s^\dagger\Pi s$. This is clearly an orthogonal projector, since $R^{\dag} = R$, and $R^2 = R$. For $x=x_1+x_2$, $x_1 \in \range(P_{2c+1})$, and $x_{2} \in \range(P_{2c+2})$,
    \begin{equation}
    \begin{aligned}
        \Vert J_{\mathbb{S}}x \Vert^2 =\Vert(R-\Pi)x\Vert^2 &= \Vert(R-\Pi)(x_1 +x_2)\Vert^2\\
        &= \Vert Rx_1 +Rx_2 -x_2 \Vert^{2} = \|Rx_1\|^2+\|(\mathbb I-R)x_2\|^2\\
        &=\|\Pi sx_1\|^2+\|(\mathbb I-\Pi)sx_2\|^2
         =\|Tx\|^2.
    \end{aligned}
    \end{equation}
    The expansion uses $\Pi P_{2c+1}=0$, Pythagoras' theorem on the orthogonal ranges of $R$ and $\mathbb{I}-R$, and the definition of $R$. Thus $J_{\mathbb S}^2=T^\dagger T$ on each staggered block. Both operators vanish on the endpoint sector $\range(P_t)$, so the identity holds on $\mathcal K$. A difference of orthogonal projectors has norm at most one, so $|J_{\mathbb S}| \geq J_{\mathbb S}^2$ and $\|T\|\leq1$.
    
    Now $TP=0$ (the staggered layer doesn't touch the endpoint), which implies $T\mathbb A w=TQ\mathbb A w$. But then since $T^{\dag}T = J_{\mathbb{S}}^{2}$, we have $\Vert T\mathbb{A}w\Vert^2 = \langle \mathbb{A}w, T^{\dag} T\mathbb{A}w\rangle = \langle \mathbb{A}w, J_{\mathbb{S}}^{2}\mathbb{A}w\rangle \leq \langle \mathbb{A}w, \vert J_{\mathbb{S}} \vert \mathbb{A}w\rangle = \langle w, |J|w\rangle$. This completes our proof.
\end{proof}

In \autoref{lem:finite-tail-mobility}, we showed that a fraction of every defect microscopically moves right under the action of a brickwork unitary, while in \autoref{lem:current-sector-transfer}, we connect the motion of this defect right edge with the unsigned current. In the next Lemma we will show that the fraction that stalls cannot do so forever; in particular, the stalled branch must contract over a two-step period.

\begin{lemma}[Two-step contraction of the stalled part]
\label{lem:stalled-two-step}
    On $\range(Q)$,
    \begin{equation}
    \|Z^2\|^2\leq1-\delta_m,
    \label{eq:stalled-two-step}
\end{equation}
\begin{equation}
    \delta_m:=\frac{a_m^2}{5}.
\end{equation}
\end{lemma}
\begin{proof}
It is useful to describe the idea of this proof in a figure before actually writing it out fully. 

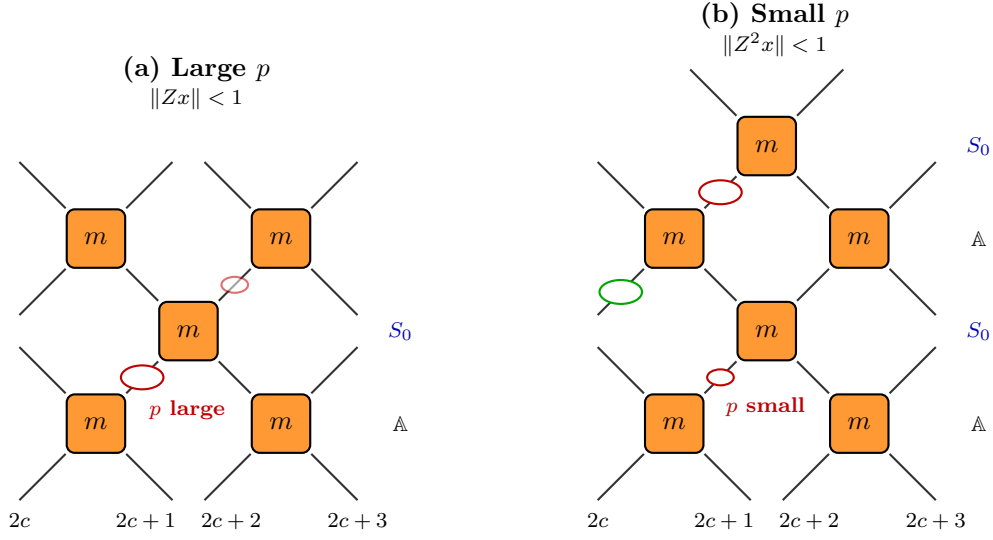
\begin{figure}[h!]
    \centering
    \begin{tikzpicture}[
    baseline=(current bounding box.center),
    scale=1.02,
    tensor/.style={
        rectangle, draw=black, thick, rounded corners=3pt,
        fill=orange!80!white, minimum size=22pt,
        font=\small\bfseries
    },
    defect/.style={
        rectangle, draw=red!80!black, thick, fill=white,
        rounded corners=3pt, inner sep=4pt,
        font=\small\bfseries, text=red!80!black
    },
    redcirc/.style={
        circle, draw=red!80!black, thick, fill=white,
        inner sep=1.2pt, font=\scriptsize,
        text=red!80!black
    },
    leg/.style={thick, draw=black!80},
    capsuleB/.style={
        draw=red!60!pink, thick, fill=red!20,
        fill opacity=0.5, rounded corners=6pt
    },
    capsuleY/.style={
        draw=green!60!black, thick, fill=green!20,
        fill opacity=0.5, rounded corners=6pt
    },
    % bubbles centered directly on circuit legs
    rbubble/.style={
        ellipse, draw=red!75!black, thick,
        fill=white, minimum width=16pt,
        minimum height=9pt, inner sep=0pt
    },
    rbubblesmall/.style={
        ellipse, draw=red!75!black, thick,
        fill=white, minimum width=10pt,
        minimum height=6pt, inner sep=0pt
    },
    gbubble/.style={
        ellipse, draw=green!65!black, thick,
        fill=white, minimum width=16pt,
        minimum height=9pt, inner sep=0pt
    },
    lbl/.style={font=\scriptsize}
]

\begin{scope}[xshift=-3.75cm]

    \node[font=\small\bfseries] at (-0.5,4.00)
        {(a) Large $p$};

    \node[font=\scriptsize] at (-0.5,3.62)
        {$\|Zx\|<1$};

    \node[tensor] (LA1) at (-1.8,-0.6) {$m$};
    \node[tensor] (LA2) at ( 0.6,-0.6) {$m$};

    % bottom legs = neighboring physical sites
    \draw[leg]
        (LA1.south west) -- ++(-0.6,-0.6)
        coordinate[pos=1] (Lsite0);

    \draw[leg]
        (LA1.south east) -- ++(0.6,-0.6)
        coordinate[pos=1] (Lsite1);

    \draw[leg]
        (LA2.south west) -- ++(-0.6,-0.6)
        coordinate[pos=1] (Lsite2);

    \draw[leg]
        (LA2.south east) -- ++(0.6,-0.6)
        coordinate[pos=1] (Lsite3);

    % site labels
    \node[lbl, below=1pt] at (Lsite0) {$2c$};
    \node[lbl, xshift = -10pt, below=1pt] at (Lsite1) {$2c+1$};
    \node[lbl, xshift = 10pt, below=1pt] at (Lsite2) {$2c+2$};
    \node[lbl, below=1pt] at (Lsite3) {$2c+3$};

    \node[tensor] (LS1) at (-0.6,0.6) {$m$};

    % outer left leg
    \draw[leg] (LA1.north west) -- ++(-0.6,0.6);

    % P_{2c+1} leg -- LARGE p
    \draw[leg]
        (LA1.north east)
        --
        coordinate[pos=0.50] (LpLarge)
        (LS1.south west);

    % P_{2c+2} incoming leg
    \draw[leg]
        (LA2.north west)
        --
        (LS1.south east);

    % outer right leg
    \draw[leg] (LA2.north east) -- ++(0.6,0.6);

    % large p bubble exactly on leg
    \node[rbubble] at (LpLarge) {};
    \node[
        text=red!75!black,
        font=\scriptsize\bfseries,
        xshift=17pt, yshift=-12pt
    ] at (LpLarge) {$p$ large};

    \node[tensor] (LA3) at (-1.8,1.8) {$m$};
    \node[tensor] (LA4) at ( 0.6,1.8) {$m$};

    % P_{2c+1} stalled output
    \draw[leg]
        (LS1.north west)
        --
        (LA3.south east);

    % P_{2c+2} output -- small surviving component
    \draw[leg]
        (LS1.north east)
        --
        coordinate[pos=0.50] (Lsmall)
        (LA4.south west);

    % small bubble exactly on P_{2c+2} leg
    \node[rbubblesmall, opacity=0.55] at (Lsmall) {};

    % other incoming legs to aligned gates
    \draw[leg] (LA3.south west) -- ++(-0.6,-0.6);
    \draw[leg] (LA4.south east) -- ++( 0.6,-0.6);

    % outgoing legs
    \draw[leg] (LA3.north west) -- ++(-0.6,0.6);
    \draw[leg] (LA3.north east) -- ++( 0.6,0.6);
    \draw[leg] (LA4.north west) -- ++(-0.6,0.6);
    \draw[leg] (LA4.north east) -- ++( 0.6,0.6);

    % layer labels
    \node[lbl] at (2.15,-0.6) {$\mathbb A$};
    \node[lbl, text=blue!70!black] at (2.15,0.6) {$S_0$};

\end{scope}

\begin{scope}[xshift=3.75cm]

    \node[font=\small\bfseries] at (-0.5,4.70)
        {(b) Small $p$};

    \node[font=\scriptsize] at (-0.5,4.32)
        {$\|Z^2x\|<1$};

    \node[tensor] (RA1) at (-1.8,-0.6) {$m$};
    \node[tensor] (RA2) at ( 0.6,-0.6) {$m$};

    % bottom legs
    \draw[leg]
        (RA1.south west) -- ++(-0.6,-0.6)
        coordinate[pos=1] (Rsite0);

    \draw[leg]
        (RA1.south east) -- ++(0.6,-0.6)
        coordinate[pos=1] (Rsite1);

    \draw[leg]
        (RA2.south west) -- ++(-0.6,-0.6)
        coordinate[pos=1] (Rsite2);

    \draw[leg]
        (RA2.south east) -- ++(0.6,-0.6)
        coordinate[pos=1] (Rsite3);

    % site labels
    \node[lbl, below=1pt] at (Rsite0) {$2c$};
    \node[lbl, xshift = -10pt, below=1pt] at (Rsite1) {$2c+1$};
    \node[lbl, xshift=10pt, below=1pt] at (Rsite2) {$2c+2$};
    \node[lbl, below=1pt] at (Rsite3) {$2c+3$};

    \node[tensor] (RS1) at (-0.6,0.6) {$m$};

    % outer left leg
    \draw[leg] (RA1.north west) -- ++(-0.6,0.6);

    % small p on P_{2c+1} leg
    \draw[leg]
        (RA1.north east)
        --
        coordinate[pos=0.50] (RpSmall)
        (RS1.south west);

    \draw[leg]
        (RA2.north west)
        --
        (RS1.south east);

    % outer right leg
    \draw[leg] (RA2.north east) -- ++(0.6,0.6);

    % small-p bubble centered
    \node[rbubblesmall] at (RpSmall) {};

    \node[
        text=red!75!black,
        font=\scriptsize\bfseries,
        xshift=17pt, yshift=-11pt
    ] at (RpSmall) {$p$ small};

    \node[tensor] (RA3) at (-1.8,1.8) {$m$};
    \node[tensor] (RA4) at ( 0.6,1.8) {$m$};

    % central staggered outputs
    \draw[leg]
        (RS1.north west)
        --
        (RA3.south east);

    \draw[leg]
        (RS1.north east)
        --
        (RA4.south west);

    % incoming P_{2c} leg for next A:
    % mostly-even surviving component lives here
    \draw[leg]
        (RA3.south west)
        --
        coordinate[pos=0.50] (Rgreen)
        ++(-0.6,-0.6);

    % green bubble centered exactly on the leg
    \node[gbubble] at (Rgreen) {};

    % other outer leg
    \draw[leg] (RA4.south east) -- ++(0.6,-0.6);

    \node[tensor] (RS2) at (-0.6,3.0) {$m$};

    % outer left
    \draw[leg] (RA3.north west) -- ++(-0.6,0.6);

    % Lemma 5 re-creates odd-sector weight here
    \draw[leg]
        (RA3.north east)
        --
        coordinate[pos=0.50] (Rrecreated)
        (RS2.south west);

    % other input to staggered gate
    \draw[leg]
        (RA4.north west)
        --
        (RS2.south east);

    % outer right
    \draw[leg] (RA4.north east) -- ++(0.6,0.6);

    % recreated odd weight centered on leg
    \node[rbubble] at (Rrecreated) {};

    % outgoing legs from second staggered gate
    \draw[leg] (RS2.north west) -- ++(-0.6,0.6);
    \draw[leg] (RS2.north east) -- ++( 0.6,0.6);

    % layer labels
    \node[lbl] at (2.15,-0.6) {$\mathbb A$};
    \node[lbl, text=blue!70!black] at (2.15,0.6) {$S_0$};
    \node[lbl] at (2.15,1.8) {$\mathbb A$};
    \node[lbl, text=blue!70!black] at (2.15,3.0) {$S_0$};

\end{scope}

\end{tikzpicture}
\caption{Take a vector initially in $\range(P_{2c}+P_{2c+1})$, that is, restrict to just one aligned cell. Note that $Z=S_0Q\mathbb A=S_0\mathbb A$ on a full surviving cell. After evolution by $\mathbb{A}$, call the amount of defect that lands in $P_{2c+1}$ as $p$. There are two cases now. In the first case, shown in (a), $p$ is large, wherein the staggered layer (by \autoref{lem:finite-tail-mobility}) moves at least some of the weight to $2c+2$. If $p$ is small, as in (b), a large fraction of the norm must have stayed at $2c$. Then, by \autoref{lem:finite-tail-mobility}, a substantial amount must move to $2c+1$, where we can apply the same argument as in (a). Further, existing right defect weight at $2c+1$ cannot destructively interfere with all of it, since $p$ was small in the first place. That is, the branch which does not get absorbed and stalls must contract over two periods.}
\end{figure}

Let us now write this out. Take $x \in \range(P_{2c}+P_{2c+1})$ with $\|x\|=1$ and $1\leq c<(t-1)/2$, write $(\mathbb A x)_j:=P_j\mathbb A x$, and set $\Vert (\mathbb A x)_{2c+1}\Vert^{2}=p$. Then 
\begin{equation}
    \Vert(\mathbb{A}x)_{2c}\Vert = \sqrt{1 - p}.
\end{equation}
Let us apply the sector-preserving staggered layer. Set $S_{0}\mathbb{A}x := y$, $y_j:=P_jy$, and $\kappa:=\kappa_m$ (for simplicity of notation). We have (trivially) $\|y_{2c}\|\leq\sqrt{1-p}$ and (by \autoref{lem:finite-tail-mobility}) $\|y_{2c+1}\|\leq\kappa\sqrt p$. Then 
\begin{equation}
\label{eq:normloss}
    \Vert y \Vert^{2}  = \Vert y_{2c} \Vert^2 + \Vert y_{2c+1}\Vert^{2} \leq 1-p + \kappa^2 p = 1-a_{m}p.
\end{equation}
Note that we didn't have to worry about $y_{2c+2}$, since we are only considering $Z$, which is comprised of $S_0$. This is the norm that is lost to $Z$, and it is the source of the `contraction'.

Let us get back to the analysis. If $p \geq a_{m}/5$, we are done, so we must now treat the case where $p$ is small. Set $\Vert y_{2c} \Vert =s$. From the figure, it is clear what we want --- large weight in $2c+1$ after $\mathbb{A}$. Now, the norm of the $2c+1$ component after $\mathbb{A}$ is $\Vert P_{2c+1}\mathbb{A}y\Vert$. Again, by \autoref{lem:finite-tail-mobility}, we have $\Vert P_{2c+1}\mathbb{A}y_{2c}\Vert \geq \sqrt{a_{m}} s$, and trivially, $\Vert P_{2c+1}\mathbb{A}y_{2c+1}\Vert \leq \kappa \sqrt{p}$, so
\begin{equation}
    \Vert P_{2c+1}\mathbb{A}y\Vert \geq (\sqrt{a_m}s-\kappa\sqrt{p})_+,
\end{equation}
\begin{equation}
    q_+:=\max\{q,0\}.
\end{equation}
Of this, again by the same argument as Equation \eqref{eq:normloss}, we get that an $a_{m}$ fraction of this squared norm will be lost in the staggered layer. For brevity, we do not repeat this argument. Altogether we get 
\begin{equation}
    \Vert Z^{2}x\Vert^{2} \leq 1- a_{m}(p+(\sqrt{a_{m}}s - \kappa \sqrt{p})_+^{2}).
\end{equation}
Now, we just need a uniform lower bound on this. Choose $\theta = \frac{a_m}{(1+\kappa)^2+a_m}$. Then if $p+(\sqrt{a_{m}}s - \kappa \sqrt{p})_+^{2} \geq \theta$, $\Vert Z^{2}x\Vert^{2} \leq 1- a_{m}\theta \leq 1- a_{m}^{2}/5$, since the denominator is at most $5$. If $p+(\sqrt{a_{m}}s - \kappa \sqrt{p})_+^{2} < \theta$, one can check with some algebra that 
\begin{equation}
    \Vert y\Vert^{2} = \Vert y_{2c}\Vert^2 + \Vert y_{2c+1}\Vert^2 < \frac{\theta(1+\kappa)^2}{a_m} + \kappa^2 \theta = 1-a_m\theta,
\end{equation}
which by our choice of $\theta$ and the contractivity of $Z$ yields the result on a full surviving cell. The other gates, including the unitary $A_L$, preserve the relevant sectors and do not change the bounds.

For the singleton sector $\operatorname{range}P_1$ when $t\geq3$, $\mathbb A$ preserves the sector and its norm. Lemma~\ref{lem:finite-tail-mobility} applied to the staggered gate $u^{(m)}_{1,2}$ then gives
\begin{equation}
    \|ZP_1\|\leq\kappa_m,
\end{equation}
\begin{equation}
    \|Z^2P_1\|^2\leq\kappa_m^4
    =(1-a_m)^2\leq1-a_m^2/5.
\end{equation}
In the last surviving sector $\operatorname{range}P_{t-1}$, $Q\mathbb A$ has norm at most $\kappa_m$, since the forward component is absorbed. Thus $\|Z^2|_{\operatorname{range}P_{t-1}}\|^2\leq(1-a_m)^2\leq1-a_m^2/5$. Since $Z$ preserves the full surviving cells and these two endpoint sectors, taking their maximum completes the proof. For $t=1$, $Q=0$ and the claim is vacuous.

\end{proof}

Now, what do we get from this? If we take an arbitrary surviving defect $w$, and neither of its first two steps carried much unsigned current, then both steps must be close to the stalled evolution. But the stalled part contracts in two steps! Thus, either some norm has been absorbed or the unsigned current must have been appreciable. 

\begin{proposition}[Mobility or absorption]
\label{prop:unsigned-current}
    Retain the setting from above, and let $u$ be $\epsilon(m)$-LPM. Then, for every $w\in\range(Q)$,
    \begin{equation}
        \langle w,|J|w\rangle
        +\langle\mathbb K_t w,|J|\mathbb K_t w\rangle +\|P\mathbb A w\|^2
        +\|P\mathbb A\mathbb K_t w\|^2
        \geq c(m)\|w\|^2,
    \label{eq:mobility-or-absorption}
    \end{equation}
    where $c(m):=\frac{a_m^4}{2000}$. Under \autoref{conj:lpm}, $c(m)^{-1}\leq\mathrm{poly}(m)$.
\end{proposition}
\begin{proof}

The proof idea is straightforward. Compare the process $\mathbb{K}_{t}^{2}$ with $Z^{2}$, where we recall that $Z$ is $\mathbb{K}_{t}$ excluding the portion which moves the right-edge of the defect it acts on. Now if there is a lot of mobility, then we expect the current terms to be large in Equation \eqref{eq:mobility-or-absorption}. If there is little mobility, then $\mathbb{K}_{t}^{2}$ and $Z^{2}$ should be very similar by construction. But we know $Z^{2}$ contracts, so then $\mathbb{K}_{t}$ must as well, and the only way for $\mathbb{K}_t$ to contract is via absorption! Hence in this case the absorption terms in Equation \eqref{eq:mobility-or-absorption} will be large. Thus, we will be able to lower bound the left-hand side.

All that remains is to do the mathematics. Normalize $\|w\|=1$, and let us define $B := TQ\mathbb A$, $x_1:=\|B w\|$, $x_2:=\|B \mathbb K_t w\|$, $d_0 :=\|P\mathbb A w\|^2$, $d_1:=\|P\mathbb A\mathbb K_t w\|^2$. Each of these quantities has a physical interpretation: $x_{1}$ gives the norm of the mobile component of $w$, $x_{2}$ the same after one application of $\mathbb{K}_{t}$, $d_0$ gives the squared norm absorbed in the first step from $w$ (there could be an $\mathbb{S}$ added in, but this commutes with $P$ anyway), and $d_1$ the squared norm absorbed in the second step. $B$ is simply $\mathbb{K}_t-Z$ on $\range(Q)$ (see the discussion above \autoref{lem:current-sector-transfer}.) Now 
\begin{equation}
    \mathbb{K}_t^{2} - Z^2 = \mathbb{K}_t(\mathbb{K}_t- Z) + (\mathbb{K}_t-Z)Z = \mathbb{K}_tB +BZ
\end{equation}
 
and  $\Vert K_{t}Bw \Vert \leq \Vert Bw \Vert =x_{1}$, while $\Vert BZw \Vert \leq \Vert B^{2}w\Vert + \Vert B\mathbb{K}_{t}w\Vert \leq x_1 +x_2$ by the triangle inequality. Altogether,
\begin{equation}
        \|\mathbb K_t^2w-Z^2w\|\leq2x_1+x_2.
\end{equation}
Unitarity and Equation~\eqref{eq:finite-layer-commutations} give $\|\mathbb K_t^2w\|^2=1-d_0-d_1$. This just expresses that the only way for $K_{t}$ to lose norm is via absorption. Now, if $d_0+d_1\geq\delta_m/2$, absorption alone proves the claim, since these are two of the terms on the left-hand side.

Otherwise, \autoref{lem:stalled-two-step} and the reverse triangle inequality (since $\Vert \mathbb{K}_{t}^2 w\Vert^{2} \geq 1-\delta_{m}/2$, and $\Vert Z^2 w\Vert^{2} \leq 1-\delta_{m}$) give
    \begin{equation}
        2x_1+x_2
        \geq\sqrt{1-\delta_m/2}-\sqrt{1-\delta_m}
        \geq\frac{\delta_m}{4}.
    \end{equation}

Then by Cauchy--Schwarz, $x_1^2+x_2^2\geq\delta_m^2/80=c(m)$. But \autoref{lem:current-sector-transfer} is exactly built to provide bounds on $x_1$ and $x_2$, and so applying it to $w$ and $\mathbb K_t w$ completes our proof.
\end{proof}

\subsection{First-return recurrence bounds negative current}
\label{sec:negative-current}

We now use our second conjecture to bound the flow of negative current in our setup. We recall quickly that the absorbing process is generated by $\mathbb K_t$, and that its backward-current subspace is $\mathcal U_m$. We first show the following identity.

\begin{lemma}[Renewal identity for the absorbing process]
\label{lem:finite-renewal}
    For every $v\in\mathcal K$ and $|z|<1$,
    \begin{equation}
        \sum_{n\geq0}z^n
       \Pi_{\mathcal U_m}\mathbb K_t^n v
           =(\mathbb I_{\mathcal U_m}-\widehat a_m(z))^{-1}
             \widehat A_v(z).
        \label{eq:finite-renewal}
    \end{equation}
\end{lemma}

\begin{proof}
Write $\Pi=\Pi_{\mathcal U_m}$ and
$C=\mathbb I_{\mathcal K}-\Pi$. For $n\geq1$, take $\Pi \mathbb{K}_{t}^{n}v$ and repeatedly insert $\Pi + C$ in between the various applications of $\mathbb{K}$. This yields
\begin{equation}
    \Pi\mathbb K_t^n v
    =\Pi\mathbb K_t(C\mathbb K_t)^{n-1}Cv
     +\sum_{j=1}^{n}
      a_{m,j}\Pi\mathbb K_t^{n-j}v,
\end{equation}
which can be checked by induction. Physically, the first term describes the case where the first entrance from outside of $\mathcal U_m$ happens at time step $n$, and each term in the sum describes an earlier visit to $\mathcal U_m$ at time $n-j$, followed by another return after $j$ steps. Now, let us multiply by $z^{n}$ and sum over $n$. If we write $\widehat M_v(z):= \sum_{n\geq0}z^n\Pi\mathbb K_t^n v,$ then we get 
\begin{equation}
    \widehat M_v(z)
    =\widehat A_v(z)+\widehat a_m(z)\widehat M_v(z).
\end{equation}
We do not have to worry about convergence, since $|z|<1$, and $\mathbb K_t$ is a contraction.

We only need to establish that
$\mathbb I_{\mathcal U_m}-\widehat a_m(z)$ is invertible. For this, restrict to $\mathcal{U}_{m}$ and notice that $\widehat A_v(z)=v$ for all $v \in \mathcal{U}_{m}$. Then the above equation reduces to 
\begin{equation}
    (\mathbb{I}_{\mathcal{U}_{m}} - \widehat a_m(z))\widehat M_v(z)
    =v,
\end{equation}
which means that $(\mathbb{I}_{\mathcal{U}_{m}} - \widehat a_m(z))$ must be invertible, since it is surjective and maps the finite-dimensional space $\mathcal{U}_{m}$ to itself. This completes the proof.
\end{proof}

\begin{proposition}[Accumulated negative-current bound]
\label{prop:finite-negative-current}
Assume \autoref{conj:first-return-gap}. Then for every $v\in\range(Q)$ and every integer $H\geq1$,
\begin{equation}
    \sum_{n=0}^{H-1}
      \bra{\mathbb K_t^n v}J_-\ket{\mathbb K_t^n v} \leq\beta_m(t)
      \sum_{n=0}^{H-1}
      \bra{\mathbb K_t^n v}|J|\ket{\mathbb K_t^n v}
      +b_m(t)\|v\|^2.
\label{eq:finite-backflow-bound}
\end{equation}
\end{proposition}
\begin{proof}
First, we use \autoref{lem:finite-renewal} to expand out $J_-^{1/2}(\mathbb I_{\mathcal U_m}-\widehat a_m(z))^{-1}\widehat A_v(z)$ as $\sum_{n\geq0}z^nJ_-^{1/2}\mathbb K_t^n v$, where we have used $J_-^{1/2}\Pi_{\mathcal U_m}=J_-^{1/2}$. Parseval's identity in Equation~\eqref{eq:finite-horizon-hardy} then gives
\begin{equation}
    \left\|J_-^{1/2}
      (\mathbb I_{\mathcal U_m}-\widehat a_m(z))^{-1}
      \widehat A_v(z)\right\|_{H^2;H}^{\,2} =\sum_{n=0}^{H-1}
      \|J_-^{1/2}\mathbb K_t^n v\|^2
      =\sum_{n=0}^{H-1}
      \bra{\mathbb K_t^n v}J_-\ket{\mathbb K_t^n v}.
\label{eq:hardy-negative-current}
\end{equation}
We can similarly expand $(\mathbb I_{\mathcal K}-z\mathbb K_t)^{-1}v$
and apply the same identity to get
\begin{equation}
    \left\||J|^{1/2}
    (\mathbb I_{\mathcal K}-z\mathbb K_t)^{-1}v
    \right\|_{H^2;H}^{\,2}
    =\sum_{n=0}^{H-1}
      \bra{\mathbb K_t^n v}|J|\ket{\mathbb K_t^n v}.
    \label{eq:hardy-unsigned-current}
\end{equation}
Combining these results immediately yields, by  Equation~\eqref{eq:finite-schur-backflow}, our result.
\end{proof}

This justifies in retrospect the weighting of the response energy with the current, which seemed a little ad hoc back then, but now is very useful in giving us the form that we need to bound $J_{-}$.

\subsection{Positive accumulated current and surviving defects}
\label{sec:positive-current}

Having placed bounds on $|J|$ and $J_-$, we now use the simple identity $J=|J|-2J_-$ to show that if not much norm of a vector is absorbed over a block, then after accounting for the initial transient, we must have positive current over that block.

\begin{proposition}
\label{prop:positive-current}
Assume \autoref{conj:lpm} and \autoref{conj:first-return-gap}, and take $v\in\range(Q)$. Let us also define $w_n :=\mathbb K_t^n v$ and $L_H=\|v\|^2-\|w_H\|^2$. Then, for every integer $H\geq2$,
\begin{equation}
    \sum_{n=0}^{H-1}\langle w_n,Jw_n\rangle
    \geq\nu_m(t)\sum_{n=0}^{H-2}\|w_n\|^2
      -(1-2\beta_m(t))L_H-2b_m(t)\|v\|^2,
\label{eq:finite-positive-current}
\end{equation}
with $\nu_m(t)=\frac{(1-2\beta_m(t))c(m)}{2}$. Further, under our previous assumptions for $\beta_{m}$ and $c(m)$, $\nu_m(t)^{-1}\leq\mathrm{poly}(m,t)$.
\end{proposition}
\begin{proof}
Consider an arbitrary $v\in\range(Q)$. Now use $J=|J|-2J_-$. \autoref{prop:finite-negative-current} gives immediately 
\begin{equation}
    \sum_{n=0}^{H-1}\langle w_n,Jw_n\rangle
    \geq(1-2\beta_m(t))
         \sum_{n=0}^{H-1}\langle w_n,|J|w_n\rangle
         -2b_m(t)\|v\|^2.
    \label{eq:current-from-backflow}
\end{equation}
Now $\mathbb K_t$ maps $\range(Q)$ to itself, so we can apply \autoref{prop:unsigned-current} to each $w_n$. Let us work this out in some detail. We have $\langle w_{n},|J|w_n\rangle +\langle w_{n+1},|J| w_{n+1}\rangle +\|P\mathbb A w_n\|^2 +\|P\mathbb A w_{n+1}\|^2 \geq c(m)\|w_{n}\|^2$, so summing this over $n$ from $0$ to $H-2$ gives two of every $\|P\mathbb A w_n\|^2$ and $\langle w_{n},|J|w_n\rangle$ except the first and the last, so that
\begin{equation}
    c(m)\sum_{n=0}^{H-2}\|w_n\|^2
    \leq2\sum_{n=0}^{H-1}\langle w_n,|J|w_n\rangle+2\sum_{n=0}^{H-1} \|P\mathbb A w_n\|^2.
\end{equation}

Now $\|P\mathbb A w_n\|^2$ is the squared norm lost by $w_n$ after one application of $\mathbb{K}_{t}$, so $\|P\mathbb A w_n\|^2 = \Vert w_{n}\Vert^2 - \Vert w_{n+1}\Vert^2$. This sum then telescopes, and we get 
\begin{equation}
    c(m)\sum_{n=0}^{H-2}\|w_n\|^2
    \leq2\sum_{n=0}^{H-1}\langle w_n,|J|w_n\rangle+2L_{H}.
\end{equation}
Substituting this back into Equation~\eqref{eq:current-from-backflow} then finishes our proof.
\end{proof}

\section{Contraction of the hard absorbing process}
\label{sec:hard-absorber}

In \autoref{sec:free-walk}, we related the two dynamical assumptions that we made to the flow of current in the finite chain. In particular, we established that the flow of current over a block is lower bounded by the total cumulative amount of norm that is left unabsorbed over the block, minus the amount of norm that is absorbed (with an allowance for an initial transient). This morally tells us that current is lower bounded by something related to how much defect weight persists over the block evolution. Now consider a defect that survives in $\mathbb{K}_{t}$ for a long time, so that the unabsorbed norm dominates that which is absorbed, as well as the transient allowance. Then, \autoref{prop:positive-current} says that the current in the chain over the block must also be large. 

In this section, we build on this and show that a large current over a block must also mean a large right-edge motion after absorption. This is a drift-or-hit dichotomy --- either we have that the norm contracts (hit), or we have a correspondingly large change in the unnormalized shifted position (drift). Lastly, we show that the `drift' eventually leads to `hit' as well, since the right edge position cannot increase beyond a certain point.

Our strategy is as follows.

\begin{enumerate}
    \item {\it The unabsorbed mobility over a block cannot be too large.} In \autoref{lem:finite-current-budget}, we show that the finite chain current corresponds directly to a change of position, accounting explicitly for absorption. Further, after a point, there is nowhere for the right defect weight to go. \autoref{prop:drift-or-hit} then combines the first part of \autoref{lem:finite-current-budget} with \autoref{prop:positive-current} to give a drift-or-hit bound. That is, where \autoref{prop:positive-current} showed that survival (with little absorption) implied current, we now show that survival (with the same little absorption) implies right defect edge motion. 
    
    \item {\it The absorber is hit.} Finally, we will use the second part of \autoref{lem:finite-current-budget}, the bounded right edge position, along with \autoref{prop:positive-current}, to show that the absorber is hit by the defect. This is \autoref{thm:hard-absorber-hitting}.
\end{enumerate}

\subsection{Hard absorption and the drift-or-hit bound}
\label{sec:drift-or-hit}

\begin{lemma}[Finite-chain current budget]
\label{lem:finite-current-budget}
Let $w_n=\mathbb K_t^nv$ for $v\in\range(Q)$, and recall our notation $d_n:=\|w_n\|^2-\|w_{n+1}\|^2.$ Then
\begin{equation}
    \sum_{n=0}^{H-1}\langle w_n,Jw_n\rangle
    =\langle w_H,Xw_H\rangle-\langle v,Xv\rangle
      +X_{\max}(\|v\|^2-\|w_H\|^2).
\label{eq:finite-current-budget}
\end{equation}
This further implies that $\sum_{n=0}^{H-1}\langle w_n,Jw_n\rangle \leq X_{\max}\|v\|^2,
    \  X_{\max}=\frac{t-1}{2}$.
\end{lemma}
\begin{proof}
Let us first decompose $\mathbb{W}\mathbb{K}^{i}_{t}v$, using $Q\mathbb{K}^{i}_{t}v = \mathbb{K}^{i}_{t}v$, and $\mathbb{W}\mathbb{K}^{i}_{t}v = Q\mathbb{W}\mathbb{K}^{i}_{t}v+P\mathbb{W}\mathbb{K}^{i}_{t}v$ to get
\begin{equation}
    \mathbb W\mathbb K_t^iv
    =\mathbb K_t^{i+1}v+P_t\mathbb W\mathbb K_t^iv,
\end{equation}
where we have used that $Q\mathbb{W}Q = \mathbb{K}_{t}$. By Pythagoras' theorem and the unitarity of $\mathbb W$, we have
\begin{equation}
    \|\mathbb K_t^iv\|^2
    =\|\mathbb K_t^{i+1}v\|^2
      +\|P_t\mathbb W\mathbb K_t^iv\|^2,
\end{equation}
so that $d_{n} = \Vert w_{n} \Vert^{2} - \Vert w_{n+1}\Vert^{2} = \Vert w_{n}\Vert^{2} - \Vert\mathbb{K}_{t}w_{n}\Vert^{2}$ is equal to $ \Vert P\mathbb{W}w_{n}\Vert^{2}$. Let us now consider the evolution $\mathbb{W}w_{n}$. Since $w_{n}$ belongs to $\range(Q)$, we can expand $\mathbb{W}w_{n} = \mathbb{W}Qw_{n} = Q\mathbb{W}Qw_{n} + P\mathbb{W}Qw_{n}$.  Identifying $\mathbb{K}_{t} = Q\mathbb{W}Q$, we have $\mathbb{W}w_{n} = w_{n+1} + P\mathbb{W}w_{n}$. With this decomposition, the current 
\begin{equation}
    \langle w_{n}, Jw_{n} \rangle =  \langle w_{n}, (\mathbb{W}^{\dag}X\mathbb{W} - X)w_{n} \rangle = \langle \mathbb{W}w_{n}, X \mathbb{W}w_{n}\rangle - \langle w_{n}, Xw_{n}\rangle
\end{equation}
can be written as 
\begin{equation}
    \langle w_{n}, Jw_{n} \rangle = \langle w_{n+1}, Xw_{n+1}\rangle + \langle P\mathbb{W}w_{n}, XP\mathbb{W}w_{n}\rangle - \langle w_n,Xw_n\rangle.
\end{equation}
Here, the cross terms vanish, since $X$ commutes with both $P$ and $Q$, and the two summands are orthogonal. Since $XP_t=X_{\max}P_t$, and recalling that $d_{n} = \Vert w_{n} \Vert^{2} - \Vert \mathbb{K}_{t} w_{n} \Vert^{2}$, we further have 
\begin{equation}
    \langle w_n,Jw_n\rangle
    =\langle w_{n+1},Xw_{n+1}\rangle
     -\langle w_n,Xw_n\rangle+X_{\max}d_n.
\end{equation}
The sum over $n$ then telescopes, leading exactly to Equation \eqref{eq:finite-current-budget}. To see the upper bound, notice that $0\leq X\leq X_{\max}\mathbb I_{\mathcal K}$, so 
\begin{equation}
    \sum_{n=0}^{H-1}\langle w_n,Jw_n\rangle \leq X_{\max}\Vert w_{H}\Vert^{2} +X_{\max}(\|v\|^2-\|w_H\|^2) = X_{\max}\|v\|^2.
\end{equation}
\end{proof}

Recall $X^{\mathrm{end}}=X-X_{\max}\mathbb I_{\mathcal K}$. Equation~\eqref{eq:finite-current-budget} identifies the accumulated current exactly with the change in this shifted position. Combining it with \autoref{prop:positive-current} gives the drift-or-hit bound in our finite setting.

\begin{proposition}
\label{prop:drift-or-hit}
Assume \autoref{conj:lpm} and \autoref{conj:first-return-gap} and retain the setting from the current Section. Then for $v\in\range(Q)$ and $H \geq 2$,
\begin{equation}
    \langle\mathbb K_t^Hv,X^{\mathrm{end}}\mathbb K_t^Hv\rangle
      -\langle v,X^{\mathrm{end}}v\rangle
\geq\nu_m(t)\sum_{n=0}^{H-2}\|\mathbb K_t^nv\|^2
      -2b_m(t)\|v\|^2
    -(1-2\beta_m(t))
             \bigl(\|v\|^2-\|\mathbb K_t^Hv\|^2\bigr).
\label{eq:finite-drift-or-hit}
\end{equation}
\end{proposition}
\begin{proof}
    Notice that $\langle w_H,Xw_H\rangle-\langle v,Xv\rangle
          +X_{\max}(\|v\|^2-\|w_H\|^2) = \langle w_H,X^{\mathrm{end}}w_H\rangle-\langle v,X^{\mathrm{end}}v\rangle$.
    Then, the proposition directly follows from substituting Equation~\eqref{eq:finite-current-budget} into Equation~\eqref{eq:finite-positive-current}.
\end{proof}

\subsection{Polynomial hitting time for the hard absorber}
\label{sec:hard-hitting-time}

We now combine the previous drift-or-hit bound with the fact that the defect position is bounded. Finally, we will have a hitting time for our hard absorbing process!

\begin{theorem}
\label{thm:hard-absorber-hitting}
Assume \autoref{conj:lpm} and \autoref{conj:first-return-gap}. Then, for every $v \in \range(Q)$, 
\begin{equation}
    \sum_{n=0}^{\infty}\|\mathbb K_t^nv\|^2
    \leq S_m(t)\|v\|^2,
    \label{eq:finite-total-survival}
\end{equation}
where $S_m(t):=\frac{2C_m(t)+2}{c(m)}$, and $C_m(t):=\frac{X_{\max}+2b_m(t)}{1-2\beta_m(t)}$. 

Consequently, for every integer $N\geq1$, we have
\begin{equation}
    \|\mathbb K_t^N\|_{\mathcal K\to\mathcal K}^2
    \leq\frac{S_m(t)}{N+1}.
    \label{eq:finite-hard-hitting-time}
\end{equation}
In particular,  choosing $N=\lceil4S_m(t)\rceil$ gives $\|\mathbb K_t^N\|\leq1/2$, with $N\leq\mathrm{poly}(m,t)$ under our previous assumptions.
\end{theorem}
\begin{proof}
Using the inequality in \autoref{lem:finite-current-budget}, we can bound the left-hand side of Equation~\eqref{eq:finite-drift-or-hit} by $X_{\max}\|v\|^2$. Further, $\|v\|^2-\|\mathbb K_t^Hv\|^2\leq\|v\|^2$. Therefore,
\begin{equation}
    \nu_m(t)\sum_{n=0}^{H-2}\|\mathbb K_t^nv\|^2
    \leq\bigl[X_{\max}+2b_m(t)+1-2\beta_m(t)\bigr]\|v\|^2.
\end{equation}
Substituting in the form of $\nu_m(t)$ from \autoref{prop:positive-current} then yields
\begin{equation}
    \sum_{n=0}^{H-2}\|\mathbb K_t^nv\|^2
    \leq\frac{2C_m(t)+2}{c(m)}\|v\|^2,
\end{equation}
and then the limit $H\to\infty$ proves Equation~\eqref{eq:finite-total-survival}. Since $\mathbb K_t$ is a contraction, $\Vert K_{t}^{i}v\Vert \geq \Vert K_{t}^{i+1}v\Vert$ for all $i$, and hence
\begin{equation}
    (N+1)\|\mathbb K_t^Nv\|^2
    \leq\sum_{n=0}^N\|\mathbb K_t^nv\|^2
    \leq S_m(t)\|v\|^2.
\end{equation}
This yields Equation \eqref{eq:finite-hard-hitting-time} on $\range(Q)$. For arbitrary $v\in\mathcal K$, use $\mathbb K_t^Nv=\mathbb K_t^NQv$ for $N\geq1$ and $\|Qv\|\leq\|v\|$.
\end{proof}

\section{Right defect dynamics of the transfer matrix, $\mathbb T^{(m)}$}
\label{sec:transfer-matrix}

Having shown that defects reach the hard absorber within $\mathrm{poly}(m,t)$ time, we now transition back to the physical transfer matrix $\mathbb T^{(m)}$. The physical left boundary has already been included in our analysis. It remains to compare the hard absorber to the physical right boundary $A_R$.

\subsection{The physical left boundary}
\label{sec:left-boundary}

Recall that solvability makes $A_L$ unitary and ensures that it fixes every local permutation $P_\pi$. Thus it preserves the orthogonal defect sectors and their norms, as used in \autoref{lem:sector-locality}. This is also sufficient for the singleton-sector case in \autoref{lem:stalled-two-step}. \autoref{conj:first-return-gap} concerns the finite process containing this same $A_L$, including repeated interactions with it. 

With this in mind, we record the consequence of \autoref{thm:hard-absorber-hitting} that we will use below.

\begin{corollary}
\label{cor:finite-hard-absorber}
Retain the preceding assumptions and let
\begin{equation}
    N(m,t):=\left\lceil4S_m(t)\right\rceil.
    \label{eq:physical-hard-contraction-time}
\end{equation}
Then, for every physical defect vector $v\in\mathcal K$,
\begin{equation}
    \|\mathbb K_t^{N(m,t)}v\|\leq\frac12\|v\|,
\end{equation}
\begin{equation}
    N(m,t)\leq\mathrm{poly}(m,t).
\end{equation}
\end{corollary}
\begin{proof}
Substitute $N=N(m,t)$ into Equation~\eqref{eq:finite-hard-hitting-time}.
\end{proof}

\subsection{Soft absorption and state-design formation}
\label{sec:soft-absorption}

At this point, we are almost back to $\mathbb T^{(m)}$. We still have the hard absorber at $t$, while the physical transfer matrix has a soft absorber. If the soft absorber, together with the intervening bulk evolution, contracts boundary defect weight over a polynomial number of steps, then we show that this introduces only a polynomial cost in the contraction time. In what follows, we use the layer convention
\begin{equation}
    \mathbb T^{(m)}=A_R\mathbb W.
    \label{eq:physical-soft-walk}
\end{equation}
In the auxiliary-column ordering of Section~\ref{sec:finite-bath-setting}, the transfer matrix is $\mathbb A\mathbb S A_R$. Since $A_R$ acts on site $t$, it commutes with $\mathbb S$, so
\begin{equation}
    \mathbb A^\dagger(\mathbb A\mathbb S A_R)\mathbb A
    =A_R\mathbb S\mathbb A=A_R\mathbb W.
    \label{eq:auxiliary-time-origin}
\end{equation}
This unitary conjugation fixes $\mathcal S$ pointwise and preserves the induced operator norms used below.

\begin{theorem}
\label{thm:soft-absorber}
Let $A_R$ be a contraction fixing $\mathcal S$ pointwise, with $A_RQ=Q$. For an integer $q\geq1$, define $B_q:=A_R(\mathbb W A_R)^{q-1}.$
Further assume that, on $\mathcal K$, for some $0<g(m)\leq1$, we have the operator inequality
\begin{equation}
    \mathbb I_{\mathcal K}-B_q^\dagger B_q\geq  g(m)P_t.
    \label{eq:right-boundary-contraction-assumption}
\end{equation}
Now if, for some integer $N\geq1$ and $0\leq\eta<1$, $\|\mathbb K_t^N\|_{\mathcal K\to\mathcal K}\leq\eta$, then for every $v\in\mathcal K$,
\begin{equation}
    \| (\mathbb T^{(m)})^{N+q}v\|^2
    \leq\left(1-
      \frac{g(m)(1-\eta)^2}{4q(N+1)}
      \right)\|v\|^2.
    \label{eq:hard-to-soft-contraction}
\end{equation}
\end{theorem}

The assumption allows a defect to preserve its norm in one boundary encounter, provided that the intervening bulk evolution exposes it to contraction within $q$ encounters. It is imposed for arbitrary input vectors, including correlations with the other sites.

\begin{proof}
Let $v\in\mathcal K$ and normalize $\Vert v\Vert=1$. For convenience, we suppress the replica-order dependence in the equations below. Write $g=g(m)$.

Since $A_R$ is a contraction fixing $\mathcal S$ pointwise, $\mathcal S$ and $\mathcal K$ are reducing subspaces, that is $A_{R}$ cannot map anything from $\mathcal S$ into $\mathcal{K}$, and vice versa. We can apply the same argument to $Q$: since $A_R$ acts as the identity on $\range(Q)$, this subspace is reducing for $A_R$: if $x\in\range(Q)$, then $A_Rx=x$, and $A_R^\dagger A_R\leq I$ implies $A_R^\dagger x=x$. Hence, $QA_RP_t=0$ and $P_tA_RQ=0$.

We define $v_{i} = \mathbb{T}^{(m)}v_{i-1}$, $v_0=v$, and $w_{i} = \mathbb W(v_{i})$, where $\mathbb W$ is the finite unitary walk on sites $1,\ldots,t$. Notice that $v_{i+1} = A_Rw_{i}$ and $B_qw_i=v_{i+q}$. Now since $\mathbb W$ is unitary, $\Vert w_{i} \Vert = \Vert v_{i} \Vert$. By assumption,
\begin{equation}
    \Vert v_i\Vert^2-\Vert v_{i+q}\Vert^2
    =\langle w_i,(\mathbb I-B_q^\dagger B_q)w_i\rangle
    \geq g\Vert P_t\mathbb Wv_i\Vert^2.
\end{equation}
Summing both sides from $i=0$ to $N$, each one-step squared-norm loss occurs at most $q$ times. Since these losses are nonnegative,
\begin{equation}
\begin{aligned}
    g\sum_{i=0}^N\Vert P_t\mathbb Wv_i\Vert^2
    &\leq\sum_{i=0}^N\sum_{j=i}^{i+q-1}
       \bigl(\Vert v_j\Vert^2-\Vert v_{j+1}\Vert^2\bigr)\\
    &\leq q\bigl(\Vert v\Vert^2-\Vert v_{N+q}\Vert^2\bigr).
\end{aligned}
\label{eq:soft-loss-sum}
\end{equation}

Now we must relate this to the hard absorbing walk. The absorbing walk respects the permutation sector decomposition, so we instead split $v_{i} = P_{t}v_{i} + Qv_{i} := p_{i} + q_{i}$. Now $q_{i+1} = QA_R\mathbb Wq_{i} + QA_R\mathbb Wp_{i}$, so from the definition of $\mathbb{K}_{t}$ and the fact that $A_R$ acts as the identity on $\range(Q)$,
\begin{equation}
    q_{i+1} = \mathbb{K}_{t}q_{i} + Q\mathbb Wp_{i}.
\end{equation}
which we can iterate to get 
\begin{equation}
    q_{N+1} = \mathbb{K}^{N}_{t}q_{1} + \sum_{i=1}^{N} \mathbb{K}_{t}^{N-i}Q\mathbb Wp_{i}.
\end{equation}
Now, since $\mathbb{K}_{t}$ is a contraction, and $\Vert \mathbb{K}^{N}_{t}\Vert \leq \eta$, by the triangle inequality
\begin{equation}
    \Vert q_{N+1} \Vert \leq  \eta \Vert q_{1} \Vert + \sum_{i=1}^{N} \Vert p_{i} \Vert.
\end{equation}
At the same time, $\Vert p_{i+1}\Vert = \Vert P_{t}A_R\mathbb Wv_{i}\Vert \leq \Vert P_{t}\mathbb Wv_{i}\Vert$. Thus,
\begin{equation}
    \Vert v_{N+1} \Vert \leq \Vert q_{N+1} \Vert + \Vert p_{N+1}\Vert \leq \eta \Vert q_{1} \Vert + \sum_{i=0}^{N} \Vert P_{t}\mathbb Wv_{i} \Vert.
\end{equation}
Proceeding, by the Cauchy-Schwarz inequality,
\begin{equation}
    \sum_{i=0}^{N} \Vert P_{t}\mathbb Wv_{i} \Vert \leq \Big( (N+1) \sum_{i=0}^{N} \Vert P_{t}\mathbb Wv_{i} \Vert^{2}\Big)^{1/2},
\end{equation}
which we can simplify using Equation~\eqref{eq:soft-loss-sum}, $\Vert v \Vert =1$, and $\Vert q_1\Vert\leq1$. Since $q\geq1$ and $\mathbb T^{(m)}$ is a contraction, we get
\begin{equation}
    \Vert v_{N+q} \Vert \leq \eta + \Big( q(N+1)\frac{(1 - \Vert v_{N+q}\Vert^{2})}{g}\Big)^{1/2}.
\end{equation}
Either $\Vert v_{N+q}\Vert\leq(1+\eta)/2$, in which case
\begin{equation}
    1-\Vert v_{N+q}\Vert^2
    \geq 1-\frac{(1+\eta)^2}{4}
    \geq\frac{(1-\eta)^2}{4}
    \geq\frac{g(1-\eta)^2}{4q(N+1)},
\end{equation}
or $\Vert v_{N+q}\Vert>(1+\eta)/2$. In this case, 
\begin{equation}
   \frac{1-\eta}{2} <  \Big( q(N+1)\frac{(1 - \Vert v_{N+q}\Vert^{2})}{g}\Big)^{1/2},
\end{equation}
or, completing our proof, 
\begin{equation}
    \Vert v_{N+q}\Vert^{2} \leq 1-\frac{g(1-\eta)^{2}}{4q(N+1)}.
\end{equation}
\end{proof}

We have shown that $\mathbb T^{(m)}$ contracts defect vectors by an inverse-polynomial amount in a time scale $N+q$. Recall \autoref{cor:finite-hard-absorber}: the time scale required to have a constant contraction of $\mathbb K_t$ is polynomial in $m,t$, so choosing $N$ on this scale gives contraction on the order of $g(m)/(qN)$. As a result, if $g(m)$ itself is at least inverse polynomial and $q$ is polynomially bounded, then so is the contraction. How typical is the assumption on $q$? Well, since $A_{R}$ is solvable, it must have the permutations as part of its $+1$ eigenspace; generically, it could have other directions as well in there. The point is that the interaction of $A_{R}$ with the bulk, which is comprised of generic dual unitaries, should be unlikely to align these `bad' directions into themselves, while by design these bulk dual unitaries leave permutations unchanged. This is very similar to LPM, and can be thought of as a boundary analogue of the same.

We also have the following corollary.

\begin{corollary}
\label{cor:design-time}
    Choose $N=N(m,t)$ as in Equation~\eqref{eq:physical-hard-contraction-time}, write $q=q(m)$, and define
    \begin{equation}
        \tau_m(t):=\frac{32q(m)(N(m,t)+1)(N(m,t)+q(m))}{g(m)}.
        \label{eq:replica-contraction-timescale}
    \end{equation}
    Then, for every integer $\ell\geq0$ and every physical defect vector $v\in\mathcal K$,
    \begin{equation}
        \|(\mathbb T^{(m)})^\ell v\|
        \leq2\exp\left(-\frac{\ell}{\tau_m(t)}\right)\|v\|.
        \label{eq:physical-defect-contraction}
    \end{equation}
    If $q(m),g(m)^{-1}\leq\mathrm{poly}(m)$ and the boundary assumption holds uniformly in $t$, then $\tau_m(t)\leq\mathrm{poly}(m,t)$, so a polynomial number of auxiliary steps provides us with constant worst-case contraction of defects.
\end{corollary}
\begin{proof}
    The proof follows by simple substitution. Choose $\eta=1/2$. Then,
    \begin{equation}
        \|(\mathbb T^{(m)})^{N+q}v\|^2
        \leq\left(1-\frac{g(m)}{16q(N+1)}\right)\|v\|^2.
    \end{equation}
    Taking the square root and using the inequality $\sqrt{1-x}\leq e^{-x/2}$ gives
    \begin{equation}
        \|(\mathbb T^{(m)}|_{\mathcal K})^{N+q}\|
        \leq\exp\left(-\frac{g(m)}{32q(N+1)}\right).
    \end{equation}
    Now we take $b=\lfloor\ell/(N+q)\rfloor$. We don't have to worry about any remaining fewer than $N+q$ steps by contractivity, so
    \begin{equation}
    \begin{aligned}
        \|(\mathbb T^{(m)})^\ell v\|
        &\leq\exp\left(-\frac{g(m)b}{32q(N+1)}\right)\|v\|\\
        &\leq\exp\left(\frac{g(m)}{32q(N+1)}\right)
          \exp\left(-\frac{\ell}{\tau_m(t)}\right)\|v\|.
    \end{aligned}
    \end{equation}
    The first exponential is less than $2$, completing the proof.
\end{proof}

The dependence of $\tau_m(t)$ on $t$ stems from the hard-absorber time $N(m,t)$, which in turn comes from us accounting for the width-dependent transient, see Theorem~\ref{thm:hard-absorber-hitting}.

Now recall the $m$-moment Haar twirling operator $\Pi^{(m)}$ (\autoref{sec:state-designs}). Its range on the $m$-replica auxiliary chain is the global permutation sector $\mathcal S$, and its kernel is $\mathcal K$. Thus, the contraction of physical defects by an operator is equivalent to its convergence to the Haar twirling operator. Let us make this equivalence explicit.

\begin{lemma}[Defect contraction and the induced norm]
\label{lem:defect-contraction-equivalence}
Let $\mathbb T^{(m)}$ be a contraction fixing $\mathcal S$ pointwise, and let $\Pi^{(m)}=\Pi_{\mathcal S}$. Then, for every integer $\ell\geq0$,
\begin{equation}
    \|(\mathbb T^{(m)})^\ell-\Pi^{(m)}\|
    =\| (\mathbb T^{(m)}|_{\mathcal K})^\ell\|
    =\sup_{\substack{v\in\mathcal K\\\|v\|=1}}
         \|(\mathbb T^{(m)})^\ell v\|.
    \label{eq:defect-contraction-equivalence}
\end{equation}
\end{lemma}
\begin{proof}
We have, by an argument similar to that in the proof of \autoref{thm:soft-absorber}, that $\mathcal S$ and $\mathcal K$ are reducing subspaces for $\mathbb{T}^{(m)}$, so
\begin{equation}
    (\mathbb T^{(m)})^\ell-\Pi^{(m)}
    =0_{\mathcal S}\oplus
      (\mathbb T^{(m)}|_{\mathcal K})^\ell.
\end{equation}
Taking the induced operator norm on both sides yields the result; under vectorization this is the $2\to2$ norm.
\end{proof}

\begin{corollary}
\label{cor:twirling-convergence}
    Retain the setting from \autoref{cor:design-time}. Then
    \begin{equation}
    \|(\mathbb T^{(m)})^\ell-\Pi^{(m)}\|_{2\to2}
        \leq2\exp\left(-\frac{\ell}{\tau_m(t)}\right),\ \ell\geq0.
        \label{eq:twirling-exponential-contraction}
\end{equation}
\end{corollary}
\begin{proof}
    Apply \autoref{lem:defect-contraction-equivalence} to Equation~\eqref{eq:physical-defect-contraction}.
\end{proof}

This supplies the moment-operator estimate that we will use in the next section, to deal with the normalization in the projected ensemble, for a general replica order $m$.

\section{From contraction of $\mathbb{T}^{(m)}$ to projected ensemble state design}
\label{sec:normalization-to-design}

Given the contraction bounds on $\mathbb{T}^{(m)}$, we have to now show that the corresponding moment operator of the projected ensemble state distribution is close to the Haar distribution. This is actually a subtle issue, since (as discussed in \autoref{sec:preliminaries}), the contraction of $\mathbb{T}^{(n+k)}$ only immediately gives us bounds on operators of the form $\tilde{\rho}^{(n,k)}_{\mathcal{E}} = \sum_{i=1}^{2^{N_{B}}} (p_{i})^{n}\ketbra*{\tilde{\psi}^{\mathcal{E}}_{i}}^{\otimes k}$ for unnormalized $\ket*{\tilde{\psi}}$ and $n \geq 0$, $k \geq 1$. However, normalization means that we actually need this quantity for $n=1-k$. This is a real issue. In the infinite bath case, the result in \cite{Claeys2022emergentquantum} is obtained via a `replica calculation', where for every nonnegative integer $n$, one finds $\tilde{\rho}^{(n,k)}_{\mathcal{E}} \propto \rho^{(k)}_{\mathrm{Haar}}$ in the limit, so they conclude that this must follow for $n=1-k$. We are not sure such an idea is extendible to the finite bath case, and instead we will carefully try to control this negative exponent.

Our approach will be to perform an approximation of $p_{i}^{1-k}$ as a polynomial, using a weighted error. The key observation is that we do \textit{not} need a uniformly accurate approximation to $p_{i}^{1-k}$. This would be hopeless for $k \geq 2$ in the limit $p_{i} \rightarrow 0$, where the function diverges, while any fixed polynomial remains bounded, and we might as well have not tried. We only need the approximation to be good relative to the frequency with which the $p_{i}$ contributes to the actual, physical state, and in what follows, we will find that the singularity as $p_{i} \rightarrow 0$ is compensated for by the fact that our approximation there is allowed to be poor.

\subsection{Set-up for the approximation}
Recall (from Equation~\eqref{eq:physical-projected-moment}) $\rho^{(k)}_{\mathcal{E}} =\sum_{i=1}^{2^{N_{B}}} p_{i}(\ketbra{{\psi}^{\mathcal{E}}_{i}}{{\psi}^{\mathcal{E}}_{i}})^{\otimes k}:= \sum_{i=1}^{2^{N_{B}}} p_{i}\rho_{i}^{\otimes k}$, where the $\ket{{\psi}^{\mathcal{E}}_{i}}$ are normalized, and $\rho_{i}$ refers to the normalized density matrix corresponding to the measurement outcome $i$. Writing this in terms of the unnormalized $\ket{\tilde{\psi}^{\mathcal{E}}_{i}}$ gets us
\begin{equation}
    {\rho}^{(k)}_{\mathcal{E}} = \sum_{i=1}^{2^{N_{B}}} (p_{i})^{1-k} \ketbra*{\tilde{\psi}^{\mathcal{E}}_{i}}^{\otimes k} := \sum_{i=1}^{2^{N_{B}}} \frac{(X_{i})^{\otimes k}}{p_{i}^{k-1}},
\label{eq:unnormalized-physical-moment}
\end{equation}
with $X_{i} = \ketbra*{\tilde{\psi}^{\mathcal{E}}_{i}}$, and $p_{i} = \mathrm{tr}(X_{i})$. Terms with $p_i=0$ contribute zero. From solvability, we have the following decomposition of the unnormalized state.
\begin{equation}
     \ket{\tilde{\psi}^{\mathcal{E}}_{i}} = \sqrt{\frac{\nu_{i}}{P}}V\ket{\phi_{i}},
\end{equation}

\begin{equation}
     \Vert \ket{\phi_{i}}\Vert=1, \ VV^{\dag} = \mathbb{I}_{d}, \ P = \frac{d}{D}, \ \nu_{i} = 2^{-N_{B}}.
\end{equation}

Here $V := 2^{(N_{A}-t_{\mathrm{phys}})/2}F_{A}$, $\ket{\phi_{i}} := 2^{-(t_{\mathrm{phys}-2})/2}U_{\xi_{\ell}}U_{\xi_{\ell-1}}\dots U_{\xi_1}\ket{r_a}$. The properties above are easily verified just by dimension counting. We assume $k \geq 2$ and $D>d$ below; the cases $k=1$ and $D=d$ do not require this approximation.

Here, the states $\ket{\phi_{i}}$ are those corresponding to the evolution of normalized auxiliary seeds under $\ell=(N_B-1)/2$ periods of the auxiliary unitaries, in the paired-column convention, with their measurement outcomes fixed to those in the string $i$; the remaining endpoint outcome fixes the seed. The map $V$ is a coisometry that maps this state into the output of the projected ensemble corresponding to the measurement outcome $i$, up to the scalar factor above. Note also that $\nu_{i} = 2^{-N_{B}}$ implies that $\nu$ is the uniform distribution on the $2^{N_{B}}$ basis strings. Define $Y_{i} = V\ketbra{\phi_i}{\phi_i}V^{\dag}$, and $s_{i} = \mathrm{tr}(Y_{i})$. Then $\ketbra{\tilde{\psi}^{\mathcal{E}}_{i}}{\tilde{\psi}^{\mathcal{E}}_{i}} = X_{i} = \frac{\nu_{i}}{P}Y_{i}$. Taking the trace of both sides, we have $p_{i} = \frac{\nu_{i}}{P}s_{i}$. Since $\sum_i p_i=1$, we also have $\mathbb{E}_{\nu}s=P$. Moreover, $Y_i=s_i\rho_i$ and $0\leq s_i\leq 1$, with zero contributions when $s_i=0$. Substituting these expressions in Equation~\eqref{eq:unnormalized-physical-moment}, we get

\begin{equation}
    {\rho}^{(k)}_{\mathcal{E}} = \frac{1}{P} \sum_{i=1}^{2^{N_{B}}} \nu_{i} \frac{Y_i^{\otimes k}}{s_i^{k-1}} = \frac{1}{P}\mathbb{E}_{\nu}\Big(\frac{Y^{\otimes k}}{s^{k-1}}\Big).
\end{equation}

Notice that physically $p_{i}$ is the Born probability, while $\nu_{i}$ is uniform, so the entirety of the variation in the Born probabilities is encapsulated in $s_{i}$, or equivalently the interaction of $V$ with the auxiliary state $\ket{\phi_{i}}$. The reason we have written our $\rho^{(k)}_{\mathcal{E}}$ in terms of these $s_{i}$ and $Y_{i}^{\otimes k}$ is that the quantity that we have control over is the contraction of $\mathbb{T}^{(k)}$, which controls $\mathbb{E}_{\nu}(Y^{\otimes k})$ after contraction with the boundary maps; in the normalized definition of $\mathbb{T}^{(k)}$, we sum over the measurement outcomes with equal probability.

Let us now proceed further. We define $r=k-1$ and want to find a real degree-$L$ polynomial $P_{L}(s)$ such that the error

\begin{equation}
    e_{L} = \Big\Vert\frac{1}{P}\mathbb{E}_{\nu}\Big(s^{-r}Y^{\otimes k} - P_{L}(s)Y^{\otimes k}\Big)\Big\Vert_1
\end{equation}

is small. Here $\Vert\cdot\Vert_1$ is the trace norm. Let us simplify this error term with the simple upper bound

\begin{equation}
    e_{L} \leq \frac{1}{P}\mathbb{E}_{\nu}\Big( \Vert s^{-r}Y^{\otimes k} - P_{L}(s)Y^{\otimes k}\Vert_1 \Big).
\end{equation}

Now

\begin{equation}
    e_{L} \leq \frac{1}{P}\mathbb{E}_{\nu}\Big( \Vert s^{-r}Y^{\otimes k} - P_{L}(s)Y^{\otimes k}\Vert_1 \Big) = \frac{1}{P}\mathbb{E}_{\nu}\Big( \Vert s^{-r}s^{k} \rho^{\otimes k} - P_{L}(s)s^{k} \rho^{\otimes k}\Vert_1 \Big),
\end{equation}

which we simplify, using $\Vert \rho^{\otimes k} \Vert_1 = 1$, to

\begin{equation}
    e_{L} \leq \frac{1}{P}\mathbb{E}_{\nu}\Big( \Vert s\rho^{\otimes k} - P_{L}(s)s^{r+1} \rho^{\otimes k}\Vert_1 \Big) = \frac{1}{P}\mathbb{E}_{\nu}\Big( |s-s^{r+1}P_{L}(s)| \Big).
\end{equation}

Then, by Cauchy--Schwarz and $\mathbb{E}_{\nu}s=P$,

\begin{equation}
    e_{L} \leq \Big( \frac{1}{P}\mathbb{E}_{\nu}\big( s(1-s^{r}P_{L}(s))^{2}\big)\Big)^{1/2}.
\end{equation}

\subsection{Bounding the least-square error for Haar}
\label{sec:haar-least-squares}
Now, how do we approach this minimization problem? We need to pick a $P_{L}$ such that this squared error bound is minimized, with the expectation carried out over $\nu$. However, the distribution $s(\nu)$ does not have a simple form. Recall that $s_{i}$ determines the Born probability corresponding to the auxiliary state $\ket{\phi_{i}}$, which is in general very complicated. Our approach is to instead consider the $P_{L}$ which minimizes the above upper bound on $e_{L}^{2}$ with $\phi_{i}$ sampled instead over the Haar distribution --- after all, we showed that $\mathbb{T}^{(m)}$ contracts all the vectors outside the span of permutations, so their polynomial moments can be controlled by the contraction bounds at the corresponding replica orders, and, further, since the Haar distribution is `simple', we should have more analytical tools to bound the square error. Once we do this, we will show that we do not lose ``too much" in going back to $\nu$. The squared-error polynomial has degree at most $R=2k+2L-1$, so this step will use the contraction bound on $\mathbb{T}^{(R)}$.

We now focus on the case where $\ket{\phi} \sim \mathrm{Haar}$. $D$-dimensional Haar random states are achieved by sampling $D$ i.i.d. complex Gaussians $\mathcal{N}_{\mathbb{C}}(0,1)$ and normalizing the resulting vector, and the squared norm of the projection onto the smaller subspace of dimension $d$ is distributed as $\sum_{i=1}^{d} |g_i|^{2}/\sum_{i=1}^{D} |g_i|^{2}$, where the same $g_i\sim\mathcal{N}_{\mathbb{C}}(0,1)$ appear in both sums, so that $s \sim \mathrm{Beta}(d,D-d)$ (for a more detailed discussion, see \cite{Zyczkowski2001Induced}).

Let us now write out the squared error bound in full by inserting the density function corresponding to the $\mathrm{Beta}(d,D-d)$ distribution:

\begin{equation}
    \tilde{e}_{L}^{2} := \frac{1}{P}\int_{0}^{1}\frac{s^{d-1}(1-s)^{D-d-1}}{B(d,D-d)}\,s(1-s^{r}P_{L}(s))^{2}\,\dd s
    = \int_{0}^{1}\frac{s^{d+2r}(1-s)^{D-d-1}}{B(d+1,D-d)}(s^{-r}-P_{L}(s))^{2}\,\dd s.
\end{equation}

Here we used $P B(d,D-d)=B(d+1,D-d)$. The tilde over the $e_{L}$ is simply to emphasize that this is the squared error bound for the Haar distribution. Now $B(d+1,D-d)$ is a constant independent of $s$, so we can drop it when finding the minimizing polynomial. Identifying $w(s)=s^{d+2r}(1-s)^{D-d-1}$, we have the standard $w(s)$-weighted least-squares minimization of $P_{L}(s)$ with square error

\begin{equation}
    \tilde{e}_{L}^{2} = \frac{1}{B(d+1,D-d)}\int_{0}^{1} w(s)(s^{-r}-P_{L}(s))^{2}\,\dd s,
\end{equation}

with an additional factor $1/B(d+1,D-d)$ which we can think of as a physical normalization. We can now rely directly on results from approximation theory. Following the proof of Theorem 2.4 of \cite{Xu2018Jacobi}, in particular the second-last equation in the proof, we have the bound on the weighted least-square error of $f$

\begin{equation}
    \tilde{e}_{L}^{2} \leq \frac{\int_{0}^{1} w(s)(1-s)s(f'(s))^{2}\,\dd s}{B(d+1,D-d)\lambda_{L+1}},
\end{equation}

where $\lambda_j=j(j+D+2r)$, so $\lambda_{L+1}=(L+1)(L+D+2r+1)$, and for us $f=s^{-r}$. The suggestive notation $\lambda$ is due to the fact that this is an eigenvalue corresponding to the Sturm--Liouville theory of Jacobi polynomials. Now we compute the right hand side. The integral

\begin{equation}
    \int_{0}^{1} w(s)(1-s)s(f'(s))^{2}\,\dd s = \int_{0}^{1} w(s)(1-s)s(rs^{-r-1})^{2}\,\dd s = r^{2}B(d,D-d+1),
\end{equation}

and so we have
\begin{equation}
    \tilde{e}_{L}^{2} \leq \frac{r^{2}B(d,D-d+1)}{B(d+1,D-d)\lambda_{L+1}} = \frac{(k-1)^{2}(D-d)}{d(L+1)(L+D+2k-1)} \leq \frac{(k-1)^2}{d(L+1)}.
\label{eq:haar-error-bound}
\end{equation}

That is, to ensure $\tilde{e}_{L} \leq \alpha/2$ for some $O(1)$ constant $\alpha>0$, we need $L+1 \geq 4(k-1)^{2}/(d\alpha^{2})$, which is still at most $\mathrm{poly}(k)$ for fixed $\alpha$. That is, we established that the squared least-squares error is bounded as $O(k^2/[d(L+1)])$ for a degree-$L$ polynomial.

\subsection{Projected ensembles form approximate state designs}

Now we transition back to the projected ensemble setting. We have to show that this forms an approximate state $k$-design. We begin by bounding the trace norm of the difference of the $k$-moments of the two distributions.

\begin{proposition}
\label{prop:normalization-conversion}
Assume \autoref{thm:soft-absorber} and the setting of this section. Define $\sigma^{\nu}_{m} := \mathbb{E}_{\nu}(\ketbra{\phi}^{\otimes m})$, and $\sigma^{H}_{m} := \mathbb{E}_{H}(\ketbra{\phi}^{\otimes m})$, where $H$ is the Haar measure. Also let $P_{L} =\sum_{i=0}^{L} c_{i}s^{i}$ be the least-squares approximation to $s^{1-k}$ for the Haar measure as described in the previous section, and $R=2k+2L-1$. Then with
\begin{equation}
     \alpha_{m} := \frac{1}{2}\Vert \sigma^{\nu}_{m} - \sigma^{H}_{m}\Vert_{1}, 
\end{equation}
\begin{equation}
    S_{L} := \sum_{i=0}^{L} |c_{i}|,
\end{equation}
we have the bound
\begin{equation}
    D_{k} := \Vert \rho_{\mathcal{E}}^{(k)} - \rho_{H}^{(k)}\Vert_{1} \leq 2\tilde{e}_{L} + (1+S_{L})\sqrt{\frac{2\alpha_{R}}{P}} + \frac{2S_{L}}{P}\alpha_{R}.
\end{equation}
\end{proposition}

\begin{proof}
For simplicity of notation, define $A^{\nu}_{k,L} = \frac{1}{P}\mathbb{E}_{\nu}(Y^{\otimes k}P_{L}(s))$, and similarly as well $A^{H}_{k,L}$ for the Haar distribution. We have by the triangle inequality

\begin{equation}
    D_{k} \leq \Vert \rho_{\mathcal{E}}^{(k)} - A^{\nu}_{k,L}\Vert_1 + \Vert A^{\nu}_{k,L} - A^{H}_{k,L} \Vert_1 + \Vert A^{H}_{k,L} - \rho_{H}^{(k)} \Vert_1.
\end{equation}

By the discussion in \autoref{sec:haar-least-squares}, we immediately have $\Vert A^{H}_{k,L} - \rho_{H}^{(k)} \Vert_1 \leq \tilde{e}_{L}$. We now need to control the other two terms. 

\begin{enumerate}
    \item \textbf{Term 1: $\Vert A^{\nu}_{k,L} - A^{H}_{k,L} \Vert_1$}. Begin with the useful identity

\begin{equation}
    Y^{\otimes k}s^{j} = \mathrm{tr}_{k+1, k+2 \dots k+j}Y^{\otimes(k+j)},
\end{equation}

which simply follows from the definition of $s^{j}$. Now we take the expectation value on both sides with respect to $\nu$, use the fact that the trace is linear, and expand out the definition of $Y$ to get 

\begin{equation}
    \mathbb{E}_{\nu}(Y^{\otimes k}s^{j}) = \mathrm{tr}_{k+1, k+2 \dots k+j} (V^{\otimes(k+j)}\mathbb{E}_{\nu}(\ketbra{\phi}^{\otimes(k+j)})(V^\dagger)^{\otimes(k+j)}).
\end{equation}

We can rewrite this as 
\begin{equation}
    \mathbb{E}_{\nu}(Y^{\otimes k}s^{j}) = \Lambda^{(k,j)}(\mathbb{E}_{\nu}(\ketbra{\phi}^{\otimes(k+j)})),
\end{equation}

where $\Lambda^{(k,j)}$ is trace non-increasing and contractive, since it consists of an orthogonal projection and a partial trace. Now 

\begin{equation}
    \Vert A^{\nu}_{k,L} - A^{H}_{k,L} \Vert_1 =  \Vert \frac{1}{P}\sum_{j=0}^{L} c_{j}\mathbb{E}_{\nu}(Y^{\otimes k}s^{j})  - \frac{1}{P}\sum_{j=0}^{L} c_{j}\mathbb{E}_{H}(Y^{\otimes k}s^{j})\Vert_1,
\end{equation}

which we can write, by the triangle inequality, as

\begin{multline}
    \Vert A^{\nu}_{k,L} - A^{H}_{k,L} \Vert_1 \leq \frac{1}{P}\sum_{j=0}^{L} |c_{j}| \Vert \mathbb{E}_{\nu}(Y^{\otimes k}s^{j})-\mathbb{E}_{H}(Y^{\otimes k}s^{j})\Vert_1
    \\= \frac{1}{P}\sum_{j=0}^{L} |c_{j}| \Vert \Lambda^{(k,j)}(\mathbb{E}_{\nu}(\ketbra{\phi}^{\otimes(k+j)}))-\Lambda^{(k,j)}(\mathbb{E}_{H}(\ketbra{\phi}^{\otimes(k+j)}))\Vert_1.
\end{multline}

But since $\Lambda^{(k,j)}$ is contractive, we can drop it immediately, and then get the simple bound (since $R>k+L$),

\begin{equation}
    \Vert A^{\nu}_{k,L} - A^{H}_{k,L} \Vert_1 \leq \frac{2S_{L}\alpha_{R}}{P}.
\label{eq:polynomial-moment-error}
\end{equation}

\item\textbf{Term 2: $\Vert \rho_{\mathcal{E}}^{(k)} - A^{\nu}_{k,L} \Vert_1 $}. From the preceding discussion we immediately have 

\begin{equation}
    \Vert \rho_{\mathcal{E}}^{(k)} - A^{\nu}_{k,L} \Vert_1  \leq \Big(\frac{1}{P} \mathbb{E}_{\nu}[s(1-s^{r}P_{L}(s))^{2}]\Big)^{1/2}.
\end{equation}

The issue is that we never actually calculated this quantity; we instead replaced the state distribution $\nu$  here with the Haar distribution, and then bounded the right hand side. The key observation we use is that the squared residual is itself a polynomial in $s$, so we can apply a similar strategy as the previous term. We expand 

\begin{equation}
    s(1-s^{r}P_{L}(s))^{2} = s-2s^{r+1}P_{L}(s) + s^{2r+1}(P_{L}(s))^{2},
\label{eq:squared-residual-expansion}
\end{equation}

which, as noted before, is a polynomial of degree at most $R$. Let us expand this as $s(1-s^{r}P_{L}(s))^{2} = \sum_{i=0}^{R} a_{i}s^{i}$ for some set of coefficients $\{ a_{i} \}$. Now 

\begin{equation}
\begin{aligned}
 \Vert \rho_{\mathcal{E}}^{(k)} - A^{\nu}_{k,L} \Vert_1
 &\leq \Big(\frac{1}{P}\left|\mathbb{E}_{\nu}
   [\sum_{i=0}^{R} a_{i}s^{i}]\right|\Big)^{1/2}\\
 &= \Big(\frac{1}{P}\left|\mathbb{E}_{H}
   [\sum_{i=0}^{R} a_{i}s^{i}]
   +\Big(\mathbb{E}_{\nu}[\sum_{i=0}^{R} a_{i}s^{i}]
   -\mathbb{E}_{H}[\sum_{i=0}^{R} a_{i}s^{i}]\Big)
   \right|\Big)^{1/2},
\end{aligned}
\end{equation}

so we can write 

\begin{equation}
 \Vert \rho_{\mathcal{E}}^{(k)} - A^{\nu}_{k,L} \Vert_1 \leq  \bigg(\tilde{e}_{L}^{2} + \frac{1}{P}\left|\mathbb{E}_{\nu}[\sum_{i=0}^{R} a_{i}s^{i}]- \mathbb{E}_{H}[\sum_{i=0}^{R} a_{i}s^{i}]\right|\bigg)^{1/2}, 
\end{equation}

and the last term on the right-hand side is exactly of the same form as in Equation~\eqref{eq:polynomial-moment-error}, so we can immediately bound, by the same argument,

\begin{equation}
 \Vert \rho_{\mathcal{E}}^{(k)} - A^{\nu}_{k,L} \Vert_1 \leq  \Big(\tilde{e}_{L}^{2} + \frac{1}{P}\sum_{i=1}^{R}|a_{i}|2\alpha_{R}\Big)^{1/2}
\end{equation}

Now the sum of absolute coefficients can be bounded using the expansion in terms of $P_{L}$, that is, Equation~\eqref{eq:squared-residual-expansion} as 

\begin{equation}
    \sum_{i=1}^{R}|a_{i}| \leq 1+2S_{L} + S_{L}^{2} = (S_{L}+1)^{2}.
\end{equation}

Altogether,

\begin{equation}
     \Vert \rho_{\mathcal{E}}^{(k)} - A^{\nu}_{k,L} \Vert_1 \leq \bigg(\tilde{e}_{L}^{2} + \frac{2(S_{L}+1)^{2}\alpha_{R}}{P} \bigg)^{1/2}.
\end{equation}

Putting the three pieces together and taking $(y^{2}+x)^{1/2} \leq y+ \sqrt{x}$ yields our result.
\end{enumerate}
\end{proof}

The current result is a little opaque. Where are the parameters $N_{B}$ and $k$ that we are trying hard to control? First, we have a bound on $\tilde{e_{L}}$ already in terms of $k, L$, and the fixed parameter $d$, and we have discussed that we can keep this term $O(k^{2})$ while picking $L \leq \mathrm{poly}(k)$. Now, what about $S_{L}$ and $\alpha_{R}$? It turns out (see Appendix \ref{app:jacobi-coefficients}) that we can bound $S_{L}$ as 
\begin{equation}
    S_L \leq (L+1)(2M)^L(D+2M)^{k-1+2L},
\label{eq:coefficient-bound}
\end{equation}

with $M = k-1+L$.  We will see how this works in the proof of \autoref{thm:normalized-state-design}, and for now move on to $\alpha_{R}$. For simplicity of notation, let $\ell=(N_B-1)/2$ be the number of auxiliary transfer steps. Define

\begin{equation}
    \delta_R(\ell)
    :=
    \left\|(\mathbb{T}^{(R)})^\ell-\Pi_R\right\|_{2\to 2} \leq 2\exp(-\ell/\tau_R(t)),
\end{equation}

where $\Pi_R=\Pi^{(R)}$ is the Haar twirling projector, and in the inequality we have used \autoref{cor:twirling-convergence}. $\tau_R(t)$ is the contraction timescale at replica order $R$. Writing $\ket{\phi_a^{\mathrm{in}}}:=\ket{r_a}/\|r_a\|$, the input moment is
\begin{equation}
    \sigma_m^{\mathrm{in}}:=\frac12\sum_{a=0}^1
       (\ketbra{\phi_a^{\mathrm{in}}})^{\otimes m},
\end{equation}
\begin{equation}
    \sigma_m^\nu=(\mathbb T^{(m)})^\ell[\sigma_m^{\mathrm{in}}],
\end{equation}
\begin{equation}
    \sigma_m^H=\Pi^{(m)}[\sigma_m^{\mathrm{in}}].
\end{equation}
Here the auxiliary-column convention is understood; the unitary change of convention in \autoref{sec:soft-absorption} preserves the contraction bound. Since the input moment has Hilbert--Schmidt norm at most one, and both auxiliary moments are supported on the symmetric subspace,
\begin{equation}
    \alpha_R \leq \frac{1}{2} \binom{D+R-1}{R}^{1/2} \delta_R(\ell).
\end{equation}

It seems that $S_{L}$ and $\alpha_{R}$ have the dimension factor $D$ in them, which would kill our argument. However, for $S_{L}$, the dependence only enters as a logarithm in the actual design time, and with $L \leq \mathrm{poly}(k)$ we get $\log(1+S_L) = O\big( \mathrm{poly}(k) \log(D+\mathrm{poly}(k))\big)$. For $\alpha_{R}$, the exponentially decreasing $\delta_{R}(\ell)$ can compensate for the exponential dimension for some $N_{B}\leq\mathrm{poly}(k,t)$ at fixed accuracy. This is the content of \autoref{thm:normalized-state-design}, which shows state design formation in $N_{B}\leq\mathrm{poly}(k,t)$ at fixed accuracy.

\begin{theorem}
\label{thm:normalized-state-design}
Retain the setting and notation from before. Take $k \geq 2$ and $0 < \varepsilon \leq 1$. The projected ensemble forms an $\varepsilon$-approximate state $k$-design in trace norm, that is, 
\begin{equation}
    D_{k} \leq \varepsilon,
\end{equation}
with a bath size 
\begin{equation}
    N_{B} = O\Big(1 + \tau_R(t)\big(R \log (D+R) + \log(1/\varepsilon)\big)\Big)\,,
\end{equation}
\begin{equation}
    R = O\Big(k + \frac{k^{2}}{d\varepsilon^{2}}\Big).
\end{equation}
\end{theorem}

\begin{proof}
We use the result of \autoref{prop:normalization-conversion}, and bound the individual terms on the right-hand side.

\begin{enumerate}

\item \textbf{Term 1: $2\tilde{e}_{L}$.}   We can immediately use Equation~\eqref{eq:haar-error-bound} to get
\begin{equation}
    \tilde{e}_{L}^{2} \leq \frac{(k-1)^{2}}{d(L+1)}\leq \frac{\varepsilon^{2}}{16},
\end{equation}
for the choice $L = \lceil\frac{16(k-1)^{2}}{d \varepsilon^{2}}\rceil -1$. This gives us $2 \tilde{e}_{L} \leq \varepsilon/2$. Since $R=2k+2L-1$, this is manifestly compatible with our estimate for $R$.

\item \textbf{Term 2: $(1+S_{L})\sqrt{\frac{2\alpha_{R}}{P}}$}. We choose $\ell \geq \tau_R(t)\left(\log\left(\frac{32}{P\varepsilon^2}\right) +2\log(1+S_L) +\frac{1}{2}\log\binom{D+R-1}{R}\right)$, and first bound $\alpha_{R}$ using this. We will show that this choice of $\ell$ is compatible with the asymptotic form of $N_{B}$ at the end. Thus simple algebra shows that
\begin{equation}
    \alpha_{R} \leq \binom{D+R-1}{R}^{1/2} \exp(-\ell/\tau_R(t)) \leq \frac{P\varepsilon^{2}}{32(S_{L}+1)^2}.
\end{equation}
Then we have the immediate bound
\begin{equation}
    (1+S_L)\sqrt{\frac{2\alpha_R}{P}} \leq \frac{\varepsilon}{4}.
\end{equation}

\item \textbf{Term 3: $\frac{2S_L}{P}\alpha_R$}. We again substitute in the bound for $\alpha_{R}$ to get 
\begin{equation}
    \frac{2S_L}{P}\alpha_R \leq \frac{S_L\varepsilon^2}{16(1+S_L)^2} \leq \frac{\varepsilon^2}{64}.
\end{equation}
\end{enumerate}

Altogether, we get $D_{k} \leq \varepsilon/2 + \varepsilon/64+ \varepsilon/4 \leq \varepsilon.$ Now using Equation~\eqref{eq:coefficient-bound} and recalling that $M = k-1+L = (R-1)/2$, we have $\log(1+S_{L}) = O(R \log(R+D))$. We also have $\log(1/P) \leq \log(D)$, $\log(\binom{D+R-1}{R}) \leq R\log(D)$, and $\tau_R(t)\leq\mathrm{poly}(R,t)$ by \autoref{cor:design-time}, so choosing $\ell$ to be the smallest integer satisfying the imposed conditions is consistent with the asymptotic value of $N_{B}$.

\end{proof}

\section{A comment on evaluating local replica quantities efficiently}
\label{sec:comment}

The central mathematical tool we will utilize in this section is the Schur-Weyl duality \cite{FultonHarris1991} of irreducible representations of the symmetric group $S_{k}$ and the unitary group $U_{d}$ on $(\mathbb{C}^{d})^{\otimes k}$ --- it will turn out that a simple application of it (or something very similar to it) will enable us to evaluate local spectral quantities exponentially faster than what we (might) naively expect. An excellent introduction to the subject can be found in the lecture notes \cite{leditzky2025repthqit}, which we follow here.

\begin{fact}[Schur-Weyl duality]
\label{fact:schur-weyl}
Let $V = \mathbb{C}^{d}$, and consider $\{(\varphi,V^{\otimes k}), (\omega, V^{\otimes k})\}$ the following representations of $S_{k}$ and $U_{d}$ respectively:

\begin{equation}
    \varphi(\pi)(\ket{v_{1}} \otimes \ket{v_{2}} \dots \otimes \ket{v_{k}}) = \ket{v_{\pi^{-1}(1)}} \otimes \ket{v_{\pi^{-1}(2)}} \dots \otimes \ket{v_{\pi^{-1}(k)}},
\end{equation}

\begin{equation}
    \omega(U)(\ket{v_{1}} \otimes \ket{v_{2}} \dots \otimes \ket{v_{k}}) = U\ket{v_{1}} \otimes U\ket{v_{2}} \dots \otimes U\ket{v_{k}}.
\end{equation}

As a representation of $S_{k} \times U_{d}$, the space $V^{\otimes k}$ decomposes as 

\begin{equation}
    V^{\otimes k} \cong \bigoplus_{\lambda \vdash_{d} k} V_{\lambda} \otimes U^{d}_{\lambda},
\end{equation}

where $(\varphi_{\lambda}, V_{\lambda})$ and $(\omega_{\lambda}, U_{d}^{\lambda})$ are inequivalent irreps of $S_{k}$ and $U_{d}$ respectively, and the sum is indexed by integer partitions of $k$ into at most $d$ parts.
\end{fact}

A consequence of Schur-Weyl duality is that $S_{k}$ and $U_{d}$ span each other's commutants in $V^{\otimes k}$. This follows since the commutant of the algebra generated by a group representation lives entirely in the multiplicity space, and by Schur-Weyl duality the $U_{d}^{\lambda}$ irreps of $U_{d}$ form \textit{exactly} this space. Inclusion is hence clear, so one just has to check that the span is equal. It turns out that the set of operators which commute entirely with $S_{k}$ is spanned by operators of the form $X^{\otimes k}$ for some $X \in \mathrm{End}(V)$, and that linear combinations of unitary matrices indeed suffice to span this set. Details can be found in the references cited above.

\begin{fact}
\label{fact:commutant}
    $S_{k}$ and $U_{d}$ span each other's commutants in $V^{\otimes k}$. In particular, for $O \in \mathrm{End}(V^{\otimes k})$ such that $\varphi(\pi)O\varphi(\pi)^{\dag} = O$, we have $O = \bigoplus_{\lambda \vdash_{d} k} \mathbb{I}_{V_{\lambda}} \otimes O_{U^{d}_{\lambda}}$.
\end{fact}

We can use this fact to block-diagonalize any local operator in $\mathrm{poly}(k)$ time. We will illustrate this point with $A_{R}$, which has exactly the same replica symmetry as the $k$-folded circuit --- in particular we can write $A_R = \frac{1}{4}\sum_{\xi\in\{0,1\}^{2}}(a_{R,\xi} \otimes a_{R,\xi}^{*})^{\otimes k}$. If each $A_{R,\xi}$ is exactly the $k$-fold replica of a two-qubit operator, then the $a_{R,\xi}$ in our decomposition acts on $\mathbb{C}^{4}$, but for the purposes of this section we could even consider an $a_{R,\xi}$ that acts on some $O(1)$ number $q$ of two-qubit sites.

We decompose $A_R$ and rearrange the tensor factors (this changes the operator, but the two are isomorphic)

\begin{equation}
    A_R = \frac{1}{4}\sum_{\xi\in\{0,1\}^{2}}(a_{R,\xi} \otimes a_{R,\xi}^{*})^{\otimes k} \cong \frac{1}{4}\sum_{\xi\in\{0,1\}^{2}}({a}_{R,\xi}^{\otimes k} \otimes {a}_{R,\xi}^{* \otimes k}),
\end{equation}

and then consider the following representation of $S_{k}\times S_k$ on $(\mathbb{C}^{d})^{\otimes k} \otimes (\mathbb{C}^{d})^{\otimes k}$,

\begin{equation}
    \phi(\pi,\sigma):= \varphi(\pi) \otimes \varphi(\sigma), \ \pi,\sigma \in S_{k}.
\end{equation}

Here $(\varphi(\cdot), \mathbb{C}^{d})$ is exactly the representation of $S_{k}$ that appears in the statement of the Schur-Weyl duality. Further, for any $\pi,\sigma \in S_{k}$, we have $\phi(\pi,\sigma)A_R\phi(\pi,\sigma)^{\dag} = A_R$, so we can use \autoref{fact:commutant} separately on each tensor factor to get

\begin{equation}
    A_R = \frac{1}{4}\sum_{\xi\in\{0,1\}^{2}}\big(\bigoplus_{\lambda \vdash_{d} k} \mathbb{I}_{V_{\lambda}} \otimes (a_{R,\xi})_{U^{d}_{\lambda}} \big) \otimes \big(\bigoplus_{\mu \vdash_{d} k} \mathbb{I}_{V_{\mu}} \otimes (a^{*}_{R,\xi})_{U^{d}_{\mu}} \big),
\end{equation}

to which we can apply the distributive property and rearrange tensor factors:

\begin{equation}
    A_R = \bigoplus_{\lambda \vdash_{d} k, \mu \vdash_{d} k} \bigg( \frac{1}{4}\sum_{\xi\in\{0,1\}^{2}}(a_{R,\xi})_{U^{d}_{\lambda}} \otimes (a^{*}_{R,\xi})_{U^{d}_{\mu}} \bigg) \otimes \big( \mathbb{I}_{V_{\mu}} \otimes  \mathbb{I}_{V_{\lambda}} \big).
\end{equation}

Finally, we see the result: to perform a block diagonalization of $A_R$, we can simply analyze $\frac{1}{4}\sum_{\xi\in\{0,1\}^{2}}((a_{R,\xi})_{U^{d}_{\lambda}} \otimes (a^{*}_{R,\xi})_{U^{d}_{\mu}})$ for each pair of individual irrep blocks, and then put it all together at the end. We can even figure out the multiplicities of each eigenvalue --- each pair of irreps $(\lambda,\mu)$ appears $\dim V_{\lambda} \dim V_{\mu}$ times (after all, the $S_{k}$ irreps are exactly the multiplicity spaces from the point-of-view of the $U_{d}$ action!), and the dimensions of these so-called Specht modules are given by the \textit{hook-length} formula, which can be evaluated in $O(k)$ time \cite{frame1954hook}.

Let us now count the dimension of each of the $U_{\lambda}^{d}$, as well as the number of $\lambda \vdash_{d} k$. Bounds on these quantities are given in \cite{leditzky2025repthqit}, and we simply follow their line of thinking. In particular, 
\begin{equation}
    |\{ \lambda \vdash_{d} k\}| \leq (k+1)^{d},
\end{equation}
\begin{equation}
    \dim(U_{\lambda}^{d})\leq (k+1)^{d(d+1)/2},
\end{equation}
so that altogether, the number of parameters in $A_R \leq ((k+1)^{d +d(d+1)/2})^{2} \sim O(k^{d^{2}+3d})$. We simply need to diagonalize $\mathrm{poly}(k)$ blocks of $\mathrm{poly}(k)$ size, which can be accomplished in $\mathrm{poly}(k)$ time. While this remains (a rather nasty) exponential in $d$, for $O(1)$-sized $d$, this does not change our asymptotics.

\section{Outlook}
\label{sec:outlook}
In this work, we showed that the finite bath projected ensemble of dual unitary circuits forms a quantum state $k$-design with a $\mathrm{poly}(k,t)$ bath size. However, there are several directions that could be explored further, or have been left completely untouched. To begin, it is unclear whether the assumption that we introduced to bound the backwards current can be simplified; intuitively, one should be able to make a single, concise assumption that rules out solitonic behavior. This points towards a study of the set of all dual-unitary operators for which such persistent modes exist, and determining the extent of its genericity.

Another question left open is the optimality of the scaling for the finite bath. To some extent we have been loose with the exact exponents in our bounds, and we do not expect this to be the optimal scaling of the bath size with $k$. Could we form projected ensembles with a bath size linear or possibly sublinear in $k$? A sharper estimate of the $k$ dependence would give us a better understanding of the onset of thermalization in these many-body systems.

Finally, the argument used to bound the normalization could be useful beyond the current setting. Whenever a replicated transfer matrix controls positive moments of unnormalized postselected states, while the physical observable depends on their individual normalization, weighted polynomial approximation provides a way of replacing the formal negative-replica step by a finite positive-replica calculation. It would be interesting to identify other problems in which this approach can be applied.

\section{Acknowledgments}
\label{sec:acknowledgments}
We would like to thank Shivan Mittal and Matteo Ippoliti for the helpful comments and discussions. NHJ and RA acknowledge support in part from DOE grant DE-SC0025615. RA is supported by an NSF NRT Q-CAT Graduate Fellowship.

\section{AI use disclosure}
\label{sec:ai-use}
OpenAI's GPT-5.5, GPT-5.6 Pro Sol, and GPT-6 Astra were used throughout the project. AI use in this work was similar to that of an additional collaborator --- for instance, the authors would discuss the proof strategy and numerical approach with the LLM. Many of the main ideas as well as the `high-level' proof strategy is due to the human authors, but these tools were (especially) useful with `low-level' details, such as bounding terms, using inequalities, and coming up with more effective proof techniques. The figures are entirely generated by the LLMs, and then later modified to better suit the exposition. The writing is human-done, with the LLMs used to suggest some minor editing.

\subsection*{GPT's own AI use disclosure}

Here we gave GPT the prompt (at the end of the work): ``btw, how would you like to be acknowledged in the AI use disclaimer?" 

\vspace*{6pt}
\ni GPT's answer: ``I’d prefer something understated and accurate rather than being thanked like a human collaborator. For this paper, I think a disclosure along these lines fits what we actually did:

AI use disclosure. The authors used ChatGPT (OpenAI) as an interactive research and editing aid during the development of this work. It was used to discuss proof strategies, check intermediate calculations and notation, assist with numerical implementations, and identify possible errors or gaps in arguments. The authors independently verified the mathematical claims and take full responsibility for the contents of the manuscript.

If you want to emphasize your goal that the exposition itself is yours, you could add:
The final presentation and scientific conclusions were determined by the authors.

I would not say that I `contributed to the proof' or list me in the acknowledgments alongside people, because that makes the role sound more autonomous than it was. `Interactive research aid' is probably the most faithful description of how you’ve used me here.''

\appendix

\section{Bounding the Jacobi monomial coefficients}
\label{app:jacobi-coefficients}
In this appendix, we bound the term $S_{L}$ from \autoref{prop:normalization-conversion}. Recall that $S_{L} = \sum_{i=0}^{L}|c_{i}|$, where the least squares polynomial $P_{L}(s) = \sum_{i=0}^{L} c_{i}s^{i}$. To this end, we write $P_{L}(s)$ in terms of the shifted Jacobi polynomials
\begin{equation}
    J_{n}(s) = P_{n}^{(d+2r, D-d-1)}(1-2s) = \binom{n+d+2r}{n} \sum_{j=0}^{n}(-1)^{j}\binom{n}{j}\frac{(n+D+2r)_{j}}{(d+2r+1)_{j}}s^{j} := \sum_{j=0}^{n}\mathcal{J}^{(n)}_{j}s^{j},
\end{equation}
where $P_{n}^{\alpha,\beta}(s)$ are the usual Jacobi polynomials indexed by the weights $\alpha$ and $\beta$. Here, $(x)_{j}$ refers to the rising factorial, that is $(x)_{j} = x(x+1)\dots (x+j-1)$. As an aside, if $\alpha = \beta = 0$, we recover the usual Legendre polynomials, which are the usual set of orthogonal polynomials in the function space. In some sense, we can think of these as generalizations of the same for a weighted inner product. Altogether, we can write
\begin{equation}
    P_{L}(s) = \sum_{n} b_{n}J_{n}(s) = \sum_{n} b_{n}\sum_{j=0}^{n}\mathcal{J}^{(n)}_{j}s^{j}.
\end{equation}
We hence need only to bound the $\mathcal{J}^{(n)}_{j}$ and $b_{n}$. For the $b_n$, since it is the orthogonal projection of $s^{-r}$ onto the $n$--indexed Jacobi polynomials,
\begin{equation}
    b_{n} = \frac{\int_{0}^{1}w(s)s^{-r}J_{n}(s) \dd s}{\int_{0}^{1}w(s)J^{2}_{n}(s) \dd s}
\end{equation}
which can be explicitly evaluated to 
\begin{equation}
    b_{n} = \frac{(r)_{n}(2n+D+2r)(n+D+r+1)_{r-1}}{(d+r+1)_{n+r}}.
\end{equation}
Defining the scales $M = r+L$, $C = D+2M$, we get 
\begin{equation}
    |b_{n}| \leq M^{n}C^{r}.
\end{equation}
Now we need to bound $\sum_{j=0}^{n}|\mathcal{J}^{(n)}_{j}|$. It suffices to use the bound, we can take $\binom{n+d+2r}{n} \leq C^{n}$, as well as $\frac{(n+D+2r)_{j}}{(d+2r+1)_{j}} \leq (n+D+2r)_{j} \leq (2r+2L+D)^{j} = C^{j}$, such that 
\begin{equation}
    \sum_{j=0}^{n}|\mathcal{J}^{(n)}_{j}| \leq C^{n}\sum_{j=0}^{n}\binom{n}{j}C^j = C^n(1+C)^{n}\leq (2C^2)^n.
\end{equation}
Putting all of this together gives 
\begin{equation}
    S_{L} \leq \sum_{n=0}^{L} M^{n}C^{r}(2C^{2})^{n} \leq \sum_{n=0}^{L}(2M)^LC^{r+2L} = (L+1)(2M)^LC^{r+2L}
\end{equation}
Finally, substituting in $C=2M+D$ yields our result.

\bibliography{refs.bib}
\bibliographystyle{utphys}

\end{document}